\documentclass[11pt]{article}

\usepackage{lmodern}
\usepackage[T1]{fontenc}
\usepackage[dvipsnames]{xcolor}
\usepackage[margin=1in]{geometry}
\usepackage{amsmath,amsthm}
\usepackage{enumerate}
\usepackage{amssymb}
\usepackage{algorithm, algorithmic}
\usepackage{graphicx}
\usepackage{url}
\usepackage[pdftex]{hyperref}
\usepackage{multirow}
\usepackage{subfigure}
\usepackage{gensymb,textcomp}
\usepackage{tikz}

\newtheorem{theorem}{Theorem}

\newtheorem{lemma}[theorem]{Lemma}

\newtheorem{definition}{Definition}

\newtheorem{keyproperty}{Property}
\renewcommand{\P}{{\cal P}} 
 
\newcommand{\etal}{{et al.~}}

\newcommand{\hy}{{\hat{y}}}

\newcommand{\bA}{\mathbf{A}}

\newcommand{\bb}{\mathbf{b}}
\newcommand{\bw}{\mathbf{w}}
\newcommand{\bc}{\mathbf{c}}

\newcommand{\bh}{\mathbf{h}}

\newcommand{\bu}{\mathbf{u}}

\DeclareMathOperator{\poly}{poly}

\newenvironment{ntheorem}[1]{\indent{\bf Theorem~#1.}\it}{\par}
\newenvironment{ndefinition}[1]{\indent{\bf Definition~#1.}\it}{\par}
\newenvironment{nlemma}[1]{\indent{\bf Lemma~#1.}\it}{\par}

\renewenvironment{proof}{\noindent{\it Proof:}}{\hfill$\Box$\par}
\newcommand{\by}{\mathbf{y}}
\newcommand{\bz}{\mathbf{z}}
\renewcommand{\L}{{\mathcal L}}
\newcommand{\M}{{\mathcal M}}

\newcommand{\Q}{{\mathcal Q}}

\renewcommand{\O}{\mathcal O}

\renewcommand{\t}[1]{\tilde{#1}}
\newcommand{\h}[1]{\hat{#1}}

\newcommand{\eat}[1]{}
\newcommand{\cvol}[1]{\left\lfloor \frac{\bnorm{#1}}{2} \right \rfloor}

\newcommand{\T}{{\mathcal T}}
 \newcommand{\Od}{{\O_\delta}}
\newcounter{lp}
\newcommand{\lptag}{\tag{LP\arabic{lp}}\addtocounter{lp}{1}}

\newcommand{\vecy}{\mathbf{y}}

\newcommand{\yone}{y^{(1)}}
\newcommand{\bone}{b^{(1)}}
\newcommand{\ytwo}{y^{(2)}}
\newcommand{\vtwo}{V^{(2)}}
\newcommand{\btwo}{b^{(2)}}

\newcommand{\mzero}{\M^{(0)}}
\newcommand{\mone}{\M^{(1)}}
\newcommand{\mtwo}{\M^{(2)}}
\newcommand{\mthree}{\M^{(3)}}
\newcommand{\mthreepa}{\M^{(3)a}}
\newcommand{\mthreepb}{\M^{(3)b}}

\newcommand{\bcnorm}[1]{\|{#1}\|_{b,c} }
\newcommand{\cone}{c^{(1)}}
\newcommand{\ctwo}{c^{(2)}}

\newcommand{\eone}{{E}^{(1)}}
\newcommand{\etwo}{{E}^{(2)}}

\newcommand{\ylong}[1]{\mbox{\sc Long} \left(#1 \right)}
\newcommand{\yshort}[1]{\mbox{\sc Short}\left(#1\right)}
\newcommand{\shortcost}[1]{\mbox{\sc Short}(#1)}
\newcommand{\shiftcost}[1]{\mbox{\sc Shift}(#1)}

\newcommand{\lpb}{\mbox{\sc \ref{lpbm}}(\bb)}
\newcommand{\lpbone}{\mbox{\sc \ref{lpbm}}(\mathbf \bone)}
\newcommand{\lpbtwo}{\mbox{\sc \ref{lpbm}}(\mathbf \btwo)}

\newcommand{\Oc}{{\O}^c}
\newcommand{\Odc}{{\O}^c_{\delta}}
\newcommand{\Qc}{{\Q}^{c}}
\newcommand{\Pc}{{\P}^{c}}

\newcommand{\bnorm}[1]{\|#1\|_b}
\newcommand{\bfloor}[1]{\left\lfloor\frac{\|#1\|_b}{2}\right\rfloor}
\newcommand{\blam}{\underline{\lambda}}

\DeclareMathAlphabet{\mathscr}{OT1}{pzc}{m}{it}

\title{Near Linear Time Approximation Schemes for Uncapacitated and Capacitated \lowercase{b}--Matching Problems in Nonbipartite Graphs\thanks{A previous extended abstract of this paper appeared in SODA 2014 \cite{AhnG13}.}}
\date{}
\author{Kook Jin Ahn\thanks{Google, 1600 Amphitheatre Parkway
Mountain View, CA 94043. Email {\tt kookjin@google.com}. This work was done while the author was at University of Pennsylvania.}
\and Sudipto Guha\thanks{
Department of Computer and Information Sciences,
    University of Pennsylvania, Philadelphia, PA. Email: {\tt
      sudipto@cis.upenn.edu}. Research supported in part by NSF Award CCF-1546151. }}

\begin{document}
\maketitle
\begin{abstract}  
  We present the first near optimal approximation schemes for the
  maximum weighted (uncapacitated or capacitated) $b$--matching
  problems for non-bipartite graphs that run in time (near) linear in
  the number of edges. For any $\delta>3/\sqrt{n}$ the algorithm
  produces a $(1-\delta)$ approximation in $O(m \poly(\delta^{-1},\log
  n))$ time. We provide fractional solutions for the standard linear
  programming formulations for these problems and subsequently also
  provide (near) linear time approximation schemes
  for rounding the fractional solutions.
  Through these problems as a vehicle, we also present several ideas
  in the context of solving linear programs approximately using fast
  primal-dual algorithms.  First, even though the dual of these
  problems have exponentially many variables and an efficient exact
  computation of dual weights is infeasible, we show that we can
  efficiently compute and use a sparse approximation of the dual
  weights using a combination of (i) adding perturbation to the
  constraints of the polytope and (ii) amplification followed by
  thresholding of the dual weights.  Second, we show that
  approximation algorithms can be used to reduce the width of the
  formulation, and faster convergence.
\end{abstract}

\section{Introduction}
The $b$--matching problem is a fundamental problem with a rich history
in combinatorial optimization, see \cite[Chapters
31--33]{Schrijver03}. In this paper we focus on finding near optimal
approximation schemes for finding fractional as well as integral
solutions for maximum $b$--matching problems in non-bipartite
graphs. The algorithms produce a $(1-O(\delta))$ approximations and
run in time $O((m +n)\cdot \poly(\log n, 1/\delta))$ time for $\delta
\geq 3/\sqrt{n}$. 

\begin{definition}\cite[Chapter 31]{Schrijver03}
  In the {\bf $\mathbf b$--matching} problem we are given a weighted (possibly
  non-bipartite) graph $G=(V,E,\{w_{ij}\},\{b_i\})$ where $w_{ij}$ is the
  weight of edge $(i,j)$ and $b_i$ is the capacity of the vertex $i$.  Let $|V|=n$ and $|E|=m$.   We assume $b_i$ are integers in $[1,\poly n]$.  We can
  select an edge $(i,j)$ with multiplicity $y_{ij}$ such that
  $\sum_{j: (i,j) \in E} y_{ij} \leq b_i$ for all vertices $i$ and the goal is 
  to maximize $\sum_{(i,j)\in E} w_{ij}y_{ij}$. Let $B=\sum_i b_i$, and note $B \geq n$.
\end{definition}

\begin{definition}\cite[Chapters 32 \& 33]{Schrijver03}
 In the {\bf Capacitated $\mathbf b$--matching} problem we have an additional restriction
that the multiplicity of an edge $(i,j) \in E$ is at most $c_{ij}$ 
where $c_{ij}$
are also given in the input (also assumed to be an
integer in $[0,\poly n]$). 
Observe that we can assume $c_{ij} \leq \min\{ b_i , b_j \}$ without loss of generality.
A problem with $c_{ij}=1$ for all $(i,j)
\in E$ is also referred to as an ``unit capacity'' or ``simple'' $b$--matching
problem in the literature.
\label{def2}
\end{definition}

M\"{u}ller-Hannemann and Schwartz \cite{MH99} provide an excellent
survey of different algorithms for variants of
$b$--matching. Approaches that solve regular matching do not extend to
$b$-matchings without significant loss of efficiency. We revisit some
of the reasons shortly. In the interest of space we summarize the main
results for the $b$-matching problem briefly. Gabow~\cite{Gabow83}
gave an $O(nm \log n)$ algorithm for the unweighted ($w_{ij}=1$)
capacitated problem. For $c_{ij}= 1$ this reduces to
$O(\min\{\sqrt{B}m,nm\log n\})$.  For the weighted uncapacitated case
Anstee~\cite{Anstee87} gave an $O(n^2m)$ algorithm; an
$\tilde{O}(m^2)$ algorithm is in \cite{Gabow83}.
Letchford \etal \cite{LetchfordRT04}, building on Padberg
and Rao~\cite{PadbergR82}, gave an $O(n^2 m \log (n^2/m))$ time algorithm for the decision version of the weighted, 
uncapacitated/capacitated problem. 
In summary the best exact algorithms to date for
the $b$--matching problem in general graphs are super--linear (see 
\cite[Chapter 31]{Schrijver03}) in the size of the input.  

It is known that solving the  bipartite relaxation for the
weighted $b$--matching problem within a $(1-\delta)$ approximation
(for any $\delta>0$) will always produce a $(\frac23
-\delta)$-approximation algorithm for general non-bipartite
graphs~\cite{Furedi81,FurediKS93}. This approximation is also tight
(consider all $b_i=1,w_{ij}=1$ for a triangle graph) --- no approach
which only uses bipartite relaxations will breach the $\frac23$
barrier. 
Distributed algorithms with $O(1)$ or weaker approximation guarantees
have been discussed by Koufogiannakis and Young \cite{KY09}.
Mestre~\cite{Mestre06} provided a $(\frac23 - \delta)$ approximation
algorithm running in $O(m (\max_i b_i) \log \frac1\delta)$ time for
weighted unit capacity $b$--matching \cite{Mestre06}. However a constant
factor approximation does not seem to be a natural stopping point.

Given the recent growth in data sets and sizes of the graphs defining
instances of matching it is natural to consider approximation
algorithms that trade off the quality of the solution versus running
time. Typically these algorithms provide an $f$-approximation, that
is, for any instance we return a feasible solution whose value is at least $f$
times the value of the true optimum for that instance (maximum version).  In particular
efficient algorithms which are $(1-\delta)$-approximation schemes (for
any absolute constant $\delta>0$, independent of $n$) and faster than computing the optimum
solution are useful in this context. It would be preferable that the
running time depended polynomially on $1/\delta$ (instead of
exponential dependence) -- even though $\delta$ is assumed
constant. It is possible that each vertex has $b_i =\sqrt{n}$ and 
a linear dependence on $B$ is not a near linear time algorithm. This paper
provides the first near linear time approximation scheme for $b$--matching.

\subsection{Existing Approaches and Challenges}
We begin with the natural question about similarity and differences vis-a-vis weighted matching, which correspond to $b_i=1$ for all vertices $i$. 
Efficient approximation schemes exist for maximum weighted matching,
even for the non-bipartite case, see \cite{DuanP10,DuanPS11} and
references therein. All of these algorithms maintain a feasible
matching and repeatedly use augmentation paths -- paths between two
unmatched vertices such that the alternate edges are matched. In the
non-bipartite case, if the two endpoints are the same vertex then this
path is known as a ``blossom''. An efficient search for good
augmentation paths, in the weighted case, requires contraction of
blossoms. However this approach does not extend to non-bipartite
$b$--matching for the case $b_i >1$.  The augmentation structures
needed for $b$--matching are not just blossoms but also blossoms with
forests that are attached to the blossom (often known as petals/arms),
see the discussion in \cite{MH99}. Searching over this space of odd
cycles with attached forests is significantly more difficult and inefficient.
 In the
language of linear programming (which we discuss in more detail
shortly), augmentation paths preserve primal feasibility for the
matching problem. In our approach we
explicitly maintain a primal infeasible solution (by violating the
capacities) except at the last step.

It is known that if we copy each node $b_i$ times then the $b$--matching problems reduce to maximum weighted matching. As an example the pairs of edges $(u,v)$ and $(v,w)$ where the vertex capacities are $3,2,3$ as shown, correspond to $8$ vertices and $12$ edges.

\medskip
 
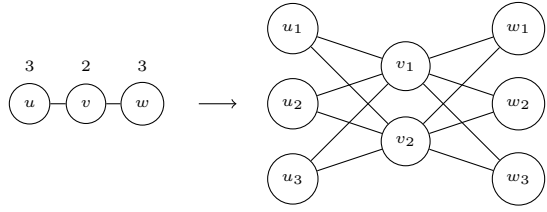
\begin{center}
 \begin{tikzpicture}[font=\tiny]
 \tikzstyle{vertex}=[draw,shape=circle,minimum size=0.3cm]
 \node[vertex] at (0,0) (u) {$u$};
 \node[above,yshift=0.3cm] at (u) {3};
 \node[vertex] at (0.75,0) (v) {$v$};
 \node[above,yshift=0.3cm] at (v) {2};
 \node[vertex] at (1.5,0) (w) {$w$};
 \node[above,yshift=0.3cm] at (w) {3};
 \draw (u) -- (v);
 \draw (v) -- (w);
 \draw[->] (2.25,0) -- (2.75,0);
 \node[vertex] at (3.5,1)  (u1) {$u_1$};
 \node[vertex] at (3.5,0)  (u2) {$u_2$};
 \node[vertex] at (3.5,-1) (u3) {$u_3$};
 \node[vertex] at (5,0.5)  (v1) {$v_1$};
 \node[vertex] at (5,-0.5) (v2) {$v_2$};
 \node[vertex] at (6.5,1)  (w1) {$w_1$};
 \node[vertex] at (6.5,0)  (w2) {$w_2$};
 \node[vertex] at (6.5,-1) (w3) {$w_3$};
 \draw (u1) -- (v1) -- (w1);
 \draw (u1) -- (v2) -- (w1);
 \draw (u2) -- (v1) -- (w2);
 \draw (u2) -- (v2) -- (w2);
 \draw (u3) -- (v1) -- (w3);
 \draw (u3) -- (v2) -- (w3);
 \end{tikzpicture}
\end{center}

The size of the graph increases significantly under such a
transformation -- consider a star graph where the central node has
$b_i=n$ and the leaf nodes have $b_i=1$ -- replication of that central
node will make the number of edges $n^2$. If we are seeking near
linear running times then transformations such as copying do not help
since the number of edges and vertices can increase by polynomial
factors. This blowup was known since \cite{Gabow83}, judicious
use of this approach has been used to achieve superlinear time (in $n$)
optimal algorithms that also depend on $B$, for example as in
\cite{GabowT91}. However near linear time algorithms have remained
elusive.

\smallskip
{\bf Linear Programming Formulations.} Consider the
following definition and linear programming formulation \ref{lpbm} for the uncapacitated $b$--matching problem.

\begin{definition}[{\sc Odd Sets and Small Odd Sets}]\label{odelta}
Given a graph $G=(V,E)$, with $|V|=n$ and $|E|=m$,
 and non-negative integer $b_i$ for each $i \in V$, for each
$U \subseteq V$ let $\bnorm{U} = \sum_{i \in U} b_i$.
Define $\O = \{ U \mid \bnorm{U} \mbox { is odd and $\geq 3$}, \mbox{$U$ has more than one vertex} \}$.
Let $\Od = \{ U \mid  U \in \O; \bnorm{U} \leq 1/\delta \}$.

{\small
\[
  \begin{array}{lll}
   \beta^* & =\lpb= \displaystyle \max \sum_{(i,j) \in E} w_{ij} y_{ij} & \\
   &\displaystyle \sum_{j:(i,j)\in E} y_{ij} \leq b_i & \forall i \in V \\
   &\displaystyle \sum_{(i,j) \in E:i,j\in U} y_{ij} \leq \left\lfloor \bnorm{U}/2 \right\rfloor & \forall U \in \O \\
   & y_{ij} \geq 0 & \forall (i,j) \in E
  \end{array}
  \lptag\label{lpbm}
\]
} 
The constraints of \ref{lpbm} represent the
``$b$-matching polytope''; any vector in this polytope can be expressed as
a convex combination of integral $b$--matching solutions, see
\cite[Chapter 31]{Schrijver03}.
\end{definition}

The constraints in \ref{lpbm} correspond to the vertices and odd
sets. The variable $y_{ij}$ (which is the same as $y_{ji}$)
corresponds to the fractional relaxation of the ``multiplicity'' of
the edge $(i,j)$ in the uncapacitated $b$--matching. It is known that the
formulation ~\ref{lpbm} has an integral optimum solution when $b_i$
are integers. The formulation has $m$ variables and $2^{\Omega(n)}$
constraints -- but can be solved in polynomial time since the oracle
for computing the maximum violated constraint can be implemented in
polynomial time using standard techniques \cite{LetchfordRT04}. That
approach is the ``minimum odd-cut'' approach of Padberg and Rao
\cite{PadbergR82}. If we only
retain the constraints for odd sets $U \in \Od$ then a fractional
solution of the modified system, when multiplied by $(1-\delta)$,
satisfies \ref{lpbm}.  That relaxed formulation, still has
$n^{1/\delta}$ constraints which is exponential in $1/\delta$. Note
that an approximate solution of the dual does not immediately provide us a solution for the
primal\footnote{In subsequent work, in manuscript \cite{access}, we
show that we can solve the dual to identify the subgraph containing
the maximum uncapacitated $b$--matching; but that manuscript uses the
results in this paper to construct an actual feasible primal solution
on that subgraph. Further the methods of \cite{access} do not apply to 
the dual of the capacitated
$b$--matching problem.}.

\smallskip
It may be tempting to postulate that applying existing multiplicative
weight algorithms such as \cite{LubyN93, PlotkinST95, GK95} and
many others (see the surveys \cite{FosterV,AroraHK12}) can help
provide us approximate solutions to \ref{lpbm} efficiently. However
that is not the case due to several reasons. First, the existing
algorithms have to maintain weights for each of the $n^{1/\delta}$
constraints.  Second, even if we are provided an approximately
feasible fractional solution, no efficient algorithm exists that
easily computes the maximum violation of the constraints in
\ref{lpbm}. Moreover it is nontrivial to verify that we have already
achieved an approximately feasible solution.  The only known
algorithms for computing the maximum violation (for just the odd-sets)
still correspond to the minimum odd-cut problem.  Those solutions are
at least cubic (see
\cite{LetchfordRT04}).

\smallskip
{\bf Capacitated $b$--Matching} The situation is more dire in presence
of edge capacities. The capacitated $b$--matching problem has two
known solution approaches. In the first one \cite[Theorem 32.2, page
564]{Schrijver03}, the matching polytope is defined by where the set
constraints are for every subset $U$ and every subset $F$ of the cut
defined by $U$.

{\small
\[
  \begin{array}{lll}
   &\displaystyle \sum_{j:(i,j)\in E} y_{ij} \leq b_i & \forall i \in V \\
   & \displaystyle y_{ij} \leq c_{ij} & \forall (i,j) \in E\\
   &\displaystyle \sum_{(i,j) \in E:i,j\in U} y_{ij} + \sum_{(i,j) \in F} y_{ij}
 \leq \left\lfloor \frac12 \left(\bnorm{U} + \sum_{(i,j) \in F} c_{ij} \right) \right\rfloor & \forall U \subseteq V, F \subseteq \{ (i,j) | i \in U, j \not \in U \} \\
& & \mbox{ and } \bnorm{U} + \sum_{(i,j) \in F} c_{ij} \mbox{is odd}  \\
   & y_{ij} \geq 0 & \forall (i,j) \in E
  \end{array}
 \lptag\label{lpbm-crazy}
\]
}

Expressing the dual of \ref{lpbm-crazy}
is already nontrivial, let alone any combinatorial manipulation. 
The second approach corresponds to compressed representations introduced in \cite{Gabow83}, see also \cite[Theorem~32.4,page 567]{Schrijver03}. It
corresponds to subdividing
each edge $e=(i,j)$ to introduce two new vertices $p_{ei}$ and
$p_{ej}$ and creating three edges, 
where $b_{p_{ei}}=c_{ij}$ as shown in the example below. There are no capacities on edges but we are constrained to {\em always}
saturate the newly created vertices $p_{ei},p_{ej}$ for every edge $(i,j)$, i.e., 
{
\small
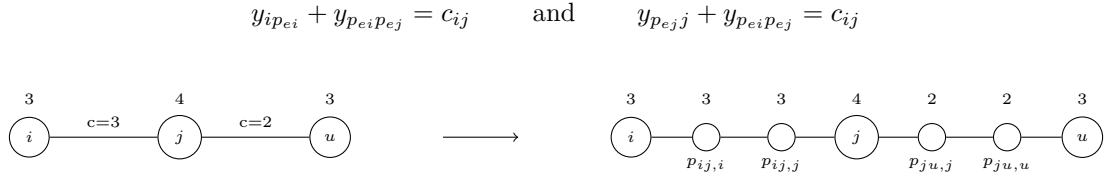
\[ y_{ip_{ei}} + y_{p_{ei}p_{ej}} = c_{ij} \qquad \mbox{ and }  \qquad y_{p_{ej}j} + y_{p_{ei}p_{ej}} = c_{ij} \]
}

\centerline{
\begin{tikzpicture}[font=\tiny]
 \tikzstyle{vertex}=[draw,shape=circle,radius=0.25cm]
 \node[vertex] at (0,0) (u) {$i$};
 \node[above,yshift=0.3cm] at (u) {3};
 \node[vertex] at (2,0) (v) {$j$};
 \node[above,yshift=0.3cm] at (v) {4};
 \node[vertex] at (4,0) (w) {$u$};
 \node[above,yshift=0.3cm] at (w) {3};
\draw (u) to node[above] {c=3} (v);
 \draw (v) to node[above] {c=2}(w);
 \draw[->] (5.5,0) to (6.5,0);
 \node[vertex] at (8,0)  (i1) {$i$};
 \node[above,yshift=0.3cm] at (i1) {3};
 \node[vertex] at (9,0) (ei1) {}; 
 \node[above,yshift=0.3cm] at (ei1) {3};
 \node[below,yshift=-0.15cm] at (ei1) {$p_{ij,i}$};
 \node[vertex] at (10,0) (ei12) {};
 \node[above,yshift=0.3cm] at (ei12) {3};
  \node[below,yshift=-0.15cm] at (ei12) {$p_{ij,j}$};
 \node[vertex] at (11,0) (i2) {$j$};
 \node[above,yshift=0.3cm] at (i2) {4};
 \node[vertex] at (12,0)  (ei2) {}; 
 \node[above,yshift=0.3cm] at (ei2) {2};
  \node[below,yshift=-0.15cm] at (ei2) {$p_{ju,j}$};
 \node[vertex] at (13,0) (ei23) {};
 \node[above,yshift=0.3cm] at (ei23) {2};
 \node[below,yshift=-0.15cm] at (ei23) {$p_{ju,u}$};
 \node[vertex] at (14,0)  (i3) {$u$};
 \node[above,yshift=0.3cm] at (i3) {3};
 \draw (i1) to node[above] {}  (ei1);
 \draw (ei1) to node[above] {}  (ei12);
 \draw (ei12) to node[above] {}  (i2);
 \draw (i2) to node[above] {}  (ei2);
 \draw (ei2) to node[above] {}  (ei23);
 \draw (ei23) to node[above] {}  (i3);
 \end{tikzpicture}
}

\noindent
Observe that the equality rules out simply scaling the vector $\by$ by
a constant smaller than $1$. The all-zero vector $\mathbf{0}$ is not
even in the polytope! Even though the polytope is convex, the lack of
closure under affine transformations makes it unwieldy for most known
techniques that produce fast approximate solutions. The transformation
creates unusual difficulties for approaches that are not based on
linear programming as well, see \cite{Hougardy}.  New ideas are
required to address these issues and the development of such is the
goal of this paper.

\subsection{Contributions}
\label{sec:contrib}
The paper combines several novel structural properties of the
$b$--matching polytope with novel modifications of the multiplicative
weights method, and uses approximation algorithms to efficiently solve
the subproblems produced by the said multiplicative weights
method. All three of these facets function in tandem, and the
overall technical theme of the solution are independently of interest.

\paragraph{\bf Main Results}
We assume that the edges in the graph $G=(V,E)$ are presented as a
read only list $\langle\ldots,(i,j,w_{ij}),\ldots\rangle$ in arbitrary
order where $w_{ij}$ is the weight of the edge $(i,j)$.  The space
complexity will be measured in words and we assume that the integers
in the input are bounded from above by $\poly n$ to avoid
bit-complexity issues. We prove the following theorems about
$b$--matching.

\begin{theorem} [{\sc Fractional $b$--matching}] \label{thmone} Given any
   non-bipartite graph, for any $\frac{3}{\sqrt{n}}<\delta\leq 1/16$, we find a
  $(1-O(\delta))$-approximate (to \ref{lpbm}) fractional weighted
  $b$-matching using additional ``work'' space (space excluding the
  read-only input) $O(n\poly (\delta^{-1},\ln n))$ and making
  $T=O(\delta^{-4} (\ln (1/\delta)) \ln n)$ passes over the list of
  edges.  The running time\footnote{The exact exponent of $\delta,\log n$ in the $\poly()$ 
term depends on \cite{lowstcuts,BhalgatHKP07} and we omit further discussion in this paper.} is $O(mT + n\poly (\delta^{-1},\ln n))$.
\end{theorem}

\begin{theorem}[{\sc Integral $b$--matching}]\label{thm:rounding}
  Given a fractional $b$-matching $\by$ for a non-bipartite graph
  which satisfies the constraints in the standard LP formulation and
  has weight $W_0$, we find an integral $b$--matching of weight at least
  $(1-2\delta)W_0$ in $O(m'\delta^{-3}\ln (1/\delta))$ time and
  $O(m'/\delta^2)$ space where $m'=|\{(i,j)|y_{ij} > 0 \}|$.
\end{theorem}

The computation for the capacitated $b$--matching problem maintains
the invariant that edge capacities are never violated at any stage of
the algorithm.  This yields a new approximation version of the
capacitated matching problem where we exceed the vertex capacities but
do not exceed the edge capacities at all and (almost) preserve the
objective function.  We prove:

\begin{theorem}[{\sc Fractional, Capacitated}] \label{thmtwo}  
Given any weighted non-bipartite graph, for any $\frac{3}{\sqrt{n}}<\delta\leq 1/16$, we
  find a $(1-O(\delta))$-approximate 
fractional capacitated $b$-matching using  $O(mR/\delta  + \min\{ B, m \}$ 
$\poly (\delta^{-1},\ln n))$ time, 
$O(\min\{m,B\} \poly (\delta^{-1},\ln n))$
  additional ``work'' space with $R=O(\delta^{-4} (\ln^2 (1/\delta)) \ln n)$ 
passes over the list of edges where $B=\sum_i b_i$.  The algorithm
  returns a solution $\{\h{y}_{ij}\}$ such that the subgraph
  $\h{E}=\{(i,j)| (i,j) \in E, \h{y}_{ij}>0\}$ satisfies $\sum_{(i,j)
    \in \h{E}} w_{ij}c_{ij} \leq 16 R \beta^{*,c}$ where $\beta^{*,c}$ is the weight of
  the integral maximum capacitated $b$--matching. 
\end{theorem}

\noindent The restriction on $\sum_{(i,j) \in \h{E}} w_{ij}c_{ij}$ is explicitly used in the next theorem.

\begin{theorem}[{\sc Integral, Capacitated}]\label{thm:crounding}
  Given a feasible fractional solution $\by$ to the linear program referred in
  Theorem~\ref{thmtwo} for a
  non-bipartite graph such that the optimum solution is at most
  $\beta^{*,c}$ and $\sum_{(i,j) \in \hat{E}} w_{ij}c_{ij} \leq 16R
  \beta^{*,c}$ where $\hat{E}=\{(i,j)|y_{ij} > 0 \}$, we find an
  integral $b$-matching of weight at least $(1-\delta)\sum_{(i,j)}
  w_{ij}y_{ij} - \delta\beta^{*,c} $ in $O(m'R\delta^{-3}\ln
  (R/\delta))$ time and $O(m'/\delta^2)$ space where $m'= |\hat{E}|$
  is the number of nontrivial edges (as defined by the linear program)
  in the fractional solution. As a consequence
we have a $(1-O(\delta))$-approximate integral solution.
\end{theorem}

\paragraph{Technical Themes} To prove the Theorems~\ref{thmone}--\ref{thm:crounding} this paper makes novel contributions towards the structure of $b$--matching polytope as well as techniques for speeding up multiplicative weights methods. 

\smallskip
{\em Multiplicative Weights Methods.} We show that we can use existing
constant factor approximation algorithms for $b$--matching to produce
a $(1-\delta)$-approximate solution. The approximation factor surfaces
in the speed of convergence of the multiplicative weights method used
but the final solution produced is a $(1-\delta)$ approximation. This
provides fairly straightforward proofs for near linear time
$(1-\delta)$ approximation schemes for bipartite graphs using standard
multiplicative weights methods. While the results for bipartite case
in this paper do not completely dominate existing results (e.g.,
\cite{AhnG11BM}), they serve as a warmup for non-bipartite
graphs. Many of the pieces which are demonstrated with relatively less
complexity in the bipartite case (initial solutions, Lagrangians,
etc.) are also re-used in the non-bipartite case.

We then use specific structural properties of the $b$--matching
polytope (and perturbations, described shortly) to show that the
non-bipartite $b$--matching problem can be solved via a sequence of
weighted bipartite $b$--matching problems. The overall approach can be
viewed as dual thresholding where we only focus on the large weights
in the multiplicative weights method (which are candidate dual
variables) and ignore the remainder. If we modify (perturb) the
$b$--matching polytope then the number of constraints with large
weights is small. However the choice of these constraints vary from
iteration to iteration -- and our algorithm differs from the
application of standard multiplicative weights techniques in this
aspect. Naturally, this requires a proof that the modified approach
converges. This is shown in Section~\ref{sec:gen} and is used to prove
Theorem~\ref{thmone} for uncapacitated $b$--matching. The framework
extends to capacities helping prove Theorem~\ref{thmtwo}.

\smallskip
{\em Polytope and Perturbations.}  We investigate the laminarity of the
sets corresponding to the unsatisfied constraints in \ref{lpbm} at the
neighborhood of any infeasible primal. A collection of sets $L$ is a
laminar family if for any two sets $U,U'\in L$, $U \cap U'$ is either
$U$, $U'$ or $\emptyset$. We show that if we modify the polytope by
introducing a small perturbation, then the constraints corresponding
to the small odd sets that are ``almost maximally violated'' define a
laminar family. Since a laminar family has $O(n)$ sets, this provides
the small subset of constraints to the modified multiplicative weights
method (note that the total number of constraints is
$\Omega(n^{1/\delta})$). In that sense this approach generalizes the minimum
odd-cut approach.

Many algorithms using the minimum odd-cut approach
rely on the following fact: the sets corresponding to the nonzero
variables of the optimum dual solution of \ref{lpbm} define a laminar
family (see Giles and Pulleyblank \cite{GilesP}, Cook~\cite{Cook},
Cunningham and Marsh
\cite{CunninghamM78}, and also Schrijver~\cite{Schrijver03}). However
all these techniques rely on the {\em exact optimality} of the pair of
primal and dual solutions. In fact, such relationships do not exist
for arbitrary candidate primal or dual solutions. It is surprising
that the maximally violated constraints of the perturbed polytope
shows this property. This is shown in Theorem~\ref{thm:main1}.

\begin{theorem}\label{thm:main1} 
For a graph $G$ with $n$ vertices and any non-negative edge weights 
$\h{\by}$ suppose that we are given $\h{\by}$ satisfying 
$\h{y}_{ii}=0$ for all $i$ and $\sum_{j:(i,j) \in E} \h{y}_{ij} \leq b_i$ for all $i$. 
Define a {\bf perturbation} of $b_i, b_U = \lfloor \bnorm{U}/2 \rfloor$ as $\t{b}_i=(1-4\delta)b_i$ and $\t{b}_U = \lfloor \bnorm{U}/2 \rfloor - \frac{\delta^2\bnorm{U}^2}{4}$. 
Let $\h{\lambda}_U = (\sum_{(i,j) \in E: i,j\in U} \h{y}_{ij})/\t{b}_U$ and 
$\h{\lambda} = \max_{U \in \Od} \h{\lambda}_U$.
If $\delta \leq \frac{1}{16}$ and $\hat{\lambda}\geq 1+3\delta$, 
the set $L_1=\{U:\hat{\lambda}_U \geq \hat{\lambda}-\delta^3; U \in \Od\}$ forms a laminar 
family. 
Moreover for any $x \geq 2$ we have $|\{ U:\hat{\lambda}_U \geq \hat{\lambda}-\delta^x; U \in \Od\}| \leq n^3 + \left(n/\delta\right)^{1+\delta^{(x-3)/2}}$.
\end{theorem}

 In other words, if we were provided an infeasible (with respect to the perturbed polytope) 
primal solution $\{\h{y}_{ij}\}$ then the constraints that are
 almost as violated as the maximum violated constraint of the
 perturbed polytope (in ratio of LHS to RHS) correspond to a laminar
 family. Intuitively,  $\sum_{(i,j):i,j \in U} \h{y}_{ij} =\h{\lambda}_U \t{b}_U $ and 
for a fixed $\h{\lambda}_U$, if we could ignore the floor and ceil functions, 
the right hand side is a concave function of
 $\bnorm{U}$. As a result if two such $U_1,U_2$ intersect at a
 non-singleton odd set $U_3 \neq U_1, U_2$ (the union $U_4 \neq
 U_1,U_2$ is also an odd set) then $\max\{ \h{\lambda}_{U_3},
\h{\lambda}_{U_4} \}$ will exceed $\min\{ \h{\lambda}_{U_1},
\h{\lambda}_{U_2} \}$ by $\delta^3$. Of course, the floor and ceil
functions, singleton sets cannot be ignored and
more details are required, and Theorem~\ref{thm:main1} is proved in Section~\ref{proof:main1}. 
However Theorem~\ref{thm:main1}, does not give us an algorithm. But the laminarity of the ``almost maximally violated'' constraints allow us to design an algorithm that finds these constraints (small odd sets) without the knowledge of the maximum violation. Since the laminarity guarantees that at most $O(n)$ such sets can be found, we can compute the maximum violated constraint more efficiently than the existing algorithms. This is formalized in Theorem~\ref{thm:main2}.

\smallskip
\begin{theorem}\label{thm:main2} 
For a graph $G$ with $n$ vertices and $\{\h{y}_{ij}\}$ 
and the definitions of $\{\h{\lambda}_{U}\}$
exactly as in the statement of Theorem~\ref{thm:main1} and $\delta\in (0,\frac1{16}]$,
if $\hat{\lambda}\geq 1+3\delta$ we can find the 
set $L_2=\{U:\hat{\lambda}_U \geq \hat{\lambda}-\frac{\delta^3}{10}; U \in \Od\}$ 
in
$O(m' + n \poly \{\delta^{-1},\log n\})$
time using $O(n \delta^{-5})$ space where $m'=|\{ (i,j) | \h{y}_{ij} > 0 \}|$.
\end{theorem}
\newcommand{\Cut}{\mathscr{Cut}}

\smallskip
The proof of Theorem~\ref{thm:main2} combines the insights of the minimum odd-cut approach
~\cite{PadbergR82} along with the fact that $L_2
\subseteq L_1 $ is a laminar family as proved in
Theorem~\ref{thm:main1}.

\paragraph{Roadmap.}
Theorems~\ref{thm:main1} and \ref{thm:main2} are proved in
Sections~\ref{proof:main1} and \ref{proof:main2} respectively. We
discuss the bipartite $b$--matching problem in Section~\ref{warmup} to
serve as a warmup as well as to develop pieces (such as initial
solutions, etc.) that would be required to solve the non-bipartite
problem. In particular we make the connection between fast constant
factor approximation algorithms and the convergence of the
multiplicative weights method. Section~\ref{sec:gen} which discusses
perturbations and thresholding and provides a modified multiplicative
weights framework which is likely of interest in other problems where
we have a large number of constraints. Theorem~\ref{thmone} follows
immediately from the application of the framework and the bipartite
relaxation discussed in Section~\ref{warmup}.
Section~\ref{sec:rounding} proves Theorem~\ref{thm:rounding}.  Section~\ref{sec:capacitated} discusses
capacitated $b$--matching.

\section{Approximations to speed up Multiplicative Weights Method}
\label{warmup}

The goal of this section is to illustrate how multiplicative weights
method can be used in the context of $b$--matching. We focus on the
bipartite case in this section.  The results obtained in this section
do not always dominate the best known results for bipartite
$b$--matching, see for example \cite{AhnG11BM}. But the main purpose
of this section is to provide a simple illustration of the ideas that
are required for the non-bipartite case. We use existing
multiplicative weights methods (see \cite{AroraHK12} for a
comprehensive review of these) and show how they apply to the
bipartite $b$--matching case without any modification.  At the end of
the section we discuss why existing techniques will {\bf not} work
directly in the non-bipartite case. However the different parts of the
overall solution for bipartite graphs will be reused in the
non-bipartite context.

From the perspective of algorithms for matching problems, the
multiplicative weights method provides an approach different from that
of augmentation paths. Instead of maintaining a feasible solution and
increasing the value of that feasible solution using augmenting paths,
we maintain an infeasible solution of a certain value and reduce the
infeasibility.  The overall algorithm is iterative, at each point we
identify parts of the graph where our solution is infeasible --- we
construct a new partial solution that reduces the effect of these parts
and consider a convex combination of the old and new
solutions. However the new partial solution, in itself can be
significantly unhelpful for the original problem! In particular the new
solution will either be a matching that allows vertex $i$ to have up
to $6b_i$ edges instead of the at most $b_i$ as specified in the
problem, or have $1/6$ the desired objective value (which depends both
on the weight of the maximum matching as well constraints in the
framework). Of course, this deviation also allows us to find the
solution efficiently.  However, even though each individual solution
is not helpful, the average of the solutions is a $(1-O(\delta))$
approximation for the original problem for a small $\delta>3/\sqrt{n}$.

\subsection{Existing Multiplicative Weights Methods} 
Let $\bA'$ be a non-negative $m \times N$ matrix, and suppose $\bb' \geq \mathbf{0}$.  
Suppose that we seek to solve $\bA'\by \leq \bb', \by \in \P'$ where $\P'
\subseteq \{\by | \by \geq \mathbf{0}\}$ is convex.  The literature on
Multiplicative Weights method shows that it suffices to repeatedly average
$\by(t)$ corresponding to iteration $t$. In iteration $t$, given a non-negative vector $\bu(t)$, 
the methods ask for an oracle to supply $\by(t)$ such that 
$\bu(t)^T\bA'\by(t) \leq (1+O(\delta))\bu(t)^T\bb', \by(t)
\in \P'$ and $\bA'\by(t) \leq \rho\bb'$ where $\rho > 1$ is the width parameter. 
The $\bu(t)$ are
referred to as the {\bf Multiplicative Weights}, because the vector
$\bu(t)$ in the expression $\bu(t)^T\bA'\by(t)$ implies an assignment
weights to the rows of $\bA'$ which correspond to constraints.  The
multiplicative weights method states that as long as we have bounded
solutions $\bA'\by(t) \leq
\rho\bb'$, a (weighted) average $\by$ of $\by(t)$ satisfies 
$\bA'\by \leq (1+O(\delta))\bb'$. 
We note that many variations of the multiplicative weights method
exist but for the purposes of this section we focus on the version in \cite{PlotkinST95}. 
In that version the average is a
predetermined weighted average and the $j$-th entry of $\bu(t)$
corresponds to a scaled exponential of $(\bA' \by')_j/\bb'_j$ where
$\by'$ is the corresponding weighted average of
$\by(0),\ldots,\by(t-1)$.  Intuitively, if the $j$-th
constraint is violated more, its weight would be large and the desired
$\by(t)$ would prioritize satisfying the $j$-th constraint.
\begin{theorem}\cite{PlotkinST95}
\label{psttheorem}
Starting from an initial solution $\by(0)$ such that $\bA'\by(0) \leq \rho \bb'$, after $O(\rho(\delta^{-2} + \log \rho) \log N)$ iterations we have a $\by \in \P'$ that satisfies $\bA'\by' \leq (1+\delta)\bb'$.
\end{theorem}

\subsection{Boosting Constant Factor Approximations to $(1-\delta)$-approximations}
We begin with Theorem~\ref{apxtheorem} and consider its applications.
\newcommand{\wi}[1]{\widehat{#1}}
\begin{theorem}[Proved in Section~\ref{verycoolproof}]
\label{apxtheorem}
Let $f_1,f_2 > 0, \bh \geq \mathbf{0}$.
Let $\wi{\Q} \subseteq \wi{\P} \subseteq \{ \by \mid \by \geq \mathbf{0}\}$. Suppose $\wi{\P},\wi{\Q}$ are convex and $\mathbf{0} \in \wi{\Q}$.
Suppose we have a subroutine that for any $\bz$ (which can be negative)  
provides a $\by \in \wi{\P}$ such that $\bz^T\by \geq (1-\delta/2) \max \{ \bz^T\by' \mid \by' \in \wi{\Q} \}$.\footnote{While it may be appealing to discuss closure of $\wi{\Q}$, note the $(1-\delta/2)$ factor and therefore $\limsup$ suffices.} 
\begin{enumerate}\parskip=0in
\item If 
$\{ \by | \bw^T \by \geq f_1, \bh^T\by \leq f_2, \by \in \wi{\Q}\}$ is non-empty then 
using $O(\ln \frac{1}{\delta})$ invocations of the subroutine we can find a $\by \in \wi{\P}$ such that $\bw^T\by \geq (1-\delta) f_1$ and $\bh^T\by \leq f_2$.
\item Suppose  $\wi{\bA},\wi{\bb}$ are 
non-negative and $\bb \in R^N$, let $\wi{\beta} = \max \{ {\bw}^T\by \mid \wi{\bA}\by \leq
\wi{\bb}, \by \in \wi{\Q}\}$. If $\{ \by/\lambda_0 \mid \by \in \wi{\P}\} \subseteq
\{\by \mid \wi{\bA}\by \leq \wi{\bb},\by \in \wi{\Q} \}$ then 
we can compute $\by$ that satisfies $\bw^T\by \geq
(1-\delta)^2 \wi{\beta}$, $\wi{\bA}\by \leq (1+\delta) \wi{\bb}$ and $\by \in
\wi{\P}$ using $O(\lambda_0 (\delta^{-2} +
\delta^{-1}\log \lambda_0)(\log N)(\log 1/\delta))$ invocations of the
subroutine.
\end{enumerate}
Note that if $\{ \lambda_0 \by \mid \wi{\bA}\by \leq \wi{\bb}, \by \in \wi{\Q} \} = \wi{\P}$ for some $\lambda_0 \geq 1$, then a $(1/\lambda_0)$-approximate solution to $\max \{ {\bw}^T\by \mid \wi{\bA}\by \leq \wi{\bb}, \by \in \wi{\Q}\}$ can be multiplied by $\lambda_0$ to achieve the subroutine mentioned above and therefore using $O((\lambda_0
 (\delta^{-2} + \log \lambda_0)\log N +\delta^{-1} \log \lambda_0) \log (1/\delta))$
invocations we find a (fractional) $\by$ as described in (2).
\end{theorem}

\paragraph{\bf Bipartite Uncapacitated $b$--matching.} The problem is expressed by linear program \ref{lp001}.
Variable $y_{ij}$ corresponds to the fraction with which $(i,j)
\in E$ is present in the solution. 
{\small
\[
  \begin{array}{lll}
   \beta^*_b = & \displaystyle \max \sum_{(i,j) \in E} w_{ij} y_{ij} & \\
   \Q :& \left\{ \begin{array}{lll}
   &\displaystyle \sum_{j:(i,j)\in E} y_{ij} \leq b_i & \forall i \in V \\
   & y_{ij} \geq 0 & \forall (i,j) \in E
  \end{array} \right.
  \end{array}
  \lptag\label{lp001}
\]
}

Observe that negative weight edges can simply be ignored by any
approximation algorithm. While many constant factor approximation
algorithms for uncapacitated $b$--matching exist, we use
Theorem~\ref{uncapapx} which has no dependence on $B=\sum_i b_i$.

\begin{theorem} 
\label{uncapapx}[Proved in Section~\ref{sec:initial}]
For the bipartite uncapacitated $b$--matching problem we can provide a $1/6$ approximation in
$O(m \log n)$ time and $O(n)$ space. 
\end{theorem}

We now define $\Q$ as in \ref{lp001} and set
$\wi{\Q} =\Q$, $\{\wi{\bA}\by \leq \wi{\bb}\} = \Q$ and $\wi{\P} = \{
6\by | \by \in \wi{Q} \}$. We multiply the solution provided by
Theorem~\ref{uncapapx} by a factor $6$ and as a consequence of
the final part of Theorem~\ref{apxtheorem} we obtain a non-negative
(fractional) solution $\{y_{ij}\}$ that satisfies $\sum_{j:(i,j)\in E}
y_{ij} \leq (1+\delta) b_i$ for all $i \in V$ corresponding to
$\wi{\bA}\by \leq (1+\delta) \wi{\bb}$. Dividing each $y_{ij}$ by
$(1+\delta)$ provides us a $((1-\delta)^2/(1+\delta))$-approximation to
the optimum bipartite $b$--matching solution in time
$O(m\delta^{-2}(\log^2 n)(\log 2/\delta))$.

\paragraph{\bf Bipartite Capacitated $b$--matching.} The problem is expressed as
a linear program in \ref{qcdef} where $c_{ij}$ are integer capacities
on the edge $(i,j) \in E$. Without loss of generality
$c_{ij} \leq \min \{ b_i, b_j\}$.

{\small
\[
  \begin{array}{lll}
   \beta^{*,c}_b=& \displaystyle \max \sum_{(i,j) \in E} w_{ij} y_{ij} & \\
   \Qc : & \left\{ \begin{array}{lll}
   &\displaystyle \sum_{j:(i,j)\in E} y_{ij} \leq b_i & \forall i \in V \\
   & y_{ij} \leq c_{ij} & \forall (i,j) \in E\\
   & y_{ij} \geq 0 & \forall (i,j) \in E
  \end{array} \right.
  \end{array}
  \lptag\label{qcdef}
\]
}

\noindent
Define $\Pc$ as:

{\small
\begin{equation}
 \Pc : \left\{ \begin{array}{lll}
   &\displaystyle \sum_{j:(i,j)\in E} y_{ij} \leq \lambda_0 b_i/2 & \forall i \in V \\
   & y_{ij} \leq c_{ij} & \forall (i,j) \in E\\
   & y_{ij} \geq 0 & \forall (i,j) \in E
  \end{array} \right. \lptag \label{qclamdef}
\end{equation}
}

\smallskip
\begin{theorem}
\label{caprecursive}[Proved in Section~\ref{sec:initial}]
If $c_{ij} \leq \min \{b_i, b_j \}$ then given any weight vector
$\bw$, using Theorem~\ref{uncapapx} at most $k=O(\log 1/\delta)$
times we can compute a solution $\by^{\dagger,c} \in \Pc$ with
$\lambda_0 =16 \ln \frac2\delta$ such that $\bw^T \by^{\dagger,c} \geq
2 \beta^{*,c}_b/\lambda_0$. If $\h{E}=\{(i,j) \mid y^{\dagger,c}_{ij} > 0 \}$ 
then $\sum_{(i,j) \in \hat{E}} w_{ij} c_{ij} \leq 8 k \beta^{*,c}_b$.
\end{theorem}

\noindent
We cannot use an arbitrary algorithm in lieu of
Theorem~\ref{caprecursive} -- because we only relax a part of the
constraints.  The final property of Theorem~\ref{caprecursive} is used
to guarantee that a fractional solution can be rounded in near linear
time (Theorem~\ref{thm:crounding}). We define $\wi{\bA}\by
\leq \wi{\bb}$ to be $\{\sum_{j:(i,j)\in E} y_{ij} \leq b_i, \forall i
\in V \}$ and let $\wi{\Q}=\Qc$ and $\wi{\P}=\Pc$ for
$\lambda_0=16\ln 2/\delta$. We apply Theorem~\ref{apxtheorem} to get a
solution which satisfies $\sum_{j:(i,j)\in E} y_{ij} \leq
(1+\delta)b_i$ for all $i$ as well as $y_{ij} \leq c_{ij}$ for all
$(i,j) \in E$. An appropriate scaling of the solution provides a
$(1-O(\delta))$-approximation.

\begin{theorem}
\label{capb}
We can compute a fractional solution which is a $(1-\delta)$
approximation to the optimum capacitated $b$--matching in
$O(m\delta^{-2}(\log^2 n)(\log^2 1/\delta))$ time in a bipartite
graph.
\end{theorem}

\subsection{Proof of Theorem~\ref{apxtheorem}}
\label{verycoolproof}

\begin{ntheorem}{\ref{apxtheorem}}
Let $f_1,f_2 > 0, \bh \geq \mathbf{0}$.
Let $\wi{\Q} \subseteq \wi{\P} \subseteq \{ \by \mid \by \geq \mathbf{0}\}$. Suppose $\wi{\P},\wi{\Q}$ are convex and $\mathbf{0} \in \wi{\Q}$.
Suppose we have a subroutine that for any $\bz$ (which can be negative)  
provides a $\by \in \wi{\P}$ such that $\bz^T\by \geq (1-\delta/2) \max \{ \bz^T\by' \mid \by' \in \wi{\Q} \}$. 
\begin{enumerate}
\item If 
$\{ \by | \bw^T \by \geq f_1, \bh^T\by \leq f_2, \by \in \wi{\Q}\}$ is non-empty then 
using $O(\ln \frac{1}{\delta})$ invocations of the subroutine we can find a $\by \in \wi{\P}$ such that $\bw^T\by \geq (1-\delta) f_1$ and $\bh^T\by \leq f_2$.

\item Suppose  $\wi{\bA},\wi{\bb}$ are 
non-negative and $\bb \in R^N$, let $\wi{\beta} = \max \{ {\bw}^T\by \mid \wi{\bA}\by \leq
\wi{\bb}, \by \in \wi{\Q}\}$. If $\{ \by/\lambda_0 \mid \by \in \wi{\P}\} \subseteq
\{\by \mid \wi{\bA}\by \leq \wi{\bb},\by \in \wi{\Q} \}$ then 
we can compute $\by$ that satisfies $\bw^T\by \geq
(1-\delta)^2 \wi{\beta}$, $\wi{\bA}\by \leq (1+\delta) \wi{\bb}$ and $\by \in
\wi{\P}$ using $O(\lambda_0 (\delta^{-2} +
\delta^{-1}\log \lambda_0)(\log N)(\log 1/\delta))$ invocations of the
subroutine.
\end{enumerate}
Note that if $\{ \lambda_0 \by \mid \wi{\bA}\by \leq \wi{\bb}, \by \in \wi{\Q} \} = \wi{\P}$ for some $\lambda_0 \geq 1$, then a $(1/\lambda_0)$-approximate solution to $\max \{ {\bw}^T\by \mid \wi{\bA}\by \leq \wi{\bb}, \by \in \wi{\Q}\}$ can be multiplied by $\lambda_0$ to achieve the subroutine mentioned above and therefore using 
$O((\lambda_0
 (\delta^{-2} + \log \lambda_0)\log N +\delta^{-1} \log \lambda_0) \log (1/\delta))$
invocations we find a (fractional) $\by$ as described in part (2).
\end{ntheorem}

\smallskip
\begin{proof}
Define $g(\varrho) = \max \{ (\bw^T - \varrho\bh^T)\by \mid \by \in \wi{\Q}  \}$.
Since $\{ \by | \bw^T \by \geq f_1, \bh^T\by \leq f_2, \by \in \wi{\Q}\}$ is non-empty, $g(\varrho)$ exists and is at least $f_1 - \varrho f_2$. 
Let $\L(\by,\varrho)=(\bw^T-\varrho\bh^T)\by$ and let $\by^{\varrho}$
 be the solution returned by the subroutine for $\bz=\bw^T - \varrho\bh^T$.

\smallskip
For $\varrho=0$, the returned 
solution $\by^{0}$ satisfies $\L(\by^{0},0)=(\bw^T - \varrho\bh^T) \by^{0} \geq (1-\delta/2) g(0) = (1-\delta/2) (f_1 - \varrho f_2)$. This implies $\bw^T\by^{0} \geq (1-\delta/2) f_1$.
	If $\by^{0}$ also satisfies $\bh^T\by^{0}\leq f_2$, then $\by^{0}$
is our desired solution for the first part of the theorem. We therefore consider the case $\bh^T\by^{0} > f_2$.

\smallskip
Consider $\varrho=f_1/f_2$ and set $\by^\varrho = \mathbf{0}$. Note we do not run the subroutine. Note $\by^\varrho \in \wi{\P}$
and $\bh^T \by^{\varrho} =0 \leq f_2$ and $\L(\by^{\varrho},\varrho)=(\bw^T - \varrho\bh^T)\by^\varrho = 0 \geq (1-\delta/2)
(f_1 - \frac{f_1}{f_2} f_2) =  (1-\delta/2) (f_1 - \varrho f_2)$
Therefore over the endpoints of the interval $\varrho \in [0,f_1/f_2] = [\varrho^-, \varrho^+]$ we have two solutions $\by^{\varrho^-},\by^{\varrho^+}$ that satisfy
\begin{enumerate}[(1)]
\item $\L(\by^{\varrho^-},\varrho^-) \geq (1-\delta/2)(f_1 - \varrho^-f_2)$, $\bh^T \by^{\varrho^-} > f_2$ 
\item $\L(\by^{\varrho^+},\varrho^+) \geq (1-\delta/2)(f_1 - \varrho^+f_2)$, $\bh^T \by^{\varrho^+} \leq f_2$
\end{enumerate}

\smallskip
Now consider running the subroutine for $\varrho=\frac12(\varrho^- + \varrho^+)$. Again 
based on the subroutine we know that we will obtain a solution $\by^{\varrho}$ which satisfies:
{\small
\[ \L(\by^{\varrho},\varrho) = (\bw^T - \varrho\bh^T) \by^{\varrho} \geq (1-\delta/2)g(\varrho) \geq (1-\delta/2) (f_1 - \varrho f_2) \]
}
If $\bh^T \by^{\varrho} > f_2$ then we focus on $[\varrho, \varrho^+]$. Otherwise we focus on 
$[\varrho^-, \varrho]$. Observe that we are maintaining the invariants (1) and (2).
Now we use binary search to find $\varrho^+, \varrho^-$ such that
 $0\leq\varrho^+-\varrho^-\leq \frac{\delta f_1}{2f_2}$. This requires $O(\ln \frac{2}{\delta})$ invocations of the subroutine.
 We take a linear combination $\by=a\by^{\varrho^+}+(1-a)\by^{\varrho^-}, a\in[0,1]$
 such that $\bh^T \by = f_2$. Since $\by^{\varrho+},\by^{\varrho^-}\in \wi{\P}$,
 their linear combination $\by$ is also in $\wi{\P}$. Note that
 \begin{align*}
  & a\L(\by^{\varrho^+},\varrho^+)+(1-a)\L(\by^{\varrho^-},\varrho^-) \geq (1-\delta/2) f_1 -  (1-\delta/2)\varrho^- f_2
    - a(1-\delta/2)(\varrho^+-\varrho^-)f_2\\ 
  & \geq (1-\delta) f_1 - \varrho^-f_2
 \end{align*}
\noindent because $a\leq 1$, $\varrho^+-\varrho^-\leq \frac{\delta f_1}{2f_2}$ and $f_2\geq 0$. Thus
{\small
 \begin{align*}
  \bw^T \by & = a\L(\by^{\varrho^+},\varrho^+) + (1-a)\L(\by^{\varrho^-},\varrho^-) 
     + a\varrho^+\bh^T\by^{\varrho^+} + (1-a)\varrho^-\bh^T\by^{\varrho^-} \\
   & \geq (1-\delta) f_1 - \varrho^-f_2 +
     a (\varrho^+-\varrho^-) \bh^T\by^{\varrho^+} 
     + \bh^T (a\varrho^-\by^{\varrho^+}+(1-a)\varrho^-\by^{\varrho^-}) \\
& \geq (1-\delta) f_1 - \varrho^-f_2 
     + \bh^T (a\varrho^-\by^{\varrho^+}+(1-a)\varrho^-\by^{\varrho^-}) \quad \mbox{(Using $\bh^T\by^{\varrho^+} \geq 0$ and $\varrho^+-\varrho^-\geq 0$)} \\
& \geq (1-\delta) f_1 - \varrho^-f_2 +
     \bh^T \varrho^-\by \quad \mbox{(Using $\by = a\by^{\varrho^+}+(1-a)\by^{\varrho^-}$)}\\
   & \geq (1-\delta) f_1  - \varrho^-(f_2-\bh^T\by) =  (1-\delta) f_1 \quad \mbox{(Since $\bh^T\by=f_2$ by construction)}
 \end{align*}
}
 The first part of the theorem follows. Note that $\bh^T\by \leq f_2$ from the case $\bh^T\by^{0} \leq f_2$.
 
 For the second part, observe that setting $\bz=\bw$ we get a $\by(0)
 \in \wi{\P}$ such that $\bz^T\by(0) \geq (1-\delta) \wi{\beta}$ using
 the subroutine. Moreover $\bz^T\by(0) \leq \lambda_0 \wi{\beta}$
 since $\by(0)/\lambda_0 \in \{\by|\wi{\bA}\by \leq \wi{\bb},\by \in
 \wi{\Q} \}$. This provides an initial solution. Observe that the
 width is $\lambda_0$ by construction. We can now apply Theorem~
 \ref{psttheorem}.  If $0 < \beta \leq \wi{\beta}$ we get a solution
 for $\by \in \wi{\P}$ that satisfies $\bw^T\by \geq (1-\delta)\beta$
 and $\bu(t)^T\wi{\bA}\by \leq \bu(t)^T\wi{\bb}$ from the first part
 of the theorem setting $f_1 =\beta$, $f_2=\bu(t)^T\wi{\bb}$. If we
 fail to find a solution to the first part for some $\beta$
then we decrease $\beta$ by a factor of $(1-\delta)$. Observe
 that we would decrease $\beta$ at most $O(\delta^{-1} \log
 \lambda_0)$ times and eventually we would reach $(1-\delta)\wi{\beta}
 \leq \beta \leq \wi{\beta}$ since the initial $\beta=\bw^T\by(0)$ is
 at most $\lambda_0\beta$. Note $\bw^T\by \geq (1-\delta)\beta \geq
 (1-\delta^2)\wi{\beta}$.  

 {\em Observe that the iterations for larger $\beta$ remain valid for
 a smaller $\beta$.} Therefore if we classify the iterations according
 to (a) decrease of $\beta$ because we did not find a solution for
 for the first part and (b) invocations where we succeed in finding a solution
 for the first part. The number corresponding to (a) is at most
 $O(\delta^{-1} \log \lambda_0)$ (decreases of $\beta$) times $O((log
 1/\delta))$, the multiplier due to the reduction. The number
 corresponding to (b) cannot be more than $O(\lambda_0 (\delta^{-2} +
 \log \lambda_0)\log N)(\log (1/\delta))$ because then we would have
 already gotten a better solution based on Theorem~\ref{psttheorem} --
 once again, because the $\by$ found for larger $\beta$ remain valid
 for a smaller $\beta$. The total number of invocations of the subroutine is $O((\lambda_0
 (\delta^{-2} + \log \lambda_0)\log N +\delta^{-1} \log \lambda_0) \log (1/\delta))$. 
 The second part of the
 theorem follows.

For the final remark, a $\lambda_0$ approximation implies that we
have a feasible solution solution $\by'$ satisfying $\wi{\bA}\by' \leq \wi{\bb}, \by'
\in \Q$.
The claim follows from the second part.

\end{proof}

\subsection{Proofs of Theorem~\ref{uncapapx} and \ref{caprecursive}}
\label{sec:initial}

In this section we provide primal-dual approximation algorithms for
both uncapacitated and capacitated $b$--matching. The capacities
$b_i,c_{ij}$, for vertices and edges respectively are integral. Each
edge $(i,j)$ has weight $w_{ij}$.  In the uncapacitated case the edge
constraints are not present; one can model that by setting
$c_{ij}=\min \{b_i, b_j\}$ for every edge $(i,j)$.  The formulation
\ref{qcdef} expresses a {\em bipartite relaxation} which omits
non-bipartite constraints. Therefore $\beta^{*,c}_b \geq \beta^{*,c}$
(the maximum capacitated $b$-matching) as well as $\beta^{*,c}_b \geq
\beta^*$ (the maximum uncapacitated $b$--matching, assuming
$c_{ij}=\min \{b_i, b_j\}$ for every edge $(i,j)$). The system
\ref{housedual} is the dual of
\ref{qcdef}. 

{\small
\[\left.
\begin{minipage}{0.45\textwidth}
{\small
\begin{align*}
& \beta^{*,c}_b= \max \sum_{(i,j) \in E} w_{ij} y_{ij} \\
& \frac{1}{b_i} \sum_{j:(i,j) \in E} y_{ij} \leq 1 \qquad \forall i \tag{\ref{qcdef}} \\
& \frac1{c_{ij}} y_{ij} \leq 1 \qquad \forall (i,j) \in E\\
& y_{ij} \geq 0 \qquad\forall (i,j) \in E
\end{align*}
}
\end{minipage} \right| \quad 
\begin{minipage}{0.45\textwidth}
{\small
\begin{align*}
& \beta^{*,c}_b =\min \sum_i p_i + \sum_{(i,j) \in E} q_{ij} \\
&  \frac{p_i}{b_i} + \frac{p_j}{b_j} + \frac{q_{ij}}{c_{ij}} \geq w_{ij} \quad \forall (i,j) \in E \lptag \label{housedual}\\
&  p_i,q_{ij} \geq 0 \quad \forall i,\forall (i,j) \in E
\end{align*}
}
\end{minipage}
\]
}

\begin{algorithm}[t]
{\small
\begin{algorithmic}[1]
\STATE We start with all $p_i=0$. Initially the graph is empty and all $y_{ij}=q_{ij}=0$. In the following $y_{ij}=y_{ji}$, the variables are defined on the edges.
\STATE Order the edges $E$ according to an arbitrary ordering and consider the edges one by one.
\STATE {\bf for} {each new edge $e=(i,j)$} {\bf do} 

\begin{enumerate}[(a)]
\item If $\frac{p_i}{b_i} + \frac{p_j}{b_j} \geq w_{ij}$ then ignore the edge, otherwise:

\item We will be eventually inserting $c_{ij}$ copies of the edge $(i,j)$. Recall for the uncapacitated case $c_{ij}=\min \{b_i,b_j\}$.

\item Suppose that $c_{ij} + \sum_{j} y_{ij} > b_i$. In that case we need to delete 
$ (\sum_{j'} y_{ij'} - b_i + c_{ij})$ edges such that when we add the
$c_{ij}$ copies of $(i,j)$ the vertex constraint $\sum_{j'} y_{ij'}
\leq b_i$ will be satisfied. Therefore we delete $x_i = \max\{0, \sum_{j'}
y_{ij'} - b_i + c_{ij}\}$ edges incident to $i$ --- {\em but we delete
the edges with the lowest $w_{ij'}$ with $y_{ij'}>0$}.

\item Likewise we delete the $x_j =  \max\{0, \sum_{i'} y_{i'j} - b_j + c_{ij}\}$ edges 
incident to $j$, with the lowest $w_{i'j}$ amongst $y_{i'j}>0$.

\item Set $y_{ij}=c_{ij}$, (if required) increase $p_i,p_j$ to be at least
$2\sum_j w_{ij} y_{ij},2\sum_i w_{ij} y_{ij}$ respectively. Set $q_{ij}=w_{ij}c_{ij}$.
\end{enumerate}
\STATE Output $\{(i,j)|y_{ij}>0\}$ and $\{p_i\}, \{q_{ij}\}$.
\end{algorithmic}
}
\caption{A near linear time algorithm for capacitated $b$--matching}
\label{alg:uncapapprox}
\end{algorithm}

\noindent Algorithm~\ref{alg:uncapapprox} satisfies the following invariants; and the next lemma is the core of the proof.
\begin{enumerate}[({\cal I}1)]
\item We maintain a feasible primal solution $\{y_{ij}\}$.
\item If we insert an edge into the solution $y_{ij}=c_{ij}$ (but some copies of this edge can be deleted later). 
\item Once an edge is 
processed (ignored or inserted) we ensure that $\frac{p_i}{b_i}
+\frac{p_j}{b_j} +\frac{q_{ij}}{c_{ij}}\geq w_{ij}$.
\item We ensure that $\{p_i,q_{ij}\}$ are non-decreasing and therefore the final $\{p_i,q_{ij}\}$ satisfies the constraints of \ref{housedual}, and $\sum_i p_i +\sum_{(i,j)} q_{ij} \geq \beta^{*,c}_b$. 
\item At the end of step 3(e), we have
the invariant $p_i \geq 2\sum_j w_{ij} y_{ij}$.
\end{enumerate}

\begin{lemma}
\label{accountlemma}
Let $\Delta$ be the decrease in $\sum_{(i,j)} w_{ij} y_{ij}$ in Steps 3(c) and 3(d) due 
to the deletions before the edge $(i,j)$ is added in Step 3(e). $\Delta \leq w_{ij} c_{ij}/2$.
\end{lemma}

\smallskip
\begin{proof}
Suppose we deleted edges at $i$ for Step 3(c) and $x_i>0$. Note that
we retained the heaviest $b_i - c_{ij}$ edges and therefore the total
retained edges have weight at least $\frac{b_i - c_{ij}}{\sum_{j'}
y_{ij'}} \sum_{j'} w_{ij'}y_{ij'}$ which is at least $\frac{b_i -
c_{ij}}{b_i} \sum_{j'} w_{ij'} y_{ij'}$ since $\sum_{j'} y_{ij'} \leq
b_i$ because $\{y_{ij'}\}$ are feasible. Thus the total weight deleted
at $i$ is at most $\frac{c_{ij}}{b_i} \sum_{j'} w_{ij'} y_{ij'}$. But
since $2\sum_{j'} w_{ij'} y_{ij'} \leq p_i$ at Step 3(e) in 
the iteration before $(i,j)$ was considered, the total
weight deleted at $i$ is at most $\frac{c_{ij}p_i}{2b_i}$. Using the
same reasoning at $j$, the the total weight deleted by $(i,j)$ at both $i,j$ is at most
$\frac{c_{ij}p_i}{2b_i} + \frac{c_{ij}p_j}{2b_j}$ which is at most
$c_{ij}w_{ij}/2$ since we are past Step 3(a). 
\end{proof}

\smallskip
Note that we now immediately have a factor $1/10$ approximation for both
capacitated and uncapacitated $b$-matching. This because the net
direct increase to $\sum_{i} p_i + \sum_{(i,j)} q_{ij}$ due to
inserting $(i,j)$ is at most $5w_{ij}c_{ij}$.  At each of the
endpoints $i,j$ the increase is $2w_{ij}c_{ij}$ and $q_{ij} \leq
c_{ij}w_{ij}$. Combined with Lemma~
\ref{accountlemma} we have a $1/10$ approximation because the total increase in 
 $\sum_{i} p_i + \sum_{(i,j)} q_{ij}$ due to $(i,j)$ is the direct
 increase from $(i,j)$ plus the increase due to all edges deleted by
 $(i,j)$ (and the edges which have been recursively deleted). But
 using Lemma~\ref{accountlemma} the total weight of all such
 recursively deleted edges is at most $w_{ij}c_{ij}$. Therefore
$10 \sum_{(i,j)} w_{ij}y_{ij} \geq \sum_{i} p_i +
\sum_{(i,j)} q_{ij} \geq \beta^{*,c}_b$.
For the remainder of the paper any absolute constant approximation
suffices. However since the approximation factor relates to the speed
of convergence, we provide a slightly better analysis, and space complexity.

\smallskip
\begin{ntheorem}{\ref{uncapapx}}
For the bipartite uncapacitated $b$--matching problem we can provide a $1/6$ approximation in
$O(m \log n)$ time and $O(n)$ space.
\end{ntheorem}

\smallskip
\begin{proof}
We first observe that $q_{ij}=0$ for every edge $(i,j)$ which already
improves the approximation to $1/8$. We then prove $6\sum_{(i,j)}
w_{ij}y_{ij} \geq \sum_{i} p_i + \sum_{(i,j)} q_{ij} \geq
\beta^{*,c}_b$. In the uncapacitated this case $c_{ij}=\min \{b_i,b_j\}$ and the edge $(i,j)$
is inserted with $y_{ij}=c_{ij}$. Therefore both $p_i, p_j \geq 2
c_{ij} w_{ij}$ due to Step 3(e). Therefore at least one of
$p_i/b_i,p_j/b_j$ is $2w_{ij}$. Thus $\frac{p_i}{b_i} +
\frac{p_j}{b_j} \geq w_{ij}$ which implies $q_{ij}=0$. This also means that at 
each insertion at least one vertex has exactly one edge (but possibly
multiple copies of it) and therefore the total number of edges in the solution is $O(n)$.
We now observe that the increase in Step~3(e) of $\sum_i p_i$ is at most 
$4 w_{ij}c_{ij} - 2\Delta$ (recall $\Delta$ is defined in Lemma~\ref{accountlemma}). Suppose that we maintained $6\sum_{(i,j)}
w_{ij}y_{ij} \geq \sum_{i} p_i $ before we considered the deletions in 
Steps 3(c) and 3(d). Then the left hand side increased by $6w_{ij}c_{ij} - 6 \Delta$ but
$$6w_{ij}c_{ij} - 6 \Delta  = (4 w_{ij}c_{ij} - 2\Delta) + (2 w_{ij}c_{ij} - 4\Delta)$$
and $\Delta \leq w_{ij}c_{ij}/2$. This implies that the increase in $6\sum_{(i,j)}
w_{ij}y_{ij}$ after Step 3(e) is more than the increase in $\sum_{i} p_i $. Therefore the invariant continues to hold and the theorem follows.
\end{proof}

\smallskip
\begin{theorem} 
\label{capapx}
We can solve the capacitated $b$--matching problem to an 
approximation factor $1/8$ in time $O(m \log n)$. If $E'$ is the set of edges $(i,j)$ 
such that $y_{ij} >0$ at any point of time in the algorithm then $\sum_{(i,j) \in E'} 
w_{ij} c_{ij} \leq 8 \beta^{*,c}_b$. 
\end{theorem}
\begin{proof}
Unlike the proof of Theorem~\ref{uncapapx} we cannot assert $q_{ij}=0$. But observe that if we maintained $ 8\sum_{(i,j)}
w_{ij}y_{ij} \geq \sum_{i} p_i +\sum_{(i,j)} q_{ij}$, then the increase to the left 
hand side is $8w_{ij}c_{ij} - 8 \Delta$  (again following the definition of $\Delta$ from Lemma~\ref{accountlemma}) and the increase to the right hand side is $4 w_{ij}c_{ij} - 2\Delta + w_{ij}c_{ij}$ (the addition is due to $q_{ij}$). But
$$ 8w_{ij}c_{ij} - 8 \Delta = 4 w_{ij}c_{ij} - 2\Delta + w_{ij}c_{ij} + 3\left(w_{ij}c_{ij} - 2 \Delta \right) \geq 4 w_{ij}c_{ij} - 2\Delta + w_{ij}c_{ij} $$
Therefore the invariant continues to hold after Step 3(e). For the second part, observe that 
$ \sum_{(i,j) \in E'} w_{ij} c_{ij} = \sum_{(i,j)} q_{ij}$ but $\sum_{(i,j)} q_{ij} \leq 8 \sum_{(i,j)} w_{ij}y_{ij}$ and $\{y_{ij}\}$ are feasible. Therefore the theorem follows. 
\end{proof}

\smallskip
\noindent We use Theorem~\ref{capapx} to prove Theorem~\ref{caprecursive}.

\smallskip
\begin{ntheorem}{\ref{caprecursive}}
Using Algorithm~\ref{alg:uncapapprox} at most $k \leq 8 \ln \frac2\delta$ times we get an integral solution that satisfies 
{\small
\begin{align*}
& \sum_{(i,j) \in E} w_{ij} y_{ij} \geq \left(1-\frac{\delta}{2}\right) \beta^{*,c}_b \\
& \sum_{j:(i,j) \in E} y_{ij} \leq \left (8 \ln \frac2\delta \right) b_i \qquad \forall i \lptag 
\label{houseprimal2} \\
& y_{ij} \leq c_{ij} \qquad \forall (i,j) \in E\\
& y_{ij} \geq 0 \qquad \forall (i,j) \in E
\end{align*}
}
Moreover if $\hat{E}=\{(i,j) \in E | y_{ij} >0\}$ then 
$\sum_{(i,j) \in \hat{E}} w_{ij} c_{ij} \leq \left(8 k \right) \beta^{*,c}_b$.
\end{ntheorem}

\smallskip
\begin{proof}
We reuse the notation $\by(t)$ since we would be using an iterative algorithm.
For any $b_i,c_{ij} \geq 0$ we can find a solution $\by(1)$ such that
$\sum_{(i,j)} w_{ij}y_{ij}(1) = \tau_1 \geq \beta^{*,c}_b/8$ using Algorithm~\ref{alg:uncapapprox} and Theorem~\ref{capapx}.

Define $\h{\beta}(1)=\beta^{*,c}_b$. We now run an iterative procedure where we 
remove the edges $(i,j)$ corresponding to  $y_{ij}(1)>0$ and decrease the corresponding capacities. The decrease in capacities corresponds to modifying \ref{qcdef} by adding the constraint $y_{ij} \leq \max \{0, c_{ij} - y_{ij}(1) \}$.
Let optimum solution of \ref{qcdef} on this modified graph be  denoted by $\h{\beta}(2)$.
We have
\begin{equation}
\label{key00}
 \h{\beta}(1) - \tau_1 \leq \h{\beta}(2) \leq \h{\beta}(1) 
\end{equation}
Consider the optimum solution of \ref{qcdef} on the unmodified graph. Let that solution be $\{y^*_{ij}\}$.
Consider $y'_{ij}=\max \{ y^*_{ij} - y_{ij}(1),0\}$. Then $\{y'_{ij}\}$
is a feasible solution of the modified \ref{qcdef} and $\sum_{i,j} w_{ij} y_{ij} \geq \h{\beta}(1) - \tau_1$. $\h{\beta}(2) \leq \h{\beta}(1)$ follows from the fact that capacities are decreased and Equation~\eqref{key00} follows.

Now we obtain a
solution $\by{(2)}$ such that $\sum_{(i,j)} w_{ij}y_{ij}{(2)} = \tau_2
\geq \h{\beta}(2)/8$. We now repeat the process by modifying
\ref{qcdef} to $y_{ij} \leq \max \{ 0, c_{ij} - y_{ij}(1) -
y_{ij}(2) \}$. Proceeding in this fashion we obtain solutions 
$\left\{y_{ij}{(\ell)} \right\}_{\ell=1}^{k}$ where $k \leq \lceil 8 \ln \frac2\delta \rceil$ or we have no further edges to pick.
Observe, that by construction $\sum_{\ell=1}^{k} y_{ij}{(\ell)} \leq c_{ij}$ for all $(i,j)$ and therefore the union of these $k$ solutions satisfies 
$y_{ij} \leq c_{ij}$. Moreover  for every $\ell$ we have $\sum_j y_{ij}(\ell) \leq b_i$ and therefore for the union of these $k$ solutions the vertex constraints hold as described in the statement of Theorem~\ref{caprecursive}.

\smallskip
We now claim that $\sum_{\ell=1}^{k} \sum_{(i,j)} w_{ij}y_{ij}(\ell) \geq \left(1- \left(\frac78\right)^k\right) \h{\beta}(1)$ by induction on $k$. The base case follows from 
$\tau_1 \geq \h{\beta}(1)/8$. In the inductive case, applying the hypothesis 
on $2,\ldots,k$ we get 
$ \sum_{\ell=2}^{k} \sum_{(i,j)} w_{ij}y_{ij}(\ell) \geq \left(1- \left(\frac78\right)^{k-1}\right) \h{\beta}(2) $. Thus:
{\small
\[
\sum_{\ell=1}^{k} \sum_{(i,j)} w_{ij}y_{ij}(\ell) \geq \tau_1 +  \left(1- \left(\frac78\right)^{k-1}\right)  \h{\beta}(2)
\geq \tau_1 +  \left(1- \left(\frac78\right)^{k-1}\right) (\h{\beta}(1) - \tau_1)
= \h{\beta}(1) - \h{\beta}(1) \left(\frac78\right)^{k-1} + \tau_1  \left(\frac78\right)^{k-1}
\]
}
\noindent and the claim follows since $\tau_1  \geq \h{\beta}(1)/8$. 
The first part of the theorem follows.
For the second part, Theorem~\ref{capapx} was applied $k$ times and the result follows.
\end{proof}

\section{Perturbations, Thresholding, and Non-bipartite $b$--matching}
\label{sec:gen}
We considered the bipartite case in Section~\ref{warmup}. We provided
an algorithm that produces an $(1-\delta)^2$-approximate solution for
$\max \{\bw^T\by \mid \wi{\bA} \by \leq \wi{\bb}, \by \in \wi{\Q} \}$,
by repeatedly, for any $\bz$ finding a solution $\by \in \wi{\P}$ such
that $\bz^T\by \geq (1-\delta/2) \max \{ \bz^T\by' \mid \by' \in
\wi{\Q}\}$ (we omit the connections between $\wi{\Q},\wi{\P}$ for the moment). 
However that algorithm relied on the Theorem~\ref{psttheorem} which
computes a multiplicative weight for each constraint/row of matrix
$\wi{\bA}$. For the bipartite case, the number of constraints was $n$
for the uncapacitated case and $n+m$ for the capacitated case for a
graph with $n$ vertices and $m$ edges.
In this section we consider non-bipartite matching --- the number of
constraints are exponential. The number of constraints can be reduced
to $n^{\Omega(1/\delta)}$ to seek a $(1-\delta)$ approximation, but
computing the multiplicative weights for all rows of the constraint
matrix is infeasible for a near linear time algorithm. We now provide
a framework that bypasses the computation of the weights for all rows
in Section~\ref{sec:newframe}. We then apply the framework 
to uncapacitated $b$--matching in Section~\ref{sec:app}.

\subsection{A Dual Thresholding Framework}
\label{sec:newframe}
Suppose that $\Q \subseteq \P \subseteq \{ \by \mid \by \geq
\mathbf{0}\}$. Suppose further that $\P,\Q$ are convex and $\mathbf{0}
\in \wi{\Q}$. The overall goal in this section to solve $\max \{\bw^T\by \mid \bA\by \leq \bb, \by \in \Q \}$, by repeatedly, for any $\bz$ finding a solution $\by \in \P$ such
that $\bz^T\by \geq (1-\delta/2) \max \{ \bz^T\by' \mid \by' \in
\Q\}$.  However, we would like to achieve the reduction by only evaluating the 
multiplicative weights for the constraints $\L$ which are close to the maximum violated 
constraint. Note that this set $\L$ would change every iteration. We achieve this 
by perturbing the constraints and focusing on $\bA\by \leq \t{\bb}$.

We present the basic Algorithm~\ref{alg:thresh}. The proof of
convergence is provided in Theorem~\ref{basetheorem}. The theorem
follows from Lemma~\ref{lem:decrease} which computes the rate of
monotonic decrease of a potential function
(Definition~\ref{def:001}). Lemma~\ref{lem:connect} demonstrates how
the ideas in Section~\ref{warmup} are used as critical pieces of
Algorithm~\ref{alg:thresh}.

\begin{algorithm}[t]
{\small
 \begin{algorithmic}[1]\parskip=0in
\STATE Let $\P,\Q$ be convex with $\Q \subseteq \P \subseteq \{ \by \geq \mathbf{0} \}$. Let $\bA$ is nonnegative matrix of dimension $M \times m$. 
\label{line1u}
\STATE \label{line2u} Fix $\delta \in (0,\frac1{16}]$.  Let $\lambda_0,K,f(\delta),\alpha$ be parameters. $\lambda_0 \geq 1$, $f(\delta) < \delta$, $\alpha \leq \frac{1}{f(\delta)} \ln \left(\frac{M\lambda_0}{\delta}\right)$. 

 \STATE Find an initial 
solution $\by_0 \in \P$ with $\bw^T\by = \beta_0$ and $\bA \by_0 \leq \lambda_0 \t{\bb}$. Set $\beta=\beta_0$. \label{initu}

\STATE Let $\epsilon=\frac18$ (note $\epsilon \geq \delta$) and $t=0$.
\STATE Start a {\bf superphase} corresponding to $\epsilon=\frac18$.
The algorithm proceeds in {\bf superphases} corresponding to a fixed value of $\epsilon$. We will be decreasing $\epsilon$.
   The algorithm ends when $\lambda \leq 1 + 8 \delta$. We will {\bf not} assume $\lambda$ decreases monotonically.
\WHILE{{\bf true}}
 \STATE Define {\footnotesize$\displaystyle \lambda_\ell = (\bA\by)_{\ell}/\t{b}_\ell$} and $\lambda=\max_{\ell} \lambda_{\ell}$.
Find $\L = \{\ell | \lambda_\ell \geq \lambda - f(\delta) \}$, assert $|\L| \leq K$. 
 \label{exactcompu}
\STATE If $(\lambda \leq 1+8\delta)$ output $\by$ which satisfies $\bw^T \by \geq (1-\delta)\beta$ and $\bA\by \leq (1+8\delta)\t{\bb}$ and stop. \label{exitruleu}
\STATE If $\lambda<1+8\epsilon$ then declare the current {\bf superphase} to be over.
\STATE Repeatedly set $\epsilon \leftarrow \max \{2\epsilon/3,\delta\}$ till $\lambda>1+8\epsilon$ and start a {\bf new superphase} corresponding to this new $\epsilon$.
\STATE Define $\bu(\L)$ as $u(\L)_{\ell}=\exp(\alpha \lambda_{\ell} )/\t{b}_{\ell}$ if  $\ell \in \L$ and $0$ otherwise. Let $\gamma=\bu(\L)^T\t{\bb}$.
\label{defgammau}

\STATE \label{effectivestepu} 
Using $O(\ln \frac2\delta)$ invocations of a subroutine that for any $\bz$ finds a $\by \in \P$ such that $\bz^T\by \geq (1-\delta/2) \max \{ \bz^T\by | \by \in \Q\}$, to find a solution $\t{\by}$ of \ref{lpu:002}, otherwise decrease $\beta \leftarrow (1-\delta)\beta$.
{\footnotesize
\[ \bw^T \t{\by} \geq (1-\delta)\beta, \quad \bu(\L)^T \bA \t{\by}  \leq \frac{\gamma}{1-\delta}, \quad \t{\by} \in \P \lptag \label{lpu:002}\]
}

\STATE Set {\footnotesize$\by \leftarrow (1-\sigma) \by + \sigma \t{\by}$} where {\footnotesize$\sigma=\epsilon/(4\alpha\lambda_0)$}. \label{updatestepu}
\ENDWHILE 
 \end{algorithmic}
 \caption{A Dual Thresholding Multiplicative Weights Algorithm. $M \gg m \gg K \geq 1$. \label{alg:thresh}}
}
\end{algorithm}

\noindent The remainder of this section uses the notation introduced in Algorithm~\ref{alg:thresh}.
\smallskip
\begin{definition}
\label{def:001}
Extend $\bu$ as $u_{\ell}=\exp(\alpha \lambda_{\ell} )/\t{b}_{\ell}$ for all $\ell$ in Line~\ref{defgammau} of Algorithm~\ref{alg:thresh}.
Define $\Psi =\sum_\ell e^{\lambda_{\ell} \alpha} = \bu^T\t{\bb}$ which only depends on the current solution.
\end{definition}

\begin{lemma}
\label{lem:connect}
Suppose $\Psi \leq \gamma + \frac{\delta\gamma}{\lambda_0}$.
If $\t{\beta}= \max \{ \bw^T\by \mid \bu(\L)^T \bA \by \leq \frac{\gamma}{1-\delta}, \by \in \Q\}$ exists then we have an algorithm for Line~\ref{effectivestepu} of Algorithm~\ref{alg:thresh} for any $\beta<\t{\beta}$. Further the final output of Algorithm~\ref{alg:thresh} satisfies $\bw^T\by \geq (1-\delta) \min \{ (1-\delta)^2 \t{\beta}, \beta_0\}$.
\end{lemma}

\begin{proof}
The assumption implies that $\bu(\L)^T \bA \by \leq \bu^T \bA \by \leq \bu^T\t{\bb} = \Psi \leq \gamma/(1-\delta)$ for any $\by$ satisfying $\bA\by \leq \bb$.
Part (1) of Theorem~\ref{apxtheorem} applies with $\wi{\P}=\P$ and $\wi{\Q}=\Q$ and we succeed in solving \ref{lpu:002}. This implies that $\beta$ cannot decrease below $(1-\delta)\t{\beta}$; the last decrease of $\beta$ corresponds to a value greater than $\t{\beta}$. 
\end{proof}

For $\alpha = \frac{1}{f(\delta)} \ln \left(\frac{M\lambda_0}{\delta} \right)$ we satisfy the precondition of Lemma~\ref{lem:connect} because 
{\small
\[ \Psi -\bu(\L)^T\t{\bb}=\bu^T\t{\bb} - \bu(\L)^T\t{\bb} = \sum_{\ell: \lambda_\ell < \lambda - f(\delta)} 
\exp(\alpha \lambda_{\ell} ) \leq \frac{\delta}{\lambda_0} e^{\alpha\lambda} \leq \frac{\delta}{\lambda_0} \bu(\L)^T\t{\bb} = \frac{\delta \gamma}{\lambda_0} \]
}
However we include the condition in the statements of Lemma~\ref{lem:decrease} and Theorem~\ref{basetheorem} because in the specific case of $b$--matching we would use a value of $\alpha$ which is better by a factor of $1/\delta$ -- therefore we can use  
Lemma~\ref{lem:decrease} and Theorem~\ref{basetheorem} without any change. A smaller value of $\alpha$ will result in faster convergence.

\smallskip
\begin{lemma}
\label{lem:decrease}
Suppose $\Psi \leq \gamma + \delta\gamma/\lambda_0$. 
Let $\Psi'$ be the new potential corresponding to the new $\by$ computed in Step 
\ref{updatestepu}. Then if $\lambda \geq 4$ and $\epsilon=1/8$ then $\Psi' \leq (1- \frac{\lambda}{128\lambda_0}) \Psi$ otherwise we have $\Psi' \leq (1- \frac{\epsilon^2}{8\lambda_0}) \Psi$.
\end{lemma}

\smallskip
\begin{proof}
Observe that the algorithm maintains the invariant $\lambda \geq 1+ 8\epsilon$, even though $\lambda$ may not be monotone.
After the update, let the new current solution be denoted
by $\by''$, i.e., $\by'' = (1-\sigma) \by + \sigma \t{\by}$ where $\t{\by}$ is the solution of \ref{lpu:002}. Recall $\alpha\sigma = \epsilon/(4\lambda_0)$. Let

{\small
\[  \lambda''_{\ell} = \left(\bA\by''\right)_\ell/\t{b}_\ell, \quad \mbox{and} \quad \t{\lambda}_{\ell} = \left(\bA\t{\by}\right)_{\ell}/\t{b}_\ell
\qquad \mbox{ therefore } \lambda''_{\ell} = (1-\sigma) \lambda_{\ell} + \sigma \t{\lambda}_{\ell} \quad \forall \ell
 \]
}
\noindent
Observe that $\sum_{\ell} e^{\alpha \lambda_\ell} \lambda_\ell  = 
\sum_{\ell} \left( \t{b}_{\ell}u_{\ell}\right) \frac{\left(\bA\by\right)_\ell}{\t{b}_{\ell}} =\bu^T\bA\by$. Likewise $\sum_{\ell} e^{\alpha \lambda_\ell} \t{\lambda}_\ell 
=\sum_{\ell} \left( \t{b}_{\ell}u_{\ell}\right) \frac{\left(\bA\t{\by}\right)_\ell}{\t{b}_{\ell}}
= \bu^T\bA\t{\by}$.

\smallskip
Since $\t{\by} \in \P$ we have $\t{\lambda}_{\ell} \leq \lambda_0$ from Step~\ref{line2u} of Algorithm~\ref{alg:thresh}.
Since we repeatedly take convex combination of the current candidate
solution $\by$ with a $\t{\by} \in \P$, and the initial solution
satisfies $\lambda \leq \lambda_0$; we have $\lambda_\ell \leq \lambda_0$ throughout the algorithm.
Since $\lambda_\ell \leq \lambda_0$ we have
all $|\alpha \sigma (\t{\lambda}_\ell - \lambda_\ell)| \leq \epsilon/4$. 
Now for $|\Delta| \leq \frac{\epsilon}4 \leq \frac14$; we have
$e^{a+\Delta} \leq e^a ( 1 + \Delta + \epsilon |\Delta|/2)$. Therefore:

{\small
\[  e^{\alpha \lambda''_\ell}  \leq e^{\alpha \lambda_\ell} \left(1 + \sigma \alpha(\t{\lambda}_\ell - \lambda_\ell) + \frac12 \epsilon\sigma \alpha(\t{\lambda}_\ell +  \lambda_\ell) \right)  = e^{\alpha \lambda_\ell} + \left(1+\frac\epsilon2\right) \sigma\alpha \t{\lambda}_\ell e^{\alpha\lambda_{\ell}} 
- \left(1-\frac\epsilon2\right) \sigma \alpha \lambda_{\ell}e^{\alpha\lambda_\ell} \]
}
\noindent which implies that
{\small
\begin{align}
\Psi' & = \sum_\ell e^{\alpha \lambda''_\ell} \leq \Psi + 
\left(1+\frac\epsilon2\right) \sigma\alpha \sum_{\ell} e^{\alpha \lambda_\ell} \t{\lambda}_\ell  - \left(1-\frac\epsilon2\right) \sigma \alpha 
\sum_{\ell} e^{\alpha \lambda_\ell} \lambda_\ell  \nonumber \\
& = \Psi + 
\left(1+\frac\epsilon2\right) \sigma\alpha \bu^T\bA \t{\by} - \left(1-\frac\epsilon2\right) \sigma \alpha \bu^T\bA \by \label{maineqn} 
\end{align}
}
\noindent $\t{\by}$ satisfies \ref{lpu:002} and therefore $\bu(\L)^T\bA \t{\by} \leq \frac{\gamma}{1-\delta}$, which along with $\t{\lambda}_\ell \leq \lambda_0$, $\delta \leq 1/16$, implies:

{\small
\begin{align}
\bu^T\bA \t{\by} & = \bu(\L)^T\bA \t{\by} + \sum_{\ell: \lambda_\ell < \lambda - f(\delta)} \t{\lambda}_{\ell} e^{\lambda_\ell \alpha} \leq \frac{\gamma}{1-\delta} + \lambda_0 \sum_{\ell: \lambda_\ell < \lambda - f(\delta)} e^{\lambda_\ell \alpha} \leq \frac{\gamma}{1-\delta} + \lambda_0 \left( \Psi - \gamma \right) \leq (1+3\delta)\gamma \label{ub1} 
\end{align}
}
\noindent Finally observe that since $\lambda > 1+8\epsilon$, $f(\delta) \leq \delta$, and $\gamma=\sum_{\ell:\lambda_\ell > \lambda - f(\delta)} \lambda_\ell e^{\lambda_\ell
\alpha}$, 

{\small
\begin{align}
\bu^T\bA\by & \geq \bu(\L)^T\bA\by  = \sum_{\ell:\lambda_\ell \geq \lambda - f(\delta)} \lambda_\ell e^{\lambda_\ell
\alpha} \geq (\lambda - f(\delta)) \left( 
\sum_{\ell: \lambda_{\ell} \geq \lambda - f(\delta)}
e^{\lambda_\ell \alpha} \right) 
 = (\lambda - f(\delta)) \gamma
\geq (\lambda -\delta)\gamma \label{lb2}
\end{align}
}
\noindent Using Equations~\eqref{maineqn}--\eqref{lb2} we have:
{\small
\begin{align}
\label{almost}
\Psi' \leq \Psi - \left(\lambda - 1 - 4\delta - \frac{\left( \lambda  + 4\delta + 1\right) \epsilon}{2}\right)\gamma \alpha \sigma 
\end{align}
}
Note $\delta\leq\epsilon\leq 1/8$.
If $\lambda \geq 4$ then $\lambda - 1 - 4\delta - \frac{\left( \lambda  + 4\delta + 1\right) \epsilon}{2} \geq \lambda/2$. From the statement of the lemma, $\Psi \leq 2 \gamma$.
Note that if $\lambda \geq 4$ and $\epsilon=\frac18$ then $\frac{\lambda\alpha\sigma}{4} = \frac{\lambda}{128\lambda_0}$. Thus in the case when $\lambda \geq 4$ and $\epsilon=1/8$,
{\small
\[ \Psi'\leq \Psi - \frac{\lambda \alpha \sigma \gamma}{2} \leq \Psi \left(1 - \frac{\lambda\alpha\sigma}{4} \right) \leq \Psi \left(1 - \frac{\lambda}{128\lambda_0} \right)
\]
}
Otherwise using $\lambda \geq 1 + 8 \epsilon$,
{\small
\[ \lambda - 1 - 4\delta - \frac{\left( \lambda  + 4\delta + 1\right) \epsilon}{2} \geq (1+8\epsilon)\left(1-\frac\epsilon2\right) - (1+4\delta)\left(1+\frac\epsilon2\right) \geq 7\epsilon - 4 \delta - (4\epsilon^2 + 2\delta\epsilon) \geq \epsilon
\]
}
Combining with Equation~\eqref{almost} we get $\Psi'\leq \Psi - \epsilon\sigma\alpha\gamma \leq \Psi \left( 1 - \frac{\epsilon\alpha\sigma}
{2} \right) = \Psi( 1 - \frac{\epsilon^2}{8\lambda_0})$.
\end{proof}

\smallskip
\noindent
Lemma~\ref{lem:decrease} proves that $\Psi$ decreases monotonically, even though $\lambda$ may not. 

\begin{theorem}\label{basetheorem}
Suppose $\Psi \leq \gamma + \delta\gamma/\lambda_0$. 
Algorithm~\ref{alg:thresh} converges within $\tau=O \left( \lambda_0 \left( \frac{ \ln (2K)}{\delta^2} + \frac{\alpha}{\delta} + \alpha \ln \lambda_0 \right) \right)$ invocations of \ref{lpu:002} and provides a solution as described in line \ref{exitruleu}.
\end{theorem}

\smallskip
\begin{proof}
Observe that $e^{\alpha\lambda} \leq \gamma \leq K e^{\alpha\lambda}$ since there are at most 
$K$ constraints in $\L$. Since $\gamma \leq \Psi \leq 2\gamma$
we know that $e^{\alpha\lambda} \leq \Psi \leq 2K e^{\alpha\lambda}$.
We partition the number of iterations into three parts:

\smallskip
\begin{enumerate}[~~~~~(${\cal C}$1)]
\item The number of iterations till we observe $\lambda < 4$ for the first time.
\item The number of iterations after we observe $\lambda < 4$ for the first time
till $\lambda <2$ for the first time. 
\item The number of iterations since $\lambda <2$ for the first time.
\end{enumerate}

\smallskip
Observe (${\cal C}$1) and (${\cal C}$2) correspond to the first superphase during which $\epsilon=\frac18$.
In case (${\cal C}$1), consider the total number of iterations when $4
\leq 2^{j} \leq \lambda \leq 2^{j+1}$. The potential $\Psi$ must be
below $2Ke^{\alpha 2^{j+1}}$. If we perform $r$ updates to $\by$ then the potential decreased by at
least $(1-2^{j-5}/\lambda_0)^r$ but if $r \geq
\frac{\lambda_0}{2^{j-5}}(\ln (2K) + 2^{j+1} \alpha)$ then the new
potential will be below $1$, which is impossible since the potential must be at least 
$e^{4\alpha}$. Therefore the {\bf total} number of updates corresponding to $4
\leq 2^{j} \leq \lambda \leq 2^{j+1}$ for a fixed $j$ is at most $128
\lambda_0 2^{-j} \ln (2K) + 256\lambda_0\alpha$. Summed over all
$j \geq 2$ the number of updates in case (${\cal C}$1) is $O(
\lambda_0 \ln (2K) + \lambda_0\alpha \ln \lambda_0)$.

\smallskip
In case (${\cal C}$2), the potential decreases by a factor $(1 - 1/(512\lambda_0))$. 
By the same exact argument as in case (${\cal C}$1), if the number of updates exceed $
512\lambda_0(\ln (2K) + 4\alpha)$ then the potential would be below $1$, which again is impossible since the potential must be at least $e^{2\alpha}$. 
Therefore the number of updates in this case is $O(\lambda_0(\ln (2K) + \alpha))$.

\smallskip
\noindent 
In case (${\cal C}$3), we partition a
superphase into a number of different {\bf phases}.
\begin{definition}
\label{phasedef}
A phase starts when a superphase starts, and we remember the $\lambda$ value at the start of a phase. Let the value of $\lambda$ at the start of phase $t$ be $\lambda_t$. If at some point of time during phase $t$, we observe $\lambda < (1-\delta)\lambda_t$, then we mark the end of phase $t$ and start phase $t+1$ with $\lambda_{t+1}=\lambda$. A phase also ends when $\lambda < 1+ 8\epsilon$ because the  corresponding superphase ends as well.
\end{definition}

\smallskip
Note that while $\lambda$ is not monotone, $\lambda_t$ are monotone and we will use $\lambda_t$ to bound the number of iterations.
Since in
each phase $\lambda$ decreases by at least $(1-\delta)$ factor the
number of phases in the superphase corresponding to $\epsilon$ is
$O\left(\log_{\frac{1}{(1-\delta)}}
\frac{(1+12\epsilon)}{(1+8\epsilon)} \right) =O( \frac{4\epsilon}{\delta})$. 
In each of these phases (say in phase $t$) we have
$e^{\alpha(1-\delta)\lambda_t} \leq \Psi \leq 2K e^{\alpha\lambda_t}$
and $\Psi$ decreases by a factor $(1-\epsilon^2/(8\lambda_0))$. Note that if
$\Psi$ decreases by a factor of $2K e^{\delta\alpha\lambda_t}$ then
the phase would be over. Therefore the number of updates in a phase is at most 
$\frac{4\lambda_0}{\epsilon^2}( \ln (2K) + 2\delta\alpha)$ using $\lambda_t \leq 2$ -- note $\lambda_t$ decreases monotonically. 
Therefore the number of updates in a superphase corresponding to an $\epsilon$ is 
{\small
\begin{equation}
\label{geomseries}
 \frac{4\epsilon}{\delta} \frac{4\lambda_0}{\epsilon^2}( \ln (2K) + 2\delta\alpha) = \frac{16 \lambda_0}{\delta \epsilon} \ln (2K) + \frac{48 \alpha \lambda_0}{\epsilon} 
 \end{equation}
 }
However note that $\epsilon$ decreases by a factor of at least $2/3$ (unless it is close to $\delta$) and the terms in Equation~\eqref{geomseries} define a geometric series each and the smallest two values of $\epsilon$ dominate because we may have $\epsilon=1.01\delta$ followed by $\epsilon=\delta$.
Therefore the total number of updates in this case is $O(\lambda_0 (\frac{\ln (2K)}{\delta^2} + \frac{\alpha}{\delta}))$. Summing up the three cases, the number of updates is $O(\lambda_0 (\frac{\ln (2K)}{\delta^2} + \frac{\alpha}{\delta} + \alpha \ln \lambda_0))$. 
The bound on the number of non-zero edges follows from multiplying the number of updates by the number of nonzero entries guaranteed by \ref{lpu:002} which is $O(n')$. 
This proves Theorem~\ref{basetheorem}.
\end{proof}

\smallskip
\noindent
Note that $\lambda_t$ and its monotonicity was used in (${\cal C}$3), even though $\lambda$ need not be monotone.

\subsection{Applying Algorithm~\ref{alg:thresh} to Uncapacitated $b$--matching} 
\label{sec:app}
$\bA\by \leq \t{\bb}, \Q$ 
and $\P$ are defined below:
{\small
\begin{align*}
\{\bA\by \leq \t{\bb} \}= &\left \{ \begin{array}{lll}
& \displaystyle \sum_{i:(i,j) \in E} y_{ij} \leq \t{b}_i & \forall i \in V \ \ \mbox{ where } \t{b}_i=(1-4\delta) b_i \\
& \displaystyle \sum_{(i,j) \in E: i,j \in U} y_{ij} \leq \t{b}_U & \forall U \in \Od \ \ \mbox{ where } \t{b}_U=\left (
\cvol{U} - \frac{\delta^2\bnorm{U}^2} {4} \right)
\end{array} \right. \\
\Q = & \left \{ \begin{array}{ll}
\displaystyle \sum_{j:(i,j) \in E} y_{ij} \leq b_i & \forall i \in V \\
y_{ij} \geq 0 & \forall (i,j) \in E
\end{array} \right. \\
\P = & \left \{ \begin{array}{ll}
\displaystyle \sum_{j:(i,j) \in E} y_{ij} \leq \lambda_0 b_i/2 & \forall i \in V \\
y_{ij} \geq 0 & \forall (i,j) \in E
\end{array} \right.  
\end{align*}
 }
\noindent Note the number of constraints in $\bA \by \leq \t{\bb}$ is $M=n^{O(1/\delta)} \gg m$, the number of edges. We observe that $\{\by/\lambda_0 \mid \by \in P\} \subseteq \{ \by \mid,  
\bA\by \leq \t{\bb}, \Q \}$ for $\delta\leq1/16$.

\begin{definition}
Let $\t{\beta} = \{ \bw^T \by \mid \bA \by \leq \t{\bb}, \by \in \Q \}$; note $\t{\beta}$ exists.
\end{definition}
We can now apply Lemma~\ref{lem:connect} and Theorem~\ref{basetheorem} and obtain a solution $\bw^T\by \geq (1-\delta)^2 \t{\beta}$ and $\bA\by \leq (1+8\delta) \t{\bb}$ (we ignore $\P$).
We can extract a $(1-O(\delta))$ approximation to the optimum uncapacitated $b$--matching from $\by$. Set $y^{\dagger}_{ij} = \frac{(1-\delta)}{(1+8\delta)} y_{ij}$. Since $\t{b}_i \leq b_i$ and $\t{b}_U \leq \cvol{U}$ the constraints corresponding to vertices and $U \in \Od$ are satisfied in \ref{lpbm}. For the $U \not \in \Od$, which have  $\bnorm{U} \geq 1/\delta$ observe that
{\small
$$
\sum_{j:(i,j) \in E} y^{\dagger}_{ij} \leq (1-\delta)b_i  \quad \implies \quad \sum_{(i,j) \in E, i,j \in U} y^{\dagger}_{ij} \leq \frac12 \sum_{i \in U}  \sum_{j:(i,j) \in E} y^{\dagger}_{ij} \leq (1-\delta) \frac{\bnorm{U}}{2} \leq \cvol{U}
$$ } At the same time $\bw^T\by \geq (1-\delta)^2 \t{\beta}$ and thus
$\bw^T\by^{\dagger} \geq
\frac{(1-\delta)^2}{(1+8\delta)}\t{\beta}$. Note that the same
argument also proves that 
if we
consider the optimum solution of $\max \{ \bw^t\by \mid \bA \by \leq
\t{\bb}, \by \in \Q \}$ and multiply by $(1-\delta)$ then we satisfy
the constraints of \ref{lpbm}. Therefore $(1-\delta)\t{\beta} \leq \beta^*$.
We observe $ \t{\beta} \geq (1-4\delta)
\beta^*$. The latter equation follows from 
Likewise consider the optimum
$b$--matching and multiply that solution by $(1-4\delta)$.
That modified solution $\by'$ satisfies
$\bA \by' \leq \t{\bb}$ when $\delta\leq 1/16$.  Thus $ \t{\beta} \geq (1-4\delta)\beta^*$.
Therefore
$\by^{\dagger}$ provides a $(1-O(\delta))$-approximation to \ref{lpbm}
(page \pageref{lpbm}), the uncapacitated $b$--matching LP that
characterizes the optimum solution.

We set $\lambda_0=12$. The initial solution is a solution of the
bipartite relaxation (Theorem~\ref{uncapapx}) multiplied by
$\lambda_0/2=6$. This solution value $\beta_0$ will be at least
$\beta^*_b$ due to the approximation guarantee. But $\beta^*_b \geq
\t{\beta}$, based on the constraints. On the other hand $\beta_0$ can be as large as $6
\beta^*_b$ (the bipartite optimum) which is at most $9 \beta^*$ since
the gap between bipartite and non-bipartite solution is at most a
factor of $1.5$. But $9\beta^* \leq \frac1{1-4\delta} \t{\beta} \leq
12 \t{\beta}$. This proves the bound on the initial solution. The
parameter $\lambda_0$ can be improved (e.g., it can be argued that
$\t{\beta} \leq \beta^*$), but that only affects the running time by a
$O(1)$ factor. Observe that an algorithm for solving \ref{lpu:002} is
also provided by Theorem~\ref{uncapapx} and multiplying the solution
by $6$. We now focus on Step~(\ref{exactcompu}).

We set $f(\delta)=\delta^3/10$. We show that in Step~(\ref{exactcompu}) the sparse set of constraints will be of size at most $K=2n$, since that collection would define a laminar family.
More specifically:

\smallskip
\begin{lemma}\label{zulambda1}
Let $\frac{3}{\sqrt{n}} \leq \delta\leq \frac1{16}$. If $\lambda > 1 + 8\delta$ then we can find 
$\L=\{ U | \lambda_U \geq \lambda -\delta^3/10; U \in \Od \}$ in 
$O(m + n \poly (\delta^{-1},\ln n))$ time. We find $\L$ without knowing $\lambda$ and once $\L$ is known, we know $\lambda$ as well.
\end{lemma}

\smallskip
\begin{proof}
Let $\blam=\max\{1,\max_i (1-4\delta) \lambda_i\} =
  \max\{1,\max_i \sum_j y_{ij}/b_i \}$ and $\h{y}_{ij} = y_{ij}/\blam$.
Let $\hat{\lambda}_U = \sum_{i,j \in U} \h{y}_{ij}/\t{b}_U$ and $\hat{\lambda} = \max_{U \in \Od} \hat{\lambda}_U$.
Observe that $\lambda_U=\blam \hat{\lambda}_U$ and if $\lambda >
\max_i \lambda_i$ then $\lambda =\blam \h{\lambda}$. Note $\sum_j \h{y}_{ij} \leq b_i$. 

Suppose that  $\hat{\lambda} \leq 1+3\delta$ and $\blam=1$.
Then for all $U$ we have $\lambda_U = \blam \h{\lambda}_U 
\leq \blam\hat{\lambda} \leq  \blam
(1+3\delta) < 1+ 8\delta$ and $\max_i \lambda_i \leq 1/(1-4\delta) \leq
1+8\delta$ for $\delta\in (0,\frac1{16}]$. This contradicts the assumption that 
$\lambda > 1 + 8\delta$. 
Therefore, if  $\hat{\lambda} \leq 1+3\delta$ then we must have $\blam>1$.
Now consider the vertex $i$ which
defined $\blam$; then 
{\small
 \begin{align*}
  \lambda \geq \lambda_i  = \frac{\blam}{1-4\delta} \geq (1+4\delta)\blam
              \geq (1+3\delta)\blam + \delta\blam > \h{\lambda} \blam + \delta
 \end{align*} } 
\noindent which implies $\lambda - \delta \geq \lambda_U$ for every $U$.
 In this case $\L=\emptyset$ and $|\{ U: \lambda_U \geq \lambda -
 \delta^x; U \in \Od\}|=0$ for $x\geq 2$. 
Therefore the remaining case is $\hat{\lambda}> 1+3\delta$.
But in this case Theorems~\ref{thm:main1} and \ref{thm:main2} apply because 
 we satisfy $\sum_j \h{y}_{ij} \leq b_i$. To find $L$, compute
 $\blam,\h{y}_{ij}$ and run the algorithm in Theorem~\ref{thm:main2}. 
 We can compute $\lambda =
 \{\blam\h{\lambda},\max_i \lambda_i \}$ and return the sets
 satisfying $\lambda_U \geq \lambda - \delta^3/10$. 
\end{proof}

\noindent To compute $\bz^T=\bw^T - \bu(\L)^T\bA$ in \ref{lpu:002}, note that
$z_{ij} = w_{ij} - \varrho \left( x_i + x_j + \sum_{U:i,j \in U}
z_U\right)$ where $x_i$ corresponds to the vertices and $z_U$
correspond to the odd set in $\bu$. These weights can be computed in
$O(1)$ time if we precompute the $\sum_{U \in \L,i,j \in U} z_U$ for
each pair of vertices $(i,j)$. Note there can be at most
$\sum_{s=1}^{O(1/\delta)} \frac{n}{s}s^2 = O(n\delta^{-2})$ such pairs
for any $\L$. If we use $\alpha = \frac{1}{f(\delta)} \ln \left(
\frac{M\lambda_0}{\delta} \right)$ we get an algorithm which converges
in $O(\delta^{-5} \log n)$ invocations of \ref{lpu:002} and provides a
$(1-O(\delta))$-approximation. We show that $\alpha$ can be chosen to
be smaller.

\begin{lemma}\label{zulambda2} 
For $\frac{3}{\sqrt{n}} \leq \delta\leq \frac1{16}$, $n\geq \lambda_0$, 
$\alpha=50\delta^{-3} \ln n$ and the definition of $\gamma$ in Algorithm~\ref{alg:thresh}:
{\small
\[
\sum_{i: \lambda_i \leq \lambda -\delta^3/10} \t{b}_ix_i + \sum_{\lambda_U < \lambda - \delta^3/10} z_U\t{b}_U 
= \sum_{i: \lambda_i \leq \lambda -\delta^3/10} e^{\lambda_i \alpha} +  \sum_{\lambda_U < \lambda - \delta^3/10} e^{\lambda_U \alpha} \leq \frac{e^{\lambda\alpha} \delta}{n}  \leq
 \frac{\delta \gamma}{n}
< \frac{\delta \gamma}{\lambda_0}
\]
}
\end{lemma}

\begin{proof}
Observe that $e^{\lambda \alpha} \leq \gamma$
since 
$\lambda = \max_i \lambda_i$ or $\lambda=\max_{U \in L} \lambda_U$. 
We first focus on $U \in \Od$. Observe the $U$ considered in the left hand side of the statement of the inequality in the Lemma can be partitioned into three classes (i) $\lambda_U \leq \lambda -\delta^2$ (ii) $\lambda - \delta^{(x_0+3)/2} \leq \lambda_U \leq \lambda - \delta^3/10$, where $x_0$ is the largest value of $x \geq 2$ such that $\delta^{(x-3)/2} \geq 2$. Note that $x_0 < 3$ exists\footnote{ Consider $h=\delta^{(x-3)/2}$ as $x$ increases from $2$ to $3$. The value of $h$ decreases from $\delta^{-1/2}$ to $1$.} given $\delta \leq 1/16$, and (iii) $\lambda - \delta^2 \leq \lambda_U < \lambda- \delta^{(x_0+3)/2}$.

For case (i) observe that the corresponding 
$e^{\lambda_U \alpha} \leq e^{\lambda\alpha - \delta^2\alpha} \leq e^{\lambda \alpha} e^{-50\delta^{-1}\ln n } =e^{\alpha\lambda}/n^{(50/\delta)}$. There are at most $n^{1/\delta}$ such sets and therefore 
$\sum_{U: \lambda_U \leq (1-\delta^2)\lambda} e^{\lambda_U \alpha} \leq e^{\lambda \alpha}/n^{(49/\delta)}$.

For case (ii), perform the same transformation as in the first two lines of Lemma~\ref{zulambda1}. The bound on $\lambda_U$ corresponds to $\hat{\lambda}_U \geq (\hat{\lambda} - \delta^{(3+x_0)/2})/\blam \geq \hat{\lambda} - \delta^{(3+x_0)/2}$ since $\blam \geq 1$. Using
Theorem~\ref{thm:main1} we know that there are at most $n^3 +
\left(n/\delta\right)^{1+\delta^{(x_0-3)/2}} \leq 2 (n/\delta)^3$ such sets.
For each such set $U$, $e^{\lambda_U \alpha} \leq e^{\lambda \alpha} e^{ - \delta^3\alpha/10} = e^{\lambda \alpha} e^{-5\ln n} = e^{\lambda \alpha}/n^5$. Summing up over such $2(n/\delta)^3$ sets the total contribution to the left hand side of the inequality in the statement of the lemma is at most $2e^{\lambda \alpha}\delta^{-3}/n^2$.

For case (iii), we partition the interval $(\lambda - \delta^2, \lambda- \delta^{(x_0+3)/2}]$ into subintervals of the form $(\lambda - \delta^x,\lambda - \delta^{(x+3)/2}]$ for different values of $x$. The last subinterval corresponds to $x=x_0$. If we set $x_1=2x_0 - 3$ we have $(3+x_1)/2 = x_0$  and thus $x_1 < x_0 < 3$, which corresponds to the second subinterval. The $j$-th subinterval is defined by $x_j$ satisfying $3+x_j = 2x_{j-1}$. The number of such subintervals is at most $2 + \log \log (1/\delta)$.

Consider the case $ \lambda_U \in (\lambda - \delta^x,\lambda - \delta^{(x+3)/2}]$. 
Again, performing the transformation as in the first two lines of Lemma ~\ref{zulambda1}, we get that $\hat{\lambda}_U \geq (\hat{\lambda} - \delta^x)/\blam \geq \hat{\lambda} - \delta^x$ (again, $\blam \geq 1$). Using Theorem~\ref{thm:main1} and $\delta \geq 1/\sqrt{n}$, the number of odd-sets corresponding to this subinterval is at most 
$$n^3 +
\left(n/\delta\right)^{1+\delta^{(x-3)/2}} \leq 2 (n/\delta)^{1+\delta^{(x-3)/2}} \leq 2n^{1.5 + 1.5 \delta^{(x-3)/2}}$$
However note that $\lambda_U \leq \lambda - \delta^{(3+x)/2}$ and therefore
$e^{\alpha \lambda_U}$ is at most 
$$e^{\alpha \lambda}/e^{\alpha \delta^{(x+3)/2}} = \frac{e^{\alpha \lambda}}{e^{50 \delta^{(x-3)/2}\ln n}} = 
\frac{e^{\alpha \lambda}}{n^{50 \delta^{(x-3)/2}}}$$
Therefore the total contribution to the left hand side of the inequality in the statement of the lemma for all $U$ such that 
$ \lambda_U \in (\lambda - \delta^x,\lambda - \delta^{(x+3)/2}]$ is at most (since $x\leq x_0$ ):
$$ \frac{2e^{\alpha \lambda}}{n^{48.5 \delta^{(x-3)/2} - 1.5}} \leq \frac{2e^{\alpha \lambda}}{n^{97  - 1.5}} \leq \frac{e^{\alpha \lambda}}{n^{94}} $$

For $i \in V$ such that $\lambda_i \leq \lambda -\delta^3/10$ the calculation as in the case (iii) applies and $\sum_{i: \lambda_i \leq \lambda -\delta^3/10} e^{\lambda_i \alpha} \leq 2e^{\lambda \alpha}\delta^{-3}/n^2$ as well. 
Since $\frac{1}{n^{49/\delta}}+ \frac{2+\log \log (1/\delta)}{n^{94}} + \frac{4\delta^{-3}}{n^2} \leq \frac{5\delta^{-3}}{n^2} \leq \frac{\delta}{n} $; 
the lemma follows.
\end{proof}

We can now
conclude Theorem~\ref{thmone} based on the discussion above,
Theorems~\ref{apxtheorem}, \ref{basetheorem} and
Lemmas~\ref{zulambda1}--~\ref{zulambda2}:

\begin{ntheorem}{\ref{thmone}}
Given any non-bipartite graph, for any $\frac{3}{\sqrt{n}}<\delta\leq 1/16$ we find a
$(1-O(\delta))$-approximate maximum fractional weighted $b$-matching using 
additional ``work'' space (space excluding the read-only input) 
$O(n\poly (\delta^{-1},\ln n))$ 
and
making $T=O(\delta^{-4} (\ln (1/\delta)) \ln n)$ passes over the
list of edges.  The running time is $O(mT + n\poly(\delta^{-1},\ln n))$.  
\end{ntheorem}

\section{Proof of Theorem~\ref{thm:main1}}
\label{proof:main1}
Before proving Theorem~\ref{thm:main1}, recall that $\t{b}_U = \cvol{U} - f(\bnorm{U})$ where 
$f(\ell) = \frac{\delta^2 \ell^2}{4}$ and $\delta \in (0,\frac1{16}]$. 
We can verify that $f(\ell)$ is convex, monotonic for $0\leq \ell \leq 2/\delta $ and: 

\begin{enumerate}[({$\mathcal F$}\!1):]\label{allaboutf}
\item For $3 \leq \bnorm{U}\leq 2/\delta - 1$ (irrespective of odd or even)
we have $\t{b}_U \geq \left(1-\delta\right) \cvol{U}$.

\item For any $\ell_1,\ell_2$; $f(\ell_1) + f(\ell_2) = f(\ell_1+\ell_2-1) - 
(2\ell_1\ell_2 - 2\ell_1 -2\ell_2 + 1)\frac{\delta^2}{4}$.

\item For integers $\ell_1,\ell_2,\ell_3,\ell_4 \in [3,2/\delta]$ and $t \geq 0$,
such that $\ell_1 + 2t \leq \ell_2 \leq \ell_3  \leq \ell_4 - 2t$ 
and $\ell_1+\ell_4=\ell_2 + \ell_3$, we have
$f(\ell_2)+ f(\ell_3) \leq f(\ell_1) + f(\ell_4) - 2t^2\delta^2$. 
\end{enumerate}

\begin{ntheorem}{\ref{thm:main1}}
For a graph $G$ with $n$ vertices and any non-negative edge weights 
$\h{y}_{ij} = \h{y}_{ji}$ such that $\h{y}_{ii}=0$ for all $i$, $\sum_j \h{y}_{ij} \leq b_i$ 
for all $i$, and $\delta\in (0,\frac1{16}]$, define:
$\h{\lambda}_U = \dfrac{\sum_{(i,j): i,j\in U} \h{y}_{ij}}{\t{b}_U}$ 
and $\h{\lambda} = \max_{U \in \Od} \h{\lambda}_U$.
If $\hat{\lambda}\geq 1+3\delta$, 
the set $L_1=\{U:\hat{\lambda}_U \geq \hat{\lambda}-\delta^3; U \in \Od\}$ defines a laminar 
family. Moreover for any $x \geq 2$ we have $|\{ U:\hat{\lambda}_U \geq \hat{\lambda}-\delta^x; U \in \Od\}| \leq n^3 + \left(n/\delta\right)^{1+\delta^{(x-3)/2}}$.
\end{ntheorem}\par

\begin{proof}
Consider two sets $A_1,A_2 \in \Od$ such that $\h{\lambda}_{A_1},\h{\lambda}_{A_2} \geq 
\hat{\lambda}-\delta^x > 1+2\delta$ (since $x\geq 2$) and neither $A_1 - A_2,A_2-A_1 \neq \emptyset$.
For any set $U$ (with $\bnorm{U} \geq 1$, even or odd, large or small) define 
$\h{Y}_{U} = \sum_{(i,j):i,j \in U} \h{y}_{ij}$ and $\t{b}_U$. For $\bnorm{U}=1$ we have $\h{Y}_{U}=0$.
Let $\h{\lambda}_{U} = \h{Y}_{U}/\t{b}_U$. 
There are now two cases.

\medskip\noindent{\bf Case I: $\bnorm{A_1\cap A_2}$ is even}. Let  $D=A_1\cap A_2$ and $t=\bnorm{D}/2$.
 Let $Q_1=\sum_{i\in D}\sum_{j\in A_1-A_2} \h{y}_{ij}$ (the cut between $D$ and $A_1-A_2$ using the edge weights $\h{y}_{ij}$) 
 and $Q_2=\sum_{i\in D}\sum_{j\in A_2-A_1} \h{y}_{ij}$. Without loss of generality,
 assume that $Q_1\leq Q_2$ (otherwise we can switch $A_1,A_2$).
 Let $C=A_1-A_2$ and $A=A_1$. Let $2\ell-1=\bnorm{C}$ which is odd. From the definitions of $\h{Y}_C,\h{Y}_D$ we have $\h{Y}_C=\h{Y}_{A}-Q_1-\h{Y}_D$ and
$\h{Y}_D\leq \frac12(\sum_{i\in D}\sum_j \h{y}_{ij} -Q_1-Q_2)\leq \frac{\bnorm{D}}{2}-\frac{Q_1+Q_2}{2}$. Using $Q_1\leq Q_2$ we get:

{\small
\begin{align}
\h{Y}_C \geq \h{Y}_{A}-\frac{\bnorm{D}}{2}-\frac{Q_1}{2}+\frac{Q_2}{2}
      \geq \h{Y}_{A}-\frac{\bnorm{D}}{2} = \h{Y}_{A} - t. 
\label{eqn1}
 \end{align}
}
\noindent Now $\h{Y}_{A} = \h{\lambda}_{A}\t{b}_A > (1+2\delta)(1-\delta)\cvol{A} \geq  \cvol{A}$ using Condition ${\mathcal F}1$, 
and the lower bound on $\h{\lambda}$. Therefore $\h{Y}_{A}>t$ and $\h{Y}_C >0$ which 
means $\bnorm{C}\geq 3$. Therefore we can refer to $\t{b}_C,\h{\lambda}_C$. 
Since $\bnorm{D}=\bnorm{A}-\bnorm{C}$, 
{\small
 \begin{align}
  \t{b}_A-\t{b}_C 
   & = \left\lfloor \frac{\bnorm{A}}{2} \right\rfloor  - f(\bnorm{A})
         -\left\lfloor \frac{\bnorm{C}}{2} \right\rfloor + f(\bnorm{C})  = \frac{\bnorm{D}}{2} - (f(\bnorm{A})-f(\bnorm{C})) \nonumber \\
   & = t - t\delta (t + 2\ell - 1)\delta \geq (1-\delta)t   \label{eqn2}
 \end{align}
}
\noindent 
where the last line uses $\frac1\delta \geq \bnorm{A} \geq (2t + 2\ell
- 1)$ because $A\in \Od$. From Equations~(\ref{eqn1}) and
(\ref{eqn2}), and $\h{Y}_C=\h{\lambda}_C \t{b}_C,
\h{Y}_{A}=\h{\lambda}_{A} \t{b}_{A}$ we get: 
{\small
\begin{align}
 \h{\lambda} \t{b}_C & \geq  \h{\lambda}_C \t{b}_C = \h{Y}_C \geq \h{Y}_{A} - t = \h{\lambda}_A \t{b}_A -  t \geq (\h{\lambda} -\delta^x)\t{b}_A   - t \nonumber 
= \h{\lambda}\t{b}_A  - \delta^x \t{b}_A  -  t  \geq \h{\lambda}(\t{b}_C + (1-\delta)t) - \delta^x \t{b}_A  -  t  \nonumber\\
& > \h{\lambda}\t{b}_C + (1+3\delta)(1-\delta)t - \delta^{x}\t{b}_A  - t \nonumber \geq \h{\lambda}\t{b}_C + \delta t - \delta^{x}\t{b}_A \nonumber 
\end{align}
}
Since $\t{b}_A\leq 1/\delta$ this implies that $t < \delta^{x-1}\t{b}_A \leq \delta^{x-2}$ which 
contradicts $A_1\cap A_2 \neq \emptyset$ for $x\geq 2$.

\medskip\noindent{\bf Case II: $\bnorm{A_1\cap A_2}$ is odd.}
 Let $C=A_1\cup A_2$, and $D=A_1\cap A_2$. 
  Let $\bnorm{A_1}=\ell_1$,$\bnorm{A_2}=\ell_2$.
 We prove 
{\small
 \begin{equation}
 \h{\lambda}_C \leq \h{\lambda} \label{mainpoint}
 \end{equation}
 }
If $\bnorm{C} \leq 1/\delta$ then Equation~\eqref{mainpoint} is true by definition since $\h{\lambda}$ explicitly optimizes over $\Od$ and $C \in \Od$. We focus on the case $\bnorm{C} > 1/\delta$. We extend the definitions $\t{b}_C = \cvol{C} - f(\bnorm{C})$ and $\h{\lambda}_{C} = \h{Y}_C/\t{b}_C$ for all odd subsets with $\bnorm{\cdot} \leq 2/\delta$. Now $\h{Y}_C \leq \frac{\bnorm{C}}{2}$ since $\sum_j \h{y}_{ij} \leq b_i$. Note that $\bnorm{C}=\bnorm{A_1}+\bnorm{A_2}-\bnorm{D}$ and $\bnorm{D}\geq 1$.
Thus $\bnorm{C}\leq 2/\delta - 1 $ and using Condition ${\mathcal F}1$:
{\small \begin{align*}
 \t{b}_C \geq (1-\delta)\cvol{C} & =  (1-\delta)\frac{\bnorm{C}}{2} \left( 1 -  \frac1{\bnorm{C}}\right) \geq (1-\delta)^2\frac{\bnorm{C}}{2} 
\end{align*}
}
which implies that $\h{\lambda}_{C} \leq (1-\delta)^{-2} \leq 
1+3\delta < \h{\lambda}$. Thus Equation~\eqref{mainpoint} holds in this case as well.

Now, $\h{Y}_C + \h{Y}_D  \geq \h{Y}_{A_1} + \h{Y}_{A_2}$ and $\cvol{C} + \cvol{D} = \cvol{A_1} + \cvol{A_2}$.
Therefore:
{\small\begin{align}
\h{Y}_C + \h{Y}_D 
   & \geq \h{Y}_{A_1} + \h{Y}_{A_2} 
   = \h{\lambda}_{A_1} \t{b}_{A_1} + \h{\lambda}_{A_2} \t{b}_{A_2} \label{eqn3} \geq (\h{\lambda} - \delta^x)(\t{b}_{A_1}+\t{b}_{A_2}) 
\end{align}
}
\noindent If $\bnorm{D}=1$, then by Condition ${\mathcal F}3$: 
$\t{b}_C =  \t{b}_{A_1}+\t{b}_{A_2} - \frac{\delta^2}{4} (2\ell_1\ell_2 - 2\ell_1 -2\ell_2 + 1)$, and using Equation~\eqref{mainpoint},
{\small
\begin{align*}
\h{\lambda}\t{b}_C & \geq   \h{\lambda}_C \t{b}_C = \h{Y}_C 
\geq (\h{\lambda} - \delta^x)(\t{b}_{A_1}+\t{b}_{A_2}) 
\geq \h{\lambda}(\t{b}_{A_1}+\t{b}_{A_2}) - \delta^x(\t{b}_{A_1}+\t{b}_{A_2}) 
\nonumber \\
& \geq \h{\lambda}\t{b}_C + \frac{\delta^2 \h{\lambda}(2\ell_1\ell_2 - 2\ell_1 -2\ell_2 + 1)}{4} - \delta^x \left(\frac{\ell_1 + \ell_2 - 2}{2}\right)
\end{align*}
}
since $\t{b}_{A_1}+\t{b}_{A_2} \leq (\ell_1 + \ell_2 - 2)/2$. Therefore we would have a contradiction if 
{\small
\begin{equation}
\h{\lambda}(2\ell_1\ell_2 - 2\ell_1 -2\ell_2 + 1) - 2\delta^{x-2}(\ell_1 + \ell_2 - 2) > 0
\label{eqn4}
\end{equation}
}
\noindent 
Observe that for $x\geq 3$ the term $
2\delta^{x-2} (\ell_1 + \ell_2 - 2)$ is at most $2$ whereas 
$(2\ell_1\ell_2 - 2\ell_1 -2\ell_2 + 1) \geq 7$ since $3 \leq \ell_1,\ell_2 \leq \frac1\delta$. 
Since $\h{\lambda}>1$ we have a contradiction for $\bnorm{D}=1, x\geq 3$.

\medskip
\noindent Now consider $\bnorm{D}\geq 3$. Without loss of generality,
$\bnorm{A_2-D}\geq\bnorm{A_1-D}$. Let $ \bnorm{A_1-D} =2t$.
Using Condition ${\mathcal F}3$, $\t{b}_C + \t{b}_D \leq  \t{b}_{A_1} + \t{b}_{A_2} - 2t^2\delta^2$.
Note $\h{\lambda}_D \leq \h{\lambda}$, and from Equation~\eqref{mainpoint} $\h{\lambda}_C \leq \h{\lambda}$. Therefore
$\h{\lambda} \left( \t{b}_C + \t{b}_D \right) \geq \h{\lambda}_C \t{b}_C + \h{\lambda}_D \t{b}_D$ and 
from Equation~(\ref{eqn3}):
{\small
\begin{align}
\h{\lambda} \left( \t{b}_C + \t{b}_D \right) & \geq \h{\lambda}_C \t{b}_C + \h{\lambda}_D \t{b}_D =
  \h{Y}_C + \h{Y}_D 
 \geq \h{\lambda}(\t{b}_{A_1}+\t{b}_{A_2}) -  
\delta^x(\t{b}_{A_1}+\t{b}_{A_2})
\nonumber \\
& \geq \h{\lambda}(\t{b}_C + \t{b}_D)  + 2t^2\delta^2\h{\lambda} - \delta^x(\t{b}_{A_1}+\t{b}_{A_2})
\label{eqn5} 
\end{align}
}
\noindent Again, this is infeasible if $x\geq 3$ since
$\t{b}_{A_1}+\t{b}_{A_2}\leq 2/\delta$ and $\h{\lambda} \geq 1$.
Therefore for $x\geq 3$, in all cases we arrived at a contradiction to
$A_1 \cap A_2 \neq \emptyset$.  Thus we have proved that
$\{U:\hat{\lambda}_U \geq \hat{\lambda}-\delta^3; U \in \Od\}$ is a
laminar family. 

We now prove the second part.
Consider $L'_{\ell}=\{U:\hat{\lambda}_U \geq
\hat{\lambda}-\delta^x; U \in \Od; \bnorm{U}=\ell\}$. From {\bf Case
I}, no two distinct sets $A_1, A_2 \in L'_\ell$ intersect when
$\bnorm{A_1 \cap A_2}$ is even. From {\bf Case II} for $\ell \geq 5$,
they cannot have $\bnorm{D}=1$ because $(2\ell^2 - 4\ell + 1) -
2(2\ell - 2) > 0$ for $\ell \geq 5$. Note $\bnorm{A_1 - D}=\bnorm{A_2
- D}$ because $\bnorm{A_1}=\bnorm{A_2}= \ell$.  Moreover for $t \geq
\delta^{(x-3)/2}$ we would have $2t^2\delta^2\h{\lambda} >
\delta^x(\t{b}_{A_1}+\t{b}_{A_2})$ in Equation~\ref{eqn5}. Therefore
two distinct sets $A_1, A_2 \in L'_\ell$ which intersect, cannot
differ by $\delta^{(x-3)/2}$ or more elements. This means that
$|L'_\ell| \leq \left(n/\delta\right)^{1+\delta^{(x-3)/2}}$ for
$\ell\geq 5$ --- to see this choose a maximal collection of disjoint
sets in $L'_\ell$. This would be at most $n$. Every other set $S$ in $L'_{\ell}$ has to intersect one of these sets in the maximal collection. To upper bound the number of such sets $S$ with intersection $t$, we can start from a set in that maximal collection; throw out $t$ elements in $\ell^{t}$ ways and include new elements in $n^{t}$ ways. Note $\ell \leq 1/\delta$. Thus the number of such sets for a fixed $t$ is $n(n/\delta)^{\delta^{(x-3)/2}}$.  Observe that $t\leq 1/\delta$ and the bound follows.
Finally note $|L'_{3}| \leq n^3$.
Thus the total number of sets is $n^3 + (n/\delta)^{1+\delta^{(x-3)/2}}$.
The lemma follows.
\end{proof}

\section{Proof of Theorem~\ref{thm:main2}}
\label{proof:main2} 

\paragraph{An Overview.} We combine the insights of the minimum odd-cut approach
~\cite{PadbergR82} along with the fact that $L_2
\subseteq L_1 $ is a laminar family as proved in
Theorem~\ref{thm:main1}. The algorithm picks the sets based on their
sizes.  Define $L_1(\ell) = \{ U | U \in L_1, \bnorm{U}=\ell\}$ and
$L_2(\ell) = \{ U | U \in L_2, \bnorm{U}=\ell\}$ for $\ell \in
[3,1/\delta]$. Note that $L_1(\ell) \supseteq L_2(\ell)$.  Observe
that it suffices to identify $L_2(\ell)$ for different values of
$\ell$. We construct an unweighted graph $G_\varphi(\ell,\t{\lambda})$
where $\varphi=O(\delta^{-4})$ with a new special node $r(\ell)$ with
the following two properties:

\begin{keyproperty}\label{cool1}
If $\t{\lambda} - \frac{\delta^3}{100} < \h{\lambda} \leq
\t{\lambda}$, then (i) all sets in $L_2(\ell)$ have a cut which is of size at most $\kappa(\ell)$ 
and (ii) all odd sets of $G_\varphi(\ell,\t{\lambda})$ which do not contain $r(\ell)$ 
and have cut of size at most $\kappa(\ell)$ belong to $L_1(\ell)$.
Here $\kappa(\ell)= \lfloor \varphi \t{\lambda}(1-\delta^2\ell^2/2) \rfloor + \frac{12\ell}{\delta} + 1 < 2\varphi$.
\end{keyproperty}

\begin{keyproperty}\label{cool2}
  We show in Lemma~\ref{maddsop} that we can extend the algorithm in
  \cite{PadbergR82} to efficiently 
  extract a collection $\bar{L}(\ell)$ of {\em maximal} odd-sets
  in $G_\varphi(\ell,\t{\lambda})$, not containing $r(\ell)$ and cut of size at most
  $\kappa(\ell)$ -- such that any such set which is not chosen must
  intersect with some set in the collection.
\end{keyproperty}

\begin{lemma}\label{maddsop}
  Given an unweighted graph $G$ with parameter $\kappa(\ell)$ 
  and a special node $r(\ell)$, in time $O(n\poly(\kappa,\log n))$
we can identify a collection $\bar{L}(\ell)$ of 
  odd-sets such that (i) each $U \in \bar{L}(\ell)$ does not contain $r(\ell)$  (ii) each $U \in \bar{L}(\ell)$ defines a cut of size at most $\kappa$ in $G$
  and (iii) every other odd set not containing $r(\ell)$ and with a cut less
  than $\kappa(\ell)$ intersects with a set in $\bar{L}(\ell)$.
\end{lemma}

The second property follows without much difficulty from the
properties of Gomory-Hu trees \cite{GomoryHu,Gusfield} -- trees which
represent all pairwise mincuts over a set of nodes. Observe that
property \ref{cool1} implies that we can restrict our attention to
only those regions of the graph $G_\varphi(\ell,\t{\lambda})$ which
have cuts of size at most $O(\delta^{-4})$. Therefore if we are given
a subset of vertices such that any partition of that vertex set
induces a large cut, then either that subset is included entirely
within one odd set or excluded completely.  This is the notion of a
Steiner Mincut which is used to compute the ``partial Gomory Hu tree''
-- where for some $\kappa$ we represent all pairwise min cuts of value
at most $\kappa$. Such representations can be computed for unweighted
undirected graphs in time $O(m + n\kappa^3\log^2 n)$ \cite{lowstcuts}
(see also improvements in \cite{BhalgatHKP07}).  The graph
$G_\varphi(\ell,\t{\lambda})$ is used exactly for this purpose.  If we
have a maximal collection $\bar{L}(\ell)$ then $\bar{L}(\ell)
\subseteq L_1(\ell)$ by condition (ii) of Property~\ref{cool1}.  Due
to Theorem~\ref{thm:main1}, the intersection of two such sets $U_1,U_2
\in L_1(\ell)$ will be either empty or have $\bnorm{\cdot}=\ell$ by laminarity --
the latter implies $U_1=U_2$. Therefore the sets in $L_1(\ell)$ are
disjoint. Any $U \in L_2(\ell) - \bar{L}(\ell)$ has a cut of size at most
$\kappa(\ell)$ using condition (i) of Property~\ref{cool1} and therefore
must intersect with some set in $\bar{L}(\ell)$. This is
impossible because $U \in L_2(\ell)$ implies $U \in L_1(\ell)$ and $\bar{L}(\ell)
\subseteq L_1(\ell)$ and we just argued that the sets in $L_1(\ell)$ are
disjoint! Therefore no such $U$ exists and $L_2(\ell) \subseteq \bar{L}(\ell)$.
We now have a complete algorithm: we perform a
binary search over the estimate $\t{\lambda} \in [1+3\delta,\frac32 +
\delta^2]$, and we can decide if there exists a set $U \in L_2(\ell)$
in time $O(n \poly(\delta^{-1},\log n))$ as we vary $\ell,\t{\lambda}$.
This gives us $\t{\lambda}$. We now find the collection $\bar{L}(\ell)$ for
each $\ell$ and compute all $\h{\lambda}_{U}$ exactly (either remembering 
the $\h{y}_{ij}$ of the the edges stored in $G_\varphi$ or by another 
pass over $G$). We can now return $\cup_\ell L_2(\ell)$.
We now prove Lemma~\ref{maddsop}.

\subsection{Proof Of Lemma~\ref{maddsop}}
 The parameter $\ell$ is not relevant to the proof and is dropped.
Algorithm~\ref{alg:exactviolated} provides the algorithm for this lemma.

\begin{nlemma}{\ref{maddsop}}  
 Given an unweighted graph $G$ with parameter $\kappa$ 
  and a special node $r$, in time $O(n\poly(\kappa,\log n))$
we can identify a collection $\bar{L}$ of 
  odd-sets such that (i) each $U \in \bar{L}$ does not contain $r$  (ii) each $U \in \bar{L}$ defines a cut of 
size at most $\kappa$ in $G$
  and (iii) every other odd set not containing $r$ and with a cut less
  than $\kappa$ intersects with a set in $\bar{L}$.
\end{nlemma}

\begin{algorithm}[H]
{\small
 \begin{algorithmic}[1]
  \STATE $\bar{L}\leftarrow\emptyset$. Initially $G'=G$. The node $r \in V(G)$. 
  \REPEAT
   \STATE 
    Assign the $r$ duplicity $b_r=1$ if $\sum_{i \in V(G')} b_i$ is odd. Otherwise let $b_r=2$.
   \STATE Construct a tree $\T$ that represents {\bf all low $s$--$t$ cuts} in $G'$
using Theorem~\ref{low}. The nodes of this tree $\T$ correspond to
   subsets of vertices of $V(G')$.  
   \STATE Make the vertex set
   containing $r$ the root of $\T$ and {\bf orient all edges towards the
   root}. The oriented edges represent an edge from a child to a
   parent. Let $D(e)$ indicate the set of descendant subsets of an edge $e$ (including the
    child subset which is the tail of the edge, but not including the
    parent subset which is the head of the edge).  
   \STATE Using  dynamic programming starting at the leaf, mark every edge as admissible/inadmissible based on the
    $\sum_{S \in D(e)} \sum_{i \in S} b_i$ over the descendant subsets of that edge being odd/even respectively. 
   \STATE Starting from the root $s$ downwards, pick the edges $e$ in parallel 
  such that (c1) the weight of $e$ (corresponding to a cut) is at most $\kappa$, 
  (c2) $\sum_{S \in D(e)} \sum_{i \in S} b_i$ is odd and (c3) no edge $e'$ on the path from $e$ to $r$ satisfies (c1) and (c2).
Let the odd-set $U_e$ corresponding to this edge $e \in \T$ be $U_e= \cup_{S \in D(e)} S$. 
   \STATE If the odd-sets found are $U_{e_1},\ldots,U_{e_f}$ then $\bar{L} \leftarrow \bar{L} \cup \{ U_{e_1},\ldots,U_{e_f} \}$. Observe that the sets $U_{e_g}$ are mutually disjoint for $1\leq g \leq f$ and do not contain $r$.
   \STATE Merge all vertices in $\bigcup^f_{g=1} U_{e_g}$ with $r$. Observe that for any set $U$ that does not contain $r$ and does not intersect with any $ U_{e_g}$, the 
cut $\Cut(U)$ is unchanged. This defines the new $G'$.
  \UNTIL{no new odd set has been found in $G'$}
  \RETURN $\bar{L}$.
 \end{algorithmic}
 \caption{Algorithm: Finding a maximal collection of odd-sets
   \label{alg:exactviolated}} }
\end{algorithm}

\begin{proof} 
First, consider the following known theorem and Lemma:

\begin{theorem}[\cite{BhalgatHKP07,lowstcuts}] 
\label{low}
  Given a graph with $n$ nodes and $m$ edges (possibly with parallel
  edges), in time $O(m) + \t{O}(n\kappa^2)$ we can construct a
  weighted tree $T$ that represents all min $s$--$t$ cuts in $G'$ of
  value at most $\kappa$. The nodes of this tree are subsets of
  vertices.  The mincut of any pair of vertices that belong to the
  same subset (the same node in the tree $T$) is larger than $\kappa$
  and for any pair of vertices $i,j$ belonging to different subsets
  (nodes in the tree $T$) the mincut is specified by the partition
  corresponding to the least weighted edge in the tree $T$ between the
  two nodes that contain $i$ and $j$ respectively.
\end{theorem}

\begin{lemma}[Implicit in \cite{PadbergR82}] 
\label{lem:gomoryhuproperty} Suppose that for a graph $G=(V,E)$, $\sum_{i\in V} b_i$ is even.
 For any odd-set $U$ in $G$ with cut $\kappa$, there exists an edge $e$ in
 the low min $s$-$t$ cut tree $\T$ such that removing $e$ from the tree results in two connected components of odd sizes and the component $U_e$ not containing the root intersects $U$.
 In addition, the cut between $U_e$ and rest of the graph is at most $\kappa$.
\end{lemma}

\begin{proof}(Of Lemma~\ref{lem:gomoryhuproperty})
 Observe that the min $u$-$v$ cut for any $u \in U$ and $v \not \in U$ is at most $\kappa$. 
 
 We provide an algorithmic proof of the existence -- this is not the algorithm to find the odd sets.
 Let $H_0=V(G)$. We will maintain the three invariants that (1) $\bnorm{H_z}$ is even (2) $H(z)$ defines a connected component in the low min $s$-$t$ cut tree $\T$ and (3) $H_z \cap U \neq \emptyset$ and $H_z \cap (V(G) - U) \neq \emptyset$. These hold for $H_0$. Staring from $H_z$ until we find a desired edge $e$ or find $H_{z+1} \subset H_z$ which satisfies the same invariants. This process has to stop eventually and we would have found the desired edge $e$.
 
 Given the invariant, there exists $u \in H_z \cap U $ and $v \in 
 H_z \cap (V(G) - U)$ such that the min $u$-$v$ cut is at most $\kappa$ and therefore there must exist an edge $e_z$ (corresponding to a min $u$--$v$ cut) within the component $H_z$ such that $e_z$ separates $u,v$. Let the two connected sub-components of $H_z$ defined by the removal of $e_z$ be $S_1$ and $S_2$. If $\bnorm{S_1},\bnorm{S_2}$ are both even, then one of them must satisfy condition (3), since $\bnorm{U}$ is odd. 
 
 This process has to stop eventually and we would have found the desired edge $e$. Observe that all the subcomponents of $\T$ created in this manner define even sets until we find $e$. If we add back all the sub-components such that we have the two components corresponding to the two sides of $e$, both of those components must have odd $\bnorm{\cdot}$. The component not containing $r$ defines $U_e$. In addition, the corresponding cut size is less than $\kappa$.
\end{proof}

\medskip
\noindent (Continuing with Proof of Lemma~\ref{maddsop}.)
All that remains to be proven is that the loop in 
Algorithm~\ref{alg:exactviolated} needs to be run only a few times.
Suppose after $t'$ repetitions $Q_{t'}$ is the maximum collection of disjoint 
odd-sets which are attached to the remainder of $\T$ with cuts of size at most $\kappa$ and we choose $U_{e_1},\ldots,U_{e_f}$ to be
added to $\L$ in the $t'+1$st iteration.
We first claim that $|Q_{t'+1}| \leq f$. To see this we first map
every odd-set in $Q_{t'+1}$ to an edge in the tree as specified by the
existence proof in Lemma~\ref{lem:gomoryhuproperty}. This map need not
be constructive -- the map is only used for this proof. Note that $\sum_i b_i$ is even, by construction, in Algorithm~\ref{alg:exactviolated} as required in Lemma~\ref{lem:gomoryhuproperty}. Observe that
this can be a many to one map; i.e., several sets mapping to the same edge.

Now every edges $e_1,\ldots,e_f$ chosen in
Algorithm~\ref{alg:exactviolated} satisfy the property for all $j$: no
edge $e'$ on the path from the head of $e_{j}$ (recall that the edges
are oriented towards the root $r$) to $r$ is one of the edges in our
map. Because in that case we would have chosen that edge $e'$ instead
of $e_j$.

Therefore the sets in $Q_{t'+1}$ could not have mapped to any edges in
the path towards $r$. Now, if a set in $Q_{t'+1}$ mapped to an edge $e'$ which is
a descendant of the tail of some $e_j$ (again, the edges are oriented
towards $r$) then this set intersects with our chosen $U_{e_j}$ which is not
possible.

{\em Therefore any set in $Q_{t'+1}$ must have mapped to
  the same edges in the tree; i.e., $e_1,\ldots,e_f$.} 
But then the vertex at the head of the
edge belongs to the set in $Q_{t'+1}$. Therefore there can be at most $f$
such sets. This proves $|Q_{t'+1}| \leq f$.

We next claim that $|Q_{t'+1}|\leq |Q_{t'}|-f$. Consider 
$Q' = Q_{t'+1}\cup\{U_{e_1},\ldots,U_{e_f}\}$. $Q'$ is a collection of
disjoint odd-sets which define a cut of size $\kappa$ in $G$ after $t'$ 
repetitions.
Obviously $|Q'|=|Q_{t'+1}|+f$ and by the definition of $Q_{t'}$, 
$|Q'|\leq |Q_{t'}|$. Therefore, $|Q_{t'+1}|\leq |Q_{t'}|-f$.

Therefore, in the worst case, $|Q_{t'}|$ decreases by a factor $1/2$
and therefore in $O(\log n)$ iterations over this loop we would
eliminate all odd-sets that define a cut of size $\kappa$
in $G'$.
\end{proof}

\subsection{Proof of Theorem~\ref{thm:main2}} 

\begin{ntheorem}{\ref{thm:main2}} 
For a graph $G$ with $n$ vertices and any non-negative edge weights 
$\h{y}_{ij} = \h{y}_{ji}$ such that $\h{y}_{ii}=0$ and $\sum_j \h{y}_{ij} \leq b_i$ 
for all $i$; and $\delta\in (0,\frac1{16}]$, define: 
$\h{\lambda}_U = \frac{\sum_{(i,j): i,j\in U} \h{y}_{ij}}{\t{b}_U}$ where $\t{b}_U = \cvol{U} - \frac{\delta^2 \bnorm{U}^2}{4}$ and $\h{\lambda} = \max_{U \in \Od} \h{\lambda}_U$.
If $\hat{\lambda}\geq 1+3\delta$ we can find the 
set $L_2=\{U:\hat{\lambda}_U \geq \hat{\lambda}-\delta^3/10; U \in \Od\}$ 
in
$O(m' + n \poly (\delta^{-1},\log n))$
time using $O(n \delta^{-5})$ space where $m'=|\{ (i,j) | \h{y}_{ij} > 0 \}|$.
\end{ntheorem}

\begin{proof} 
  We first observe that $L_2$ is a laminar family using
  Theorem~\ref{thm:main1} and $L_2 \subseteq L_1$.
   Second, observe that for any $U$ we have $\sum_{(i,j):i,j \in U} \h{y}_{ij} \leq \frac12 \sum_{i \in U} \sum_j \h{y}_{ij} \leq \frac12 \sum_{i \in U} b_i = \bnorm{U}/2$. Therefore $\h{\lambda} \leq \frac32/(1-\frac{\delta^2}{4}) < \frac32 + \delta^2$.
 We maintain an estimate $\t{\lambda}$ of
such that $\t{\lambda} - \frac{\delta^3}{100} < \h{\lambda}
\leq \t{\lambda} \leq \frac32 + \delta^2 
$. This estimate can be found using binary search 
(as described below) We now show how to find the sets $U \in L_2$ with $\bnorm{U}=\ell$, denoted by $L_2(\ell)$.

\smallskip\noindent
Create a graph $G_\varphi$ with $p_{ij} = \lfloor
\varphi \h{y}_{ij} \rfloor$ parallel edges between $i$ and $j$ where
$\varphi=50/\delta^{4}$ (this parameter can be optimized but we omit
that in the interest of simplicity). This is an unweighted graph. This
graph can be constructed in a single pass over $\{ (i,j) \}$.  We also
``merge'' all pairs of vertices $i$ and $j$ if $p_{ij}$ exceeds $2\varphi$.
Moreover delete vertices $i$ with $2\varphi/\delta$ edges -- note that these 
vertices must have $b_i \geq \sum_j \h{y}_{ij} > 1/\delta$ and cannot participate in any odd set 
in $\Od$.
This gives us a graph $G_\varphi$ with at most $O(n\delta^{-5})$
edges.

Now for an odd $\ell \in [3, 1/\delta]$ and $\t{\lambda}$,
create $G_\varphi(\ell,\t{\lambda})$ as follows: We begin with $G_\varphi$. Let $q_i(\ell) = \lfloor
\varphi \t{\lambda}(1-\delta^2\ell)b_i \rfloor$ for all $i$. Since
$q_i(\ell) > (1+ \delta) \varphi b_i > \sum_j p_{ij}$ (because $\t{\lambda}$
is large) we can add a new node $r(\ell)$ and add $q_i(\ell) - \sum_j p_{ij}$
edges between $r(\ell)$ and $i$ (for all $i$).  This gives us a graph
$G_\varphi(\ell,\t{\lambda})$ of size $O(n\delta^{-5})$ edges for all
$\ell$.  Let $\kappa(\ell) = \lfloor \varphi
\t{\lambda}(1-\delta^2\ell^2/2) \rfloor + \frac{12 \ell}{\delta} + 1 <
2\varphi$.  Now: 
{\small
\[
q_i(\ell) - \kappa(\ell)  \geq 
\varphi \t{\lambda} (1-\delta^2\ell) - 1 - \varphi \t{\lambda} (1-\delta^2\ell^2/2)
- \frac{12 \ell}{\delta} - 1  
 = \frac{\varphi\t{\lambda}\delta^2\ell(\ell-2)}{2} - \frac{12 \ell}{\delta} - 2
\]
}
which is positive for $\varphi=50/\delta^4$ and $\ell \geq 3$.
Therefore $q_i(\ell)  >  \kappa(\ell)$.

Define $\Cut(U)$ to be the cut induced by $U$ in $G_{\varphi}(\ell,\t{\lambda})$, that is, $\Cut(U) =    \sum_{i \in U} q_i - 2 \sum_{(i,j):i,j \in U} p_{ij}$.
We now show that for $\bnorm{U}>1/\delta$, $\Cut(U) > \kappa(\ell)$.
For any odd set $U \in \O$ with $\bnorm{U}>1/\delta$:
{\small
\begin{eqnarray*}
\Cut(U) - \kappa(\ell)  & =  &   \sum_{i \in U} q_i (\ell) - 2 \sum_{(i,j):i,j \in U} p_{ij} - \kappa(\ell)  \\
& \geq  & \sum_{i \in U} ( \varphi \t{\lambda}(1-\delta^2\ell)b_i - 1 ) - 2 \varphi 
\displaystyle \sum_{(i,j):i,j \in U} \h{y}_{ij} - \varphi \t{\lambda}(1-\delta^2\ell^2/2)
- \frac{12 \ell}{\delta} - 1  \\
& \geq  & \varphi \t{\lambda}(1-\delta^2\ell) \bnorm{U} - |U| - \varphi \bnorm{U} - \varphi \t{\lambda}(1-\delta^2\ell^2/2) - \frac{12\ell}{\delta} - 1  \quad \mbox{(Since $\sum_{(i,j):i,j \in U}
\h{y}_{ij} \leq \bnorm{U}/2$)} \\
& \geq & \varphi \t{\lambda}(1-\delta) \bnorm{U}  
- \varphi \bnorm{U} - \varphi \t{\lambda} - \delta^2 \varphi \bnorm{U} \quad \mbox{(Since $\ell \leq 1/\delta$ and $\delta^2 \varphi \bnorm{U} > |U| +  \frac{12\ell}{\delta} + 1$)}\\
& = & \varphi \left( \t{\lambda}(1-\delta) \bnorm{U} - \t{\lambda} - (1+\delta^2) \bnorm{U} \right) \\
& \geq & \varphi \left( (1+3\delta)(1-\delta) \bnorm{U} - \frac32 - \delta^2 - (1+\delta^2) \bnorm{U} \right) \quad \mbox{ (Since $1+3\delta \leq \h{\lambda} \leq \t{\lambda} \leq \frac32 + \delta^2$)} \\
& > & \varphi \left( 2\delta(1-2\delta) \bnorm{U} -  \frac32 - \delta^2 \right) 
 >  \varphi \left( 2(1-2\delta) -  \frac32 - \delta^2 \right) > 0 \quad \mbox{ (Since $\delta \bnorm{U} > 1$)}
\end{eqnarray*} 
}
where the last inequality follows $\delta \in (0,\frac{1}{16}]$.  
 Therefore no odd-set with $\bnorm{U} > 1/\delta$ satisfies $\Cut(U) \leq \kappa(\ell)$.

\smallskip
We now show Property~\ref{cool1}, namely:
If $\t{\lambda} - \frac{\delta^3}{100} < \h{\lambda} \leq
\t{\lambda}$, then (i) all sets in $L_2(\ell)$ have a cut which is at most $\kappa(\ell)$ 
and (ii) all odd sets of $G_\varphi(\ell,\t{\lambda})$ which do not contain $s$ 
and have cut at most $\kappa(\ell)$ belong to $L_1(\ell)$.
For part (i) for a set $U\in L_2(\ell)$ with $\bnorm{U}=\ell$, note $|U| \leq \bnorm{U} = \ell$ and:
{\small
\begin{eqnarray*}
\Cut(U)  & = & \sum_{i \in U} q_i - 2 \sum_{(i,j):i,j \in U} p_{ij}  \leq \sum_{i \in U} \varphi \t{\lambda}(1-\delta^2\ell)b_i - 2 \varphi 
\displaystyle \sum_{(i,j):i,j \in U} \h{y}_{ij} + |U|^2 \nonumber\\
& \leq  & \varphi \t{\lambda}(1-\delta^2\ell) \bnorm{U} - 2 \varphi \h{\lambda}_{U}\t{b}_U + \ell^{2} 
\leq   \varphi \t{\lambda}(1-\delta^2\ell) \bnorm{U} - 2 \varphi \left(\t{\lambda} - \frac{\delta^3}{100} - \frac{\delta^3}{10}\right)\t{b}_U + \ell^{2} \nonumber\\
& = & \varphi \t{\lambda}(1-\delta^2\ell^2/2) + \frac{11\delta^3\varphi\t{b}_U}{50} 
+ \ell^{2} 
 =  \varphi \t{\lambda}(1-\delta^2\ell^2/2) + \frac{11\t{b}_U}{\delta} 
+ \ell^{2} \\
& \leq & \varphi \t{\lambda}(1-\delta^2\ell^2/2) + \frac{12\ell}{\delta}
\leq \kappa(\ell) \quad \mbox{(since $\t{b}_U < \bnorm{U} = \ell \leq 1/\delta$)}
\end{eqnarray*}
} 
\noindent To prove part (ii) if $\Cut(U)\leq
  \kappa(\ell) $ then:
{\small
\begin{eqnarray*}  
\sum_{(i,j):i,j \in U} p_{ij}
& = & \frac12 \left( \sum_{i \in U} q_i - \Cut(U') \right) \geq  
\frac12 \left( \sum_{i \in U} \left( \varphi \t{\lambda}(1-\delta^2\ell)b_i - 1 \right) -  \kappa(\ell) \right) \nonumber \\
&\geq &
\frac12 \left( \sum_{i \in U} \left( \varphi \t{\lambda}(1-\delta^2\ell)b_i - 1 \right) -  \varphi \t{\lambda}(1-\delta^2\ell^2/2) \right) - \frac{12\ell}{\delta} - 1
\nonumber \\
& \geq & \varphi \t{\lambda} \left( \cvol{U} - \frac{\delta^2\bnorm{U}^2}{4} \right) + 
 \frac{\varphi \t{\lambda} \delta^2}{4}  \left( \bnorm{U} -\ell \right)^2 - \frac{|U|}{2} -  \frac{12\ell}{\delta} - 1 \nonumber \\
& = &\varphi \t{\lambda} \t{b}_U +  \frac{\varphi \t{\lambda} \delta^2}{4}  \left( \bnorm{U} -\ell \right)^2 - \frac{|U|}{2} -  \frac{12\ell}{\delta} - 1
\nonumber
\end{eqnarray*}
}
But since $\t{\lambda} \geq \h{\lambda} \geq \h{\lambda}_U$ and $\varphi \t{b}_U \h{\lambda}_U = \varphi \sum_{(i,j):i,j \in U} \h{y}_{ij} \geq 
\sum_{(i,j):i,j \in U} p_{ij}$ we have
\begin{equation}\label{goodcut2} 
 \varphi \t{\lambda} \t{b}_U \geq \varphi \h{\lambda}_U \t{b}_U \geq  \varphi \t{\lambda} \t{b}_U +  \frac{\varphi \t{\lambda} \delta^2}{4}  \left( \bnorm{U} -\ell \right)^2 - \frac{|U|}{2} -  \frac{12 \ell}{\delta} - 1 
\end{equation}
But that is a contradiction unless $\bnorm{U}=\ell$, otherwise the quadratic term, 
$\frac{\varphi \t{\lambda} \delta^2}{4}  \left( \bnorm{U} -\ell \right)^2 \geq 12.5 \delta^{-2}$ is larger than the negative terms which are at most 
$\frac1{2\delta} + \frac{12}{\delta^2} + 1$ in the RHS of 
Equation~\ref{goodcut2}. Therefore 
$\Cut(U) \leq \kappa(\ell)$ for an odd-set implies $\bnorm{U}=\ell$.
But then Equation~\ref{goodcut2} implies (again using $|U| \leq \bnorm{U} = \ell$):
\begin{eqnarray*}
\varphi \h{\lambda}_U \t{b}_U 
\geq \varphi \t{\lambda} \t{b}_U - \frac{\ell}{2} -  \frac{12 \ell}{\delta} - 1
\end{eqnarray*}
Now $\t{b}_U \geq \frac{\ell}{3}(1-\frac{3\delta}{4})$ when $\bnorm{U}=\ell \geq 3$; thus: 
\begin{eqnarray*}
\h{\lambda}_U \geq \t{\lambda} - \frac{\ell}{2\varphi\t{b}_U } -  \frac{12 \ell}{\delta\varphi\t{b}_U} - \frac{1}{\varphi\t{b}_U} \geq \t{\lambda} - \frac{3\delta^4}{100(1-\frac{3\delta}{4})} - \frac{36\delta^3}{50(1-\frac{3\delta}{4})} - \frac{\delta^4}{50} >  \t{\lambda} - \delta^3 \geq \h{\lambda} - \delta^3
\end{eqnarray*}
in other words, $\Cut(U) \leq \kappa(\ell)$ for an odd-set implies $U \in L_1(\ell)$, as claimed in part(ii).

\smallskip\noindent
We now apply Lemma~\ref{maddsop} to extract a collection $\bar{L}(\ell)$ of odd-sets in $G_\varphi(\ell,\t{\lambda})$, not containing $r(\ell)$ and cut at most
$\kappa(\ell)$ -- such that any such set which is not chosen must
intersect with some set in the collection $\bar{L}(\ell)$.

If we have a maximal collection $\bar{L}(\ell)$ then $\bar{L}(\ell)
\subseteq L_1(\ell)$ by part (ii) of Property~\ref{cool1}.  Due
to Theorem~\ref{thm:main1}, the intersection of two such sets $U_1,U_2
\in L_1(\ell)$ will be either empty or of size $\ell$ by laminarity --
the latter implies $U_1=U_2$. Therefore the sets in $L_1(\ell)$ are
disjoint. Any $U \in L_2(\ell) - \bar{L}(\ell)$ has a cut of size at most
$\kappa(\ell)$ using part (i) of Property~\ref{cool1} and therefore
must intersect with some set in $\bar{L}(\ell)$. This is
impossible because $U \in L_2(\ell)$ implies $U \in L_1(\ell)$ and $\bar{L}(\ell)
\subseteq L_1(\ell)$ and we just argued that the sets in $L_1(\ell)$ are
disjoint. Therefore no such $U$ exists and $L_2(\ell) \subseteq \bar{L}(\ell)$.
We now have a complete algorithm: we perform a
binary search over the estimate $\t{\lambda} \in [1+3\delta,\frac32 +
\delta^2]$, and we can decide if there exists a set $U \in L_2(\ell)$
in time $O(n \poly(\delta^{-1},\log n))$ as we vary $\ell,\t{\lambda}$.
This gives us $\t{\lambda}$. We now find the collections $\bar{L}(\ell)$ for 
each $\ell$ and compute all $\h{\lambda}_{U}$ exactly. We can now return $\cup_\ell L_2(\ell)$.
Observe that $G_\varphi$ does not
  need to be constructed more than once; it can be stored and reused.
 The running time follows
from simple counting.
\end{proof}

\section{Rounding Uncapacitated $b$-matchings}
\label{sec:rounding}

\begin{ntheorem}{\ref{thm:rounding}}{\rm (Integral $b$--matching)}
 Given a fractional $b$-matching $\by$ for a non-bipartite graph
  which satisfies $\lpb$ (parametrized over $\bb$) where
  $|\{(i,j)|y_{ij} > 0 \}|=m'$, we find an integral $b$--matching
  of weight at least $(1-2\delta)\sum_{(i,j)} w_{ij}y_{ij}$ in
  $O(m'\delta^{-3}\log (1/\delta))$ time and $O(m'/\delta^2)$ space.
\end{ntheorem}\par

\begin{algorithm}[H]
{\small
\begin{algorithmic}[1]
\STATE First Phase: {\bf (large multiplicities)} Let $t=\lceil 2/\delta \rceil$ and $\mzero=\emptyset$.
\begin{enumerate}[(a)]\parskip=0in
\item If $y_{ij}\geq t$ add $\h{y}^{(0)}_{ij}=\lfloor y_{ij} \rfloor - 1$ copies of $(i,j)$ to $\mzero$.
 \item Set $y^{(1)}_{ij} =0$ if $y_{ij}\geq t$ and $y^{(1)}_{ij} =y_{ij}$ otherwise.
\item Let $b^{(1)}_i = \min \left \{ b_i - \sum_j \h{y}^{(0)}_{ij}, \lceil \sum_j
  \yone_{ij} \rceil + 1 \right \}$.
\end{enumerate}
\STATE Second Phase: {\bf (large capacities)}
\begin{enumerate}[(a)]\parskip=0in
\item {\bf While} {$\exists i$ s.t. $\sum_j \yone_{ij} \geq 3t$} {\bf do}
\begin{enumerate}[(i)]\parskip=0in
\item Order the vertices adjacent to $i$ arbitrarily. Select the
prefix $S$ in that order such that the sum is between $t$ and $2t$
(each edge is at most $t$ from Step 1b). Create a new copy $i'$ of $i$
with this prefix and $\yone_{i'j} = \yone_{ij}$ for $j \in S$ and delete the edges from $S$ incident to $i$.
Observe that the procedure describes a process where given a set of numbers $q_1,\ldots,q_k$ such that each $q_j \leq 1$
and $\sum_j q_j = Y \geq 3$; we partition the set of numbers such
that each partition $S$ satisfies $1 \leq \sum_{j \in S} q_j \leq 2$.
\end{enumerate}
\item If no copies of $i$ were created then $\btwo_i=\bone_i$. For every new $i'$ (corresponding to $i$) created from the partition $S$ (which may have now become $S'$ with subsequent splits), assign $\btwo_{i'} = \lfloor \sum_{j \in S'} \yone_{ij} \rfloor$. Note $\btwo_i \leq 3t$ for all vertices.
We now have a vertex set $V^{(2)}$.
Set $\ytwo_{ij}= (1-\delta) \yone_{ij}$ for $i,j \in V^{(2)}$.
\end{enumerate}
\STATE Third Phase: {\bf Reduction to  weighted matching.}
\begin{enumerate}[(a)]\parskip=0in
 \item For each $i \in V^{(2)}$ with $b^{(2)}_i$, create $i(1), i(2), \cdots, i(b^{(2)}_i)$.
 \item For each edge $(i,j)$, create a complete bipartite graph between
  $i(1), i(2), \cdots$ and $j(1), j(2), \cdots$ with every edge having weight $w_{ij}$. Let this new graph be $G^{(3)}$.\label{ex-step-1}

\item Run any fast approximation for finding a $(1-\epsilon)$-approximate maximum weighted matching in $G^{(3)}$ let this matching be $\mthree$. Matching $\mthree$ provides a $b$--matching $\mtwo$ in $G^{(2)}$ of same weight (merge edges).
(ii) Matching $\mtwo$ provides a $b$--matching $\mone$ in $G^{(1)}$ of same weight (merge vertices).
\end{enumerate}
\STATE {\bf Output:} $\mzero \cup \mone$.
\end{algorithmic}
}
\caption{Rounding a fractional $b$--matching}
\label{alg:round1}
\end{algorithm}

 As an example of Step~3(b), consider

 
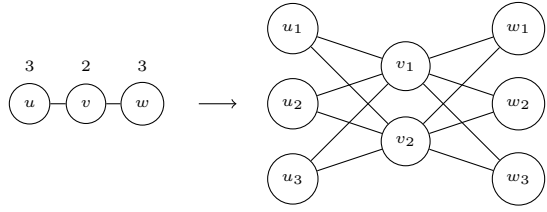
\begin{center}
 \begin{tikzpicture}[font=\tiny]
 \tikzstyle{vertex}=[draw,shape=circle,minimum size=0.3cm]
 \node[vertex] at (0,0) (u) {$u$};
 \node[above,yshift=0.3cm] at (u) {3};
 \node[vertex] at (0.75,0) (v) {$v$};
 \node[above,yshift=0.3cm] at (v) {2};
 \node[vertex] at (1.5,0) (w) {$w$};
 \node[above,yshift=0.3cm] at (w) {3};
 \draw (u) -- (v);
 \draw (v) -- (w);
 \draw[->] (2.25,0) -- (2.75,0);
 \node[vertex] at (3.5,1)  (u1) {$u_1$};
 \node[vertex] at (3.5,0)  (u2) {$u_2$};
 \node[vertex] at (3.5,-1) (u3) {$u_3$};
 \node[vertex] at (5,0.5)  (v1) {$v_1$};
 \node[vertex] at (5,-0.5) (v2) {$v_2$};
 \node[vertex] at (6.5,1)  (w1) {$w_1$};
 \node[vertex] at (6.5,0)  (w2) {$w_2$};
 \node[vertex] at (6.5,-1) (w3) {$w_3$};
 \draw (u1) -- (v1) -- (w1);
 \draw (u1) -- (v2) -- (w1);
 \draw (u2) -- (v1) -- (w2);
 \draw (u2) -- (v2) -- (w2);
 \draw (u3) -- (v1) -- (w3);
 \draw (u3) -- (v2) -- (w3);
 \end{tikzpicture}
\end{center}

\noindent The algorithm is given in Algorithm~\ref{alg:round1}. 
We begin with the following lemma:

\begin{lemma}\label{lem:roundp1}(First Phase and the Output Phase)
  Suppose that all vertex constraints are satisfied and $\sum_j
  y_{ij}\leq b_i-1$ for some $i \in V$. Then, for any odd set $U$ that contains $i$, the
  corresponding odd set constraint is satisfied.  The fractional
  solution $\{\yone_{ij}\}$ obtained in the first phase of
  Algorithm~\ref{alg:round1} is feasible for $\lpbone$ --- and 
  an integral $\mone$ which is a $(1-2\delta)$-approximation of $\lpbone$
  can be output along with $\mzero$.
\end{lemma}

\begin{proof}
For any $U \in \O, i \in U$, {\small $\displaystyle 
  \sum_{i',j\in U} y_{i'j} 
   \leq \frac{1}{2}\sum_{i'\in U} \sum_j y_{i'j}
   \leq \frac{1}{2}((\sum_{i'\in U} b_{i'}) - 1)
   = \frac{\bnorm{U}-1}{2} = \bfloor{U}$.}
Thus it follows that any vertex which has an edge incident to it in
 $\mzero$ cannot be in any violated odd-set in $\lpbone$. Then any
 violated odd-set in $\lpbone$ with respect to $\{\yone_{ij}\}$ must also be a violated odd-set in
 $\lpb$; contradicting the fact that we started with a $\{y_{ij}\}$ is
 feasible for $\lpb$.
Now $\M^{(0)} \cup \mone$ is feasible since both are integral and we know that 
$b^{(1)}_i \leq b_i - \sum_j \h{y}^{(0)}_{ij}$.
Observe that $w (\mzero) \geq (1-\delta) \sum_{(i,j)\in E} w_{ij} \left
  (y_{ij} - \yone_{ij} \right)$ where $w (\M^{(0)}) = \sum_{(i,j) \in E}
\hy_{ij}^{(0)} w_{ij} $. Therefore if $w (\mone) \geq (1-2\delta) 
\sum_{(i,j)\in E} w_{ij} \yone_{ij}$ then $w (\M^{(0)})+w (\mone)$ is at least 
$(1-2\delta) \sum_{(i,j)\in E} w_{ij} y_{ij}$ as desired.
\end{proof}

\begin{lemma}\label{lem:roundp2}(Second Phase)
If $\{\yone_{ij}\}$ satisfies $\lpbone$ over $V$, then 
  $\{\ytwo_{ij}\}$ satisfies $\lpbtwo$ over $G^{(2)}$ and 
$\sum_{i,j} w_{ij} \ytwo_{ij}= (1-\delta) \sum_{i,j} w_{ij} \yone_{ij}$.
\end{lemma}

\begin{proof}
  Observe that any vertex which participates in any split produces
  vertices which have (fractionally) at least $t$ edges. After
  scaling we have $(1-\delta) \sum_j \yone_{ij} \leq
  \sum_j \yone_{ij} - \delta t \leq \sum_j \yone_{ij} - 2 \leq \btwo_i - 1$ from the 
  definition of $\btwo$ in line (3b) of Algorithm~\ref{alg:round1}.  
  Therefore the new vertices
  cannot be in any violated vertex or set constraint; from the first part of 
  Lemma~\ref{lem:roundp1} (now applied to $\lpbtwo$ instead of $\lpb$).
  Therefore the Lemma follows.
\end{proof}

\noindent Finally, observe that any {\bf integral} $b$--matching in $G^{(2)}$
has an integral matching in $G^{(3)}$ of the same weight and vice
versa --- moreover given a matching for $G^{(3)}$ the integral
$b$--matching for $G^{(2)}$ can be constructed trivially.
Also, the number of edges in $G^{(3)}$ is at most $O(\delta^{-2}m')$
since each vertex in $G^{(2)}$ is split into $O(\delta^{-1})$ vertices
in $G^{(3)}$. 
We are guaranteed a maximum $b$--matching in $G^{(2)}$ of weight at
least $\sum_{(i,j) \in E^{(2)}} w_{ij} \ytwo_{ij}$ since
$\{\ytwo_{ij}\}$ satisfies $\lpbtwo$ over $G^{(2)}$.  Therefore we are
guaranteed a matching of the same weight in $G^{(3)}$.  Now, we use
the approximation algorithm in \cite{DuanP10,DuanPS11} which returns a
$(1-\delta)$-approximate maximum weighted matching in $G^{(3)}$ in
$O(m' \delta^{-3} \log (1/\delta))$ time and space. 
From the
$(1-\delta)$-approximate maximum matching we construct a
$b$--matching in $G^{(2)}$ of the same weight (and therefore a
$b$--matching $\mone$ in $G^{(1)}$ of the same weight). 
Theorem~\ref{thm:rounding} follows.

\newcommand{\tbc}{\widetilde{\beta^c}}
\newcommand{\tyc}{\widetilde{\by^c}}
\newcommand{\bbc}{\widetilde{\bb^c}}

\section{The Capacitated $b$--Matching Problem}
\label{sec:capacitated}

\begin{ndefinition}{\ref{def2}}\cite[Chapters 32 \& 33]{Schrijver03}
The {\bf Capacitated $\mathbf b$--matching} problem is a $b$--matching problem where 
we have an additional restriction
that the multiplicity of an edge $(i,j) \in E$ is at most $c_{ij}$. The vertex and edge capacities \
$\{b_i\},\{c_{ij}\}$ are given as input and for this paper are assumed to be integers 
in $[0,\poly n]$. Observe that we can assume $c_{ij} \leq \min\{ b_i , b_j \}$ without loss of generality.
\end{ndefinition}

\noindent{\bf Long and Short Representations:}
We follow the reduction of the capacitated problem to the uncapacitated problem outlined in \cite[Chapter 32]{Schrijver03}, with modifications.

\begin{definition}
\label{deftwoc}
Given a graph  $G=(V,E)$ with vertex and edge capacities. Consider
subdividing each edge $e=(i,j)$ into 
$(i,p_{ij,i}),(p_{ij,i},p_{ij,j}),(p_{ij,j},j)$ where
$p_{ij,i},p_{ij,j}$ are new additional vertices with capacity
$b^c_{p_{ij,i}}=b^c_{p_{ij,j}}=c_{ij}$. For $i\in V$ set $b^c_i=b_i$. 
We use the weights denoted by $\bw^c$ to be $\frac12w_{ij},0,\frac12w_{ij}$ for
$(i,p_{ij,i}),(p_{ij,i},p_{ij,j}),(p_{ij,j},j)$ respectively.
Let the transformation of $G$ be denoted as $\ylong{G}$; let the
vertices and edges of $\ylong{G}$ be $V^c$ and $E^c$
respectively. $\ylong{G}$ does not have any edge capacities.

For $ U^c \subseteq V^c$ let $\bnorm{U^c}=\sum_{s \in U^c}
b^c_s$ as before.  The odd-sets in $\ylong{G}$
, i.e. $U^c$ such that $\bnorm{U^c}$
is odd, are denoted by $\Oc$ and define $\Odc=\{U^c \in \Oc, \bnorm{U^c}
\leq 1/\delta \}$.
\end{definition}

\smallskip \noindent 
The above transformation is inspired by the proof of \cite[Theorem~32.4, Vol A, page
567]{Schrijver03} which used the weights $w_{ij},w_{ij},w_{ij}$ instead of
$\frac12w_{ij},0,\frac12w_{ij}$ for $(i,p_{ij,i}),(p_{ij,i},p_{ij,j}),(p_{ij,j},j)$ 
respectively. In fact \cite[Theorem~32.4]{Schrijver03} computes an optimum
solution of value $\beta^{*,c} + \sum_{(i,j) \in E} w_{ij} c_{ij}$. However an 
approximation of $\beta^{*,c} + \sum_{(i,j) \in E} w_{ij} c_{ij}$ need not 
provide an approximation of  $\beta^{*,c}$ because $\sum_{(i,j) \in E} w_{ij} c_{ij}$ 
can be significantly larger. We will eventually use the algorithm in 
\cite[Theorem~32.4]{Schrijver03} to find an integral solution in Section~\ref{sec:crounding}.
We need to bound $\sum_{(i,j) \in \hat{E}} w_{ij} c_{ij}$ where
$\hat{E}$ is the edgeset in our candidate fractional solution.  An
example of the transformation is as follows (the edges only have
weight in the new graph).

\smallskip
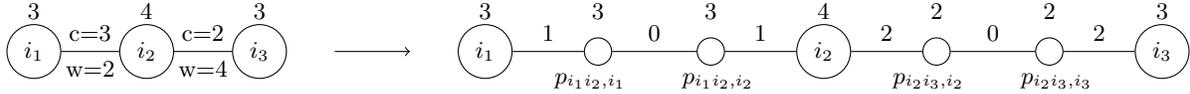
\centerline{
\begin{tikzpicture}[font=\footnotesize, scale=1]
 \tikzstyle{vertex}=[draw,shape=circle,radius=0.25cm]
 \node[vertex] at (0,0) (u) {$i_1$};
 \node[above,yshift=0.3cm] at (u) {3};
 \node[vertex] at (1.5,0) (v) {$i_2$};
 \node[above,yshift=0.3cm] at (v) {4};
 \node[vertex] at (3,0) (w) {$i_3$};
 \node[above,yshift=0.3cm] at (w) {3};
 \draw (u) to node[above] {c=3}  node[below] {w=2} (v);
 \draw (v) to node[above] {c=2}  node[below] {w=4} (w);
 \draw[->] (4,0) to (5,0);
 \node[vertex] at (6,0)  (i1) {$i_1$};
 \node[above,yshift=0.3cm] at (i1) {3};
 \node[vertex] at (7.5,0) (ei1) {}; 
 \node[above,yshift=0.3cm] at (ei1) {3};
 \node[below,xshift=-0.1cm,yshift=-0.15cm] at (ei1) {$p_{i_1i_2,i_1}$};
 \node[vertex] at (9,0) (ei12) {};
  \node[below,xshift=0.1cm,yshift=-0.15cm] at (ei12) {$p_{i_1i_2,i_2}$};
 \node[above,yshift=0.3cm] at (ei12) {3};
 \node[vertex] at (10.5,0) (i2) {$i_2$};
 \node[above,yshift=0.3cm] at (i2) {4};
 \node[vertex] at (12,0)  (ei2) {}; 
 \node[above,yshift=0.3cm] at (ei2) {2};
  \node[below,xshift=-0.1cm,yshift=-0.15cm] at (ei2) {$p_{i_2i_3,i_2}$};
 \node[vertex] at (13.5,0) (ei23) {};
 \node[above,yshift=0.3cm] at (ei23) {2};
 \node[below,xshift=0.1cm, yshift=-0.15cm] at (ei23) {$p_{i_2i_3,i_3}$};
 \node[vertex] at (15,0)  (i3) {$i_3$};
 \node[above,yshift=0.3cm] at (i3) {3};
 \draw (i1) to node[above] {1}  (ei1);
 \draw (ei1) to node[above] {0}  (ei12);
 \draw (ei12) to node[above] {1}  (i2);
 \draw (i2) to node[above] {2}  (ei2);
 \draw (ei2) to node[above] {0}  (ei23);
 \draw (ei23) to node[above] {2}  (i3);
 \end{tikzpicture}
}

\noindent 
{\bf Notation}: We will use $i,j$ to denote vertices (and edges) in
the original graph $G$ and use $s,r,u,v$ to denote vertices (and
edges) in $\ylong{G}$. We will use the superscript such as $y^c,U^c$ to
indicate variables, subsets in $\ylong{G}$ to distinguish them from
$G$. However we can switch between $G$ and $\ylong{G}$ as described
next.

\begin{definition} 
Let $\lambda^c_0$ be a parameter which is determined later. Define:
{\small
\[ 
\ylong{\Qc} :  \left \{\begin{array}{ll}
\displaystyle  \sum_{r:(s,r) \in E^c} y^c_{sr} \leq b^c_s  & \forall s \in V^c \\
  y^c_{ip_{ij,i}}+y^c_{p_{ij,i}p_{ij,j}} = c_{ij} & \forall(i,j)\in E \\
  y^c_{p_{ij,i}p_{ij,j}} + y^c_{p_{ij,j}j} = c_{ij} & \forall(i,j)\in E \\
 y^c_{sr}\geq 0 &  \forall (s,r) \in E^c
 \end{array} \right. \quad \mbox{and} \quad
\Qc :  \left \{\begin{array}{ll}
\displaystyle  \sum_{j:(i,j) \in E} y_{ij} \leq b_i & \forall i \in V \\
y_{ij} \leq c_{ ij} & \forall(i,j)\in E \\
y_{ij} \geq 0  & \forall(i,j)\in E 
\end{array}\right.
\]
}
And likewise:
{\small
\[ 
\ylong{\Pc} :  \left \{\begin{array}{ll}
\displaystyle  \sum_{r:(s,r) \in E_L} y^c_{sr} \leq \lambda^c_0 b^c_s/2  & \forall s \in V^c \\
  y^c_{ip_{ij,i}}+y^c_{p_{ij,i}p_{ij,j}} = c_{ij} & \forall(i,j)\in E \\
  y^c_{p_{ij,i}p_{ij,j}} + y^c_{p_{ij,j}j} = c_{ij} & \forall(i,j)\in E \\
 y^c_{sr}\geq 0 &  \forall (s,r) \in E^c
 \end{array} \right. \quad \mbox{and} \quad
\Pc :  \left \{\begin{array}{ll}
\displaystyle  \sum_{j:(i,j) \in E} y_{ij} \leq \lambda^c_0 b_i/2 & \forall i \in V \\
y_{ij} \leq c_{ ij} & \forall(i,j)\in E \\
y_{ij} \geq 0  & \forall(i,j)\in E 
\end{array}\right.
\]
}
Given $\by^c \in \ylong{\Pc}$  define $\yshort{\by^c}$ as 
$y_{ij} \leftarrow y^c_{i,p_{ij,i}} (= y^c_{p_{ij,j},j})$. Observe that 
$\yshort{\by^c} \in \Pc$.
Likewise given a $\by\in \Pc$ define $\ylong{\by}$ as 
$y^c_{i,p_{ij,i}}, y^c_{p_{ij,j},j} \leftarrow y_{ij}$ and $y_{p_{ij,i},p_{ij,j}} \leftarrow (c_{ij} - y_{ij})$. Observe that $\ylong{\by} \in \ylong{\Pc}$.
Moreover $\ylong{\cdot}, \yshort{\cdot}$
are inverse operations; $\yshort{\by^c}=\by$ iff $\ylong{\by}=\by^c$ 
and define bijections between $\ylong{\Pc},\Pc$ and between
$\ylong{\Qc},\Qc$.

Moreover for any $\by^c \in \ylong{\Pc}$ (therefore also $\ylong{\Qc}$) we have
$\bw^T\yshort{\by^c} = (\bw^c)^T\by^c$. Similarly for any $\by \in \Pc$ (therefore also 
$\Qc$) $(\bw^c)^T\ylong{\by} = \bw^T\by$.
\label{longdef1}
\end{definition}

\noindent The next theorem provides the linear program we will use for capacitated $b$--matching.

\begin{theorem}
\label{capacity:thm}
The maximum integral weighted capacitated $b$--matching problem is expressed by
the following linear programming relaxation on $\ylong{G}$.
{\small
\begin{align*}
& \displaystyle \beta^{*,c} = \max \sum_{(s,r) \in E^c} w^c_{sr} y^c_{sr} \\
& \begin{array}{r l}
\{ \bA^c\by^c \leq \bb^c \} = &\left\{ \begin{array}{ll}
\displaystyle  \sum_{r:(s,r) \in E^c} y^c_{sr} \leq b^c_s  & \forall s \in V^c \\
\displaystyle \sum_{(s,r) \in E^c:s,r\in U} y^c_{sr} \leq \left \lfloor \frac{\bnorm{U^c}}{2} \right \rfloor  &  \forall U^c \in \Odc \end{array}\right. \\
& \displaystyle \sum_{(s,r) \in E^c:s,r\in U} y^c_{sr} \leq \left \lfloor \frac{\bnorm{U^c}}{2} \right \rfloor  \quad  \forall U^c \in \Oc - \Odc \\
\ylong{\Qc} = &  \left \{\begin{array}{ll}
\displaystyle  \sum_{r:(s,r) \in E^c} y^c_{sr} \leq b^c_s  & \forall s \in V^c \\
  y^c_{ip_{ij,i}}+y^c_{p_{ij,i}p_{ij,j}} = c_{ij} & \forall(i,j)\in E \\
  y^c_{p_{ij,i}p_{ij,j}} + y^c_{p_{ij,j}j} = c_{ij} & \forall(i,j)\in E \\
 y^c_{sr}\geq 0 &  \forall (s,r) \in E^c
 \end{array} \right. 
\end{array} 
 \lptag\label{lp:cbm-long-int}
\end{align*}
}
The final solution is given by $\by \leftarrow \yshort{\by^c}$. Some of the constraints are redundant by design.
\end{theorem}

\begin{proof} Given an integral feasible solution $\by$ for capacitated $b$--matching, the constraints $\{\bA^c\ylong{\by} \leq \bb^c\}$ hold because $\ylong{\by}$ defines an integral uncapacitated $b$-matching over $\ylong{G}$. The new constraints $\ylong{\Q^c}$ are satisfied since $\by$ is feasible, i.e., $\by \leq \bc$. 
Note that the objective function value does not change as a consequence of Definition~\ref{deftwoc}.
This proves that $\beta^{*,c}$ is an upper bound on the maximum capacitated integral $b$--matching. 

In the reverse direction, given a fractional solution $\by^c$ with objective value $\beta^{*,c}$, observe that $\by^c$ satisfies the conditions of being in the uncapacitated $b$-matching polytope of $\ylong{G}$ (recall these constraints are in \ref{lpbm}). Therefore $\by^c$ can be expressed as a convex combination of integral uncapacitated $b$--matchings over $\ylong{G}$. Since $\by^c$ satisfies that the vertex capacities in $V^c - V$ as an equality (see $\ylong{\Qc}$) -- { \em every integral uncapacitated $b$--matching in the decomposition of $\by^c$ must satisfy the vertex capacities $V^c-V$ as equality.} Therefore there exists at least one integral uncapacitated $b$--matching $\t{\by^c}$ in the decomposition of $\by^c$ which has objective value at least $\beta^{*,c}$ and satisfies  the vertex capacities for $V^c - V$ as equality.  Now $\yshort{\t{\by^c}}$ is an integral capacitated $b$--matching in $G$ of weight at least $\beta^{*,c}$.
\end{proof}

\noindent{\bf Approximate Satisfiability.} Since we will not be satisfy the constraints \ref{lp:cbm-long-int} exactly the next lemma provides an ability to scale solutions.

\begin{lemma}
\label{scalelemma}
Let $q$ be an arbitrary integer and let $\zeta \geq 1$.
Suppose that we have a $\by^c \in \ylong{\Pc}$ which for all $U^c \subseteq V^c$ in $\ylong{G}$ with $\bnorm{U^c} \leq q$ satisfies
{\small
\[  \zeta \left\lfloor \frac{\bnorm{U^c}}{2} \right\rfloor \geq \sum_{(s,r) \in E^c, s,r \in U^c} y^c_{sr} \]
}
then $\h{\by}^c=\ylong{\frac1{\zeta}\yshort{\by^c}}$  satisfies for all $U^c \subseteq V^c$ in $\ylong{G}$ with $\bnorm{U^c} \leq q$, 
{\small
\begin{equation}
\label{meh}  \left\lfloor \frac{\bnorm{U^c}}{2}\right\rfloor \geq \sum_{(s,r) \in E^c, s,r \in U^c} \h{y}^c_{sr} 
\end{equation}
}
\end{lemma}

\begin{proof}
Suppose not. Consider the subset $U^c$ with the smallest $\bnorm{U^c}$ which violates the assertion \ref{meh}. Observe that $U^c$ cannot contain both $p_{ij,i},p_{ij,j}$  for any edge $(i,j) \in E$ (in the original $G$). Because in that case $U^c - \{p_{ij,i},p_{ij,j}\}$ will be a smaller set which violates the assertion -- since the LHS of Equation~\ref{meh} will decrease by $c_{ij}$ as well as the RHS!
But if $U^c$ does not contain both $p_{ij,i},p_{ij,j}$  for any edge $(i,j) \in E$ (in $G$) then 

{\small
\[ \sum_{(s,r) \in E^c, s,r \in U^c} \h{y}^c_{sr}  \leq \frac{1}{\zeta} \sum_{(s,r) \in E^c, s,r \in U^c} y^c_{sr} \leq \left\lfloor \frac{\bnorm{U^c}}{2}\right\rfloor \]
} 
which is a contradiction.
The lemma follows.
\end{proof}

Therefore the scaling operation still succeeds (on $\yshort{\by^c}$) but its proof is more global compared to the proof in the uncapacitated case. Here we are proving the statement for all subsets of a certain size simultaneously, whereas in the uncapacitated case the proof of feasibility of $U^c$ followed from the bound of $\sum_{(i,j) \in E, i,j  \in U^c} y_{ij}$ for that particular subset $U^c$ itself.

\subsection{Algorithm for Capacitated $b$--Matching} The algorithm is provided in Algorithm~\ref{alg:cap}.

\begin{algorithm}[t]
{\small
\begin{algorithmic}[1]\parskip=0in
\STATE Define $\ylong{\Qc},\ylong{\Pc}$ as in Definition~\ref{longdef1}. Define $\bA^c\by^c \leq \bbc$ as: 
{\footnotesize
\begin{align*}
 \begin{array}{rl l}
\{\bA^c\by^c \leq \bbc\} =  & \left \{
\begin{array}{ll}
   \displaystyle \sum_{r:(s,r) \in E^c} y^c_{sr} \leq \t{b}^c_s  
   & \forall s \in V^c, \quad \mbox{where } \t{b}^c_s= (1-4\delta) b^c_s 
\\
\displaystyle   \sum_{(s,r) \in E^c:s,r\in U} y^c_{sr} \leq \t{b^c}_U 
   & \forall U^c \in \Odc \quad \mbox{where } \t{b}^c_U= \left \lfloor \frac{\bnorm{U^c}}{2}\right \rfloor - \frac{\delta^2 \bnorm{U^c}^2 }{4}
\end{array} \right. 
\lptag \label{lpc:001} 
\end{array}
\end{align*}
}
\label{line1c}

\STATE \label{line2c} 
Fix $\delta \in (\frac{1}{\sqrt{5n}},\frac1{16}]$. Let $\lambda^c_0 = 16 \ln \frac2\delta$.
Let $\alpha = 50\delta^{-3} \ln (2m+n)$.

\STATE Find a solution $\by^c \in \ylong{\Pc}$ where $\beta_0=(\bw^c)^T\by^c$ and $\bA^c\by^c\leq \lambda_0 \bbc $.
\label{initcap}

\STATE Let $\epsilon=\frac18$ (note $\epsilon \geq \delta$) and $t=0$.

\WHILE{{\bf true}}

\STATE Define $\lambda=\max \{ \max_{i} \lambda_i, \max_{U^c\in \Odc} \lambda_{U^c} \} $ where 
$
\left\{ \begin{array}{ll} \lambda_s= \sum_{r:(s,r) \in E^c} y^c_{sr}/\t{b}^c_s & \forall s \in V^c\\
\lambda_{U^c} =\sum_{(s,r) \in E^c:s,r\in U^c} y^c_{sr}/\t{b}^c_U \qquad &\forall U^c \in \Odc \end{array}\right.$
\STATE Find a collection of odd sets $L^c =\{ U^c \mid U^c \in \Odc, \lambda_{U^c} \geq \lambda - \frac{\delta^3}{10}\}$ (without computing all $\lambda_{U^c}$).
\label{stepcomputec}
\STATE \label{outputc}
If $(\lambda \leq 1+8\delta)$ output $\frac{(1-\delta)}{(1+8\delta)}\yshort{\by^c}$ and stop. 
\STATE If $\lambda<1+8\epsilon$ then a new {\bf superphase} starts; repeatedly set $\epsilon \leftarrow \max \{2\epsilon/3,\delta\}$ till $\lambda \geq 1+8\epsilon$.

\STATE Set \label{thresholdstepc}
$ \left\{ \begin{array}{l} 
x_s=\exp(\alpha \lambda_s )/\t{b}^c_s \mbox{ if $\lambda_s > \lambda - \delta^3/10$ and $0$ otherwise} \\ 
z_{U^c}=\exp(\alpha \lambda_{U^c} )/\t{b}^c_U \mbox{ if $\lambda_{U^c} > \lambda - \delta^3/10$ and $0$ otherwise} 
\end{array} \right.$. 
Let $\gamma^c=\sum_s x_s \t{b}^c_s + \sum_{U^c \in \Odc} z_{U^c} \t{b}^c_U$. 

\STATE \label{effectivestepc} Define $\displaystyle \eta_{sr}=(x^c_s+x^c_r+\sum_{U^c \in \L^c;s,r\in U^c} \hspace{-0.1cm} z_{U^c})$.
Find a solution $\tyc$ of \ref{lpc:002}, otherwise decrease $\beta \leftarrow (1-\delta)\beta$.
{\footnotesize
\[ \left\{ \sum_{(s,r) \in E^c} w^c_{sr} \widetilde{y^c}_{sr} \geq (1-\delta)\beta, \qquad
\sum_{(s,r) \in E^c} \widetilde{y^c}_{sr} \eta_{sr} \leq \frac{\gamma^c}{1-\delta}, \qquad \tyc \in \ylong{\Pc} \lptag \label{lpc:002} \right\}
\]
}
\STATE Set $\by^c \leftarrow (1-\sigma) \by^c + \sigma \tyc$ where $\sigma=\epsilon/(4\alpha\lambda^c_0)$.

\ENDWHILE 
 \end{algorithmic}
 \caption{An approximation scheme for capacitated $b$--matching \label{alg:cap}.}
}
\end{algorithm}

Note that Step~\ref{stepcomputec} follows from
Lemma~\ref{zulambda1}. Moreover if we adjust $\alpha$ for the number
of vertices, Lemma~\ref{zulambda2} also follows. Note that $\lambda^* = \min\{
\lambda \mid \by^c \in \ylong{\Qc}, \bA^c \by^c \leq \lambda \bbc \}$
is not $1$. In Lemma~\ref{lem:exists} we show that $\lambda^* \leq
1/(1-4\delta)$ and moreover we can always find a solution of
\ref{lpc:002} for $\beta \leq (1-4\delta) \beta^{*,c}$. However the
choice of $\bA^c \by^c \leq \bbc$ implied that we can reuse
Lemma~\ref{zulambda1} and \ref{zulambda2} without any modification.

Before discussing the algorithm for \ref{lpc:002} we argue that the returned solution
returned in Line~\ref{outputc} of Algorithm~\ref{alg:cap} is a
feasible capacitated $b$--matching. We apply Lemma~\ref{scalelemma} with
$\zeta = (1+8\delta)$ and $q=1/\delta$. Consider
$\by^{c,\dagger}=\ylong{\frac1{1+8\delta}\yshort{\by^c}}$. Since
$\bA^c\by^c \leq (1+8\delta)\t{\bb^c} \leq (1+8\delta)\bb^c$. Note
that this operation will imply that all the vertex constraints in
$V_c$ are satisfied as well as constraints corresponding to all $U^c
\in \Odc$. For the odd subsets $U^c$ with $\bnorm{U^c} \geq 1/\delta$,
since the vertex constraints are satisfied we have: 
{\small
\[ \sum_{(s,r) \in E^c, s,r \in U^c} y^{c,\dagger}_{sr}  \leq \frac{\bnorm{U^c}}{2} \leq \frac{1}{(1-\delta)} \left \lfloor \frac{\bnorm{U^c}}{2} \right \rfloor 
\]
}
the violation is at most $\zeta=\frac1{(1-\delta)}$ for any odd set. 
We now apply the lemma again with 
$\zeta=\frac1{(1-\delta)}$ for all odd sets, i.e., $q=\infty$. The result of the two operations compose and correspond to the output in Line~\ref{outputc}.

\paragraph{Solving \ref{lpc:002}.}
We now focus on the algorithm for solving \ref{lpc:002}. Before
providing the algorithm we prove Lemma~\ref{lem:pij} which proves
structural properties of the weights resulting from the dual
thresholding.

\begin{lemma}\label{lem:pij}
Suppose that $\lambda > 1+8\delta$ (otherwise the algorithm has stopped) 
and the current candidate solution in Algorithm~\ref{alg:cap} is $\vecy^c$. 
\begin{enumerate}[~~(a)]
\item $x_s=0$ for any $s \in V^c - V$ (the new vertices that are introduced).
 \item Suppose 
$U^c \in \Odc$ contains $p_{ij,i},p_{ij,j}$ for some edge $(i,j) \in E$ (of $G$). If
neither $i,j \notin U$,  $z_{U^c} =0$. 
\item If for some edge $(i,j) \in E$ we have $y^c_{ip_{ij,i}}=y^c_{p_{ij,j}j}=0$, then neither $p_{ij,i},p_{ij,j}$ belong to an add set $U^c \in \Odc$ with $z_{U^c} >0$. As a consequence, We can compute $\L^c$ in time $O(m' \poly
\{\delta^{-1},\log n\})$ where $m'=|\{(i,j)| y^c_{p_{ij,j}j} \neq 0\}|$ 
because the other edges cannot define any odd set in
$\L_c$. 
\item
Let $\shortcost{\boldsymbol{\eta}}_{ij} = \eta_{ip_{ij,i}}+\eta_{p_{ij,j}j} - \eta_{p_{ij,i}p_{ij,j}}$. Then $\shortcost{\boldsymbol{\eta}}_{ij} \geq 0$ for every $(i,j) \in E$.
\item Let $\displaystyle \shiftcost{\boldsymbol{\eta}} = \sum_{(i,j)\in E} c_{ij} \eta_{p_{ij,i}p_{ij,j}}$ then $
\frac{\gamma^c}{(1-\delta)} \geq \shiftcost{\boldsymbol{\eta}}$.
\end{enumerate}
\end{lemma}

\begin{proof}
Part (a) follows from the fact that $\lambda_s = \frac1{1-4\delta} < \lambda - \delta^3/10$.

For part (b)
suppose that $p_{ij,i}, p_{ij,j} \in U^c, z_{U^c} \neq 0$ for some $U^c \in \Odc$. Note $\bnorm{U^c}\leq 1/\delta$ and thus:
{\small
\begin{equation}
\sum_{(s,r) \in E^c:s,r \in U^c} \hspace{-0.2in} y^c_{sr} \geq \left(\lambda -\frac{\delta^3}{10}\right)\t{b}^c_U =
\left(\lambda -\frac{\delta^3}{10}\right)
\left( \left \lfloor \frac{\bnorm{U^c}}{2} \right \rfloor - \frac{\delta^2\bcnorm{U^c}^2}{4}\right)
\geq \left(\lambda -\frac{\delta^3}{10}\right)\frac{\bnorm{U^c}-1}{2} - \lambda \frac{\delta^2\bcnorm{U}^2}{4}
\label{pij1}
\end{equation}
}
Consider $U^c_1 = U^c - \{p _{ij,j},p_{ij,i}\}$.
 Let $\bnorm{U^c_1}=\ell$. Note since $\bnorm{U^c}$ is odd, $\ell \geq 1$ and $\ell$ is odd. 
 $\bnorm{U}=\ell+2c_{ij} \geq 3$. Since $y_{p_{ij,i}p_{ij,j}} \leq c_{ij}$ and $i,j \not \in U^c$,
{\small
\begin{equation}
\sum_{(s,r) \in E^c:s,r \in U^c_1} y^c_{sr} 
= \sum_{(s,r) \in E^c:s,r \in U} y^c_{sr} - y_{p_{ij,i}p_{ij,j}} \geq \sum_{(s,r) \in E^c:s,r \in U^c} y^c_{sr}  - c_{ij}
\label{pij2}
\end{equation}
}
If $U^c_1$ is a singleton node then the LHS of Equation~\eqref{pij2} is $0$. Using Equation~\eqref{pij1} and $3 \leq \bcnorm{U} \leq 1/\delta$,
{\small
\[ 0 = \sum_{(s,r) \in E^c:s,r \in U^c} y^c_{sr} - y_{p_{ij,i}p_{ij,j}}
\geq \left(\lambda -\frac{\delta^3}{10}\right) 
\t{b}^c_U  - c_{ij}\geq 
\left(\lambda -\frac{\delta^3}{10}\right) 
(1-\delta) \left \lfloor \frac{\ell+2c_{ij}}{2} \right \rfloor - c_{ij}
\geq c_{ij} \left( \left(\lambda -\frac{\delta^3}{10}\right) 
(1-\delta) - 1 \right)
\]
}
which is impossible for $\lambda > 1 + 8\delta$.
If $U^c_1$ is an odd set, then it is in $\Odc$ and thus
{\small
\begin{equation}
\lambda \left( \frac{\ell-1}{2} - \frac{\delta^2\ell^2}{4}\right) = \lambda \t{b}_{U^c_1} \geq  \left(\lambda -\frac{\delta^3}{10}\right)\frac{\ell+2c_{ij}-1}{2} - \lambda \frac{\delta^2(\ell+2c_{ij})^2}{4} - c_{ij}
\label{pij3}
\end{equation}
} but Equation~\ref{pij3} rearranges to 
{\small
\[
\frac{\delta^3}{10}\left(\frac{\ell-1}{2}\right) + \lambda c_{ij} \delta^2 (c_{ij}+\ell) \geq \left(\lambda - \frac{\delta^3}{10} - 1\right) c_{ij}
\]
}
which in turn (if we divide by $c_{ij}$ and use $\ell+2c_{ij} \leq 1/\delta$ ) implies $1 + \frac{\delta^3}{10} + \frac{\delta^2}{20} \geq (1-\delta) \lambda$ which is impossible for 
$\lambda > 1 + 8\delta$. Part (b) of the Lemma follows.

For part (c), suppose for contradiction, $p_{ij,i} \in U^c$ $y^c_{p_{ij,i}p_{ij,j}}=c_{ij}$ and $z_{U^c} > 0$. Observe Equation~\eqref{pij1} applies because $z_{U^c}>0$.
If $p_{ij,j} \in U^c$ then we consider $U^c_1= U^c - \{p_{ij,i},p_{ij,j}\}$ and in part (b). Equation~\eqref{pij2} of part (b) holds irrespective of $i,j \in U^c$ because neither $p_{ij,i},p_{ij,j}$ have nonzero edges in $\by^c$ to any other vertex. The remainder of part (b) applies as well and we have a contradiction. 

Therefore we need to only consider the case $p_{ij,j} \not \in U^c$. 
But then consider $U^c_2=U^c-\{p_{ij,i}\}$. In this case, since $p_{ij,i}$ has no non-zero 
edge to any vertex in $U^c_2$:
{\small
\begin{equation}
\sum_{(s,r) \in E^c:s,r \in U^c_2} y^c_{sr} 
= \sum_{(s,r) \in E^c:s,r \in U^c} y^c_{sr} > 0
\label{pij4}
\end{equation}
}
Again let $\bnorm{U^c_2}=\ell$, thus $\bnorm{U^c}=\ell+c_{ij}$. Note $\bnorm{U^c} \leq 1/\delta$. Now
{\small
\[ 
\sum_{(s,r) \in E^c:s,r \in U^c_2} y^c_{sr} \leq \frac12 \sum_{s \in U^c_2} \sum_{r:(s,r) \in E^c} y^c_{sr} \leq \frac12 \sum_{s \in U^c_2} \lambda\t{b^c}_s =\frac{(1-4\delta) \ell}{2} \lambda
\]
}
Combining the above with Equations~\ref{pij4} and (first part of) \ref{pij1}
{\small
\[
\frac{(1-4\delta) \ell}{2} \lambda \geq  \left(\lambda -\frac{\delta^3}{10}\right)\t{b}_{U^c} \geq \left(\lambda -\frac{\delta^3}{10}\right)(1-\delta) \left \lfloor \frac{\bnorm{U^c} }{2}\right\rfloor = \left(\lambda -\frac{\delta^3}{10}\right)(1-\delta) \left\lfloor \frac{\ell + c_{ij}}{2} \right \rfloor 
\]
}
and since $c_{ij} \geq 1$, the above implies
{\small
\[ \frac{(1-4\delta) \ell}{2} \lambda \geq
  \left(\lambda -\frac{\delta^3}{10}\right)(1-\delta) \frac{\ell}{2} \quad \implies \quad (1-\delta)\frac{\delta^3}{10} \geq 3 \delta \lambda \]
}
which is not possible for $\lambda > 1 + 8\delta$. Part (c) follows.

\smallskip For part (d) observe that:
{\small
\begin{align*} 
\shortcost{\boldsymbol{\eta}}_{ij} & = \left(x_i + x_{p_{ij,i}} + \sum_{U^c; i,p_{ij,i} \in U^c} z_{U^c} \right) + \left(x_j + x_{p_{ij,j}} + \sum_{U^c; j,p_{ij,j} \in U} z_{U^c} \right) \\
& \qquad - \left( x_{p_{ij,i}} + x_{p_{ij,j}} + \sum_{U^c; p_{ij,i} ,p_{ij,j} \in U^c} z_{U^c} \right) \\
& = x_i + x_j + \sum_{U^c; i,p_{ij,i} \in U^c} z_{U^c} + \sum_{U^c; j,p_{ij,j} \in U^c} z_{U^c}  -  \sum_{U^c; p_{ij,i} ,p_{ij,j} \in U^c} z_{U^c}
\end{align*}
}
\noindent but then $\shortcost{\boldsymbol{\eta}}_{ij}$ can be negative only if there exists a set $U^c \in L^c$ such that $p_{ij,i} ,p_{ij,j} \in U^c$, neither $i,j \not \in U^c$ and $z_{U^c}>0$. The first part of Lemma rules out that possibility.

\smallskip Finally for part (d) observe that there exists a solution $y^c_{p_{ij,i}p_{ij,j}}=c_{ij}$ and $y^c_{i,p_{ij,i}}=y^c_{i,p_{ij,j}}=0$. This solution corresponds to not picking any edges in the original graph $G$. This solution belongs to $\Qc$ and satisfies $\bA^c \by^c \leq \bbc$. Therefore for this solution, for every odd set $U^c \in \Odc$:
{\small
\begin{equation}
\sum_{(s,r) \in E^c:s,r \in U} y^c_{sr} \leq \left \lfloor \frac{\bnorm{U}}{2} \right \rfloor \leq \frac{1}{(1-\delta)} \t{b}_{U^c}
\label{pij5}
\end{equation}
}
Observe that $y^c_{sr} \neq 0$ only for $(s,r)=(p_{ij,i},p_{ij,j})$ for every edge $(i,j) \in E$; and that $y^c_{sr}=c_{ij}$. Therefore multiplying Equation~\ref{pij5} by $z_{U^c} \geq 0$ and  
summing over all $U^c \in \Odc$ we get:
{\small
\begin{align*}
& \sum_{U^c \in \Odc} z_U \left( \sum_{(i,j) \in E: p_{ij,j},p_{ij,i} \in U^c}  c_{ij} \right) \leq \frac{1}{(1-\delta)} \sum_{ U \in \Odc} z_{U^c} \t{b}_{U^c} \leq \frac1{(1-\delta)} \gamma^c \Longrightarrow \\
& \frac1{(1-\delta)} \gamma^c \geq \sum_{(i,j) \in E} c_{ij} \left( \sum_{U^c \in \Odc,  p_{ij,j},p_{ij,i} \in U^c} z_{U^c} \right) =   \sum_{(i,j) \in E} c_{ij} \left( \eta_{p_{ij,j},p_{ij,i}}
- x_{p_{ij,i}} - x_{p_{ij,j}} \right) =  \sum_{(i,j) \in E} c_{ij} \eta_{p_{ij,j},p_{ij,i}} 
\end{align*}
}
where the last part follows from $x_{s}=0$ for any $s \in V^c \setminus V$ (since $\lambda_s = 1/(1-4\delta) < \lambda -\delta^3/10$). The conclusion (c) follows from the definition of $\shiftcost{\boldsymbol{\eta}}$.
\end{proof}

\smallskip
\noindent We now provide a solution for \ref{lpc:002}, but notice that the solution is only provided for $\beta \leq (1-4\delta)\tbc$. This reduces the approximation ratio but the overall approximation remains a $(1-O(\delta))$ approximation.
 
\smallskip
\begin{lemma}\label{lem:exists}
Recall \ref{lpc:002} in Algorithm~\ref{alg:cap}.
{\footnotesize
\[ \left 
\{ \sum_{(s,r) \in E^c} w^c_{sr} \t{y}^c_{sr} \geq (1-\delta)\beta \qquad
\sum_{(s,r) \in E^c} \tyc_{sr} \eta_{sr} \leq \frac{\gamma^c}{1-\delta}  \qquad \tyc  \in \ylong{\Pc} \right \} \tag{\ref{lpc:002}}
\]
} 
where $\eta_{sr}$ are as defined in Step~\ref{thresholdstepc}. A
solution of \ref{lpc:002} is always found for $\beta \leq
(1-4\delta)\tbc$. The solution requires at
most $\ell = O(\ln^2 \frac1\delta)$ invocations of
Theorem~\ref{capapx} and returns a solution $\h{\by}^c$ such
that the subgraph (in $G$) $\h{E}=\{(i,j)| (i,j) \in E,
\yshort{\h{y}^c}_{ij}>0\}$ satisfies $\sum_{(i,j) \in \h{E}}
w_{ij}c_{ij} \leq (16 \ell)\beta^{*,c}$. Recall that $\beta^{*,c}$
is the weight of the optimum capacitated $b$--matching.
\end{lemma}

\smallskip
\begin{proof}
First observe that for any $H_1,H_2$ and $\by=\yshort{\by^c}$ (equivalently $\by^c=\ylong{\by}$), 
{\small
\[
\hspace*{-0.4in}
\left. \begin{minipage}{0.5\textwidth}
\begin{align*}
\begin{array}{ll}
&\displaystyle \sum_{(s,r) \in E_L} w^c_{sr}y^c_{sr} = H_1   \\
& \displaystyle\sum_{r:(s,r) \in E_L} \eta_{sr} y^c_{sr} = H_2 \\
& \by^c \in \ylong{\Pc} \quad \mbox{ (resp. $\ylong{\Q^c}$)}
\end{array}
\end{align*}
\end{minipage} \right\} \quad \Longleftrightarrow \quad \left \{
\begin{minipage}{0.45\textwidth}
\begin{align*}
\begin{array}{lll}
& \displaystyle\sum_{(i,j) \in E} w_{ij}y_{ij} = H_1 \\
& \displaystyle\sum_{(i,j)\in E} \shortcost{\boldsymbol{\eta}}_{ij} y_{ij} = H_2 - \shiftcost{\boldsymbol{\eta}} \\
& \by \in \Pc \quad \mbox{ (resp. $\Qc$)}
\end{array}
\end{align*}
\end{minipage} \right.
\]
}
Suppose that we can provide a solution for the system:
{\small
\[ \displaystyle\sum_{(i,j) \in E} w_{ij}y_{ij} \geq (1-\delta)\beta \qquad \sum_{(i,j)\in E} \shortcost{\boldsymbol{\eta}}_{ij} y_{ij} \leq \frac{\gamma^c}{(1-\delta)} - \shiftcost{\boldsymbol{\eta}}, \qquad \by \in \Pc  \lptag \label{lp:exists}
\]
}
for any $\beta \geq (1-4\delta)\beta^*$, then we have proved the lemma by considering $\ylong{\by}$.

\smallskip
Consider the optimum capacitated $b$--matching $\by^{c,*}$, and let $\by^{c,\dagger}=\ylong{(1-4\delta)\yshort{\by^{*,c}}}$. Observe $\by^{c,\dagger} \in \ylong{\Q_c}$. Note that for $i \in V$,
{\small
\[ \sum_{r} y^{c,*}_{ir} \leq b_i \quad \implies  \sum_{r} y^{c,\dagger}_{ir} \leq (1-4\delta)b_i = \t{b}_i \]
}
We argue that for any $U^c \in \Odc$ such that if $z_{U^c} > 0$,
{\small
\begin{equation}
\sum_{(s,r):s,r \in U^c} y^{c,*}_{sr} \geq 
\sum_{(s,r):s,r \in U^c} y^{c,\dagger}_{sr}
\label{monotone}
\end{equation}
}

If for every edge $(i,j) \in E$ both $p_{ij,i},p_{ij,j}$ are not
present in $U^c$ then Equation~\eqref{monotone} follows immediately
because in the transformation of $\by^{c,*}$ to $\by^{c,\dagger}$ the
only $y^c$ values that increase correspond to
$y^c_{p_{ij,i},p_{ij,j}}$ for some edge $(i,j) \in E$. On the other hand, 
suppose that for some $(i,j) \in E$  both $
p_{ij,i},p_{ij,j}$ are present then using Lemma~\ref{lem:pij}, either $i$ or $j \in U^c$. 
Without loss of generality, suppose $i \in U^c$. But then the increase in 
$y^c_{p_{ij,i},p_{ij,j}}$ cancels out the the decrease in $y^c_{ip_{ij,i}}$.
Therefore Equation~\eqref{monotone} follows.

Note $\sum_{r:(s,r) \in E^c} w^c_{rs} y^{c,\dagger}_{sr} =
(1-4\delta)\beta^{*,c}$ (discussion following Algorithm~\ref{alg:cap}) and
$x_s=0$ for $s \in V^c \setminus V$ (Lemma~\ref{lem:pij},
part(a)). Omitting the implied $U^c \in \Odc$ for notational
simplicity in the sum below, we get: {\small
\begin{align*}
\displaystyle\sum_{r:(s,r) \in E_c} \eta_{sr} y^{c,\dagger}_{sr} & = \sum_{r:(s,r) \in E^c} \left( x_s+x_r+\hspace{-0.5cm}\sum_{U^c \in \L^c;s,r\in U} \hspace{-0.5cm} z_{U^c} \right) y^{c,\dagger}_{sr} = \sum_{s} x_s \left(\sum_{r} y^{c,\dagger}_{sr} \right) + \sum_{U^c: z_{U^c} >0} z_{U^c} \left( \sum_{(s,r):s,r \in U^c} y^{c,\dagger}_{sr} \right)\\
& = \sum_{i: x_i >0} x_i \left(\sum_{r} y^{c,\dagger}_{ir} \right) + \sum_{U^c: z_{U^c} >0} z_{U^c} \left( \sum_{(s,r):s,r \in U^c} y^{c,\dagger}_{sr} \right) \leq \sum_{s: x_s >0} x_s \t{b}_s + \sum_{U: z_{U^c} >0} z_{U^c} \frac{\t{b}_{U^c}}{(1-\delta)} \\
& \leq \frac1{(1-\delta)}\left(\sum_{s: x_s >0} x_s \t{b}_s + \sum_{U^c: z_{U^c} >0} z_{U^c}\t{b}_{U^c} \right) = \frac{\gamma^c}{(1-\delta)}
\end{align*}
}
Therefore there exists a solution for 
{\small
\[ \left\{\sum_{(s,r) \in E^c} w^c_{sr}y^{c,\dagger}_{sr} = (1-4\delta)\beta^{*,c} \qquad
 \sum_{r:(s,r) \in E^c} \eta_{sr} y^c_{sr} \leq  \frac{\gamma^c}{(1-\delta)} \qquad  
  \by^{c,\dagger} \in \ylong{\Q^c}\right \}
\]
}
which,  by the observation made in this proof, implies that for $\beta \leq (1-4\delta)\beta^{*,c}$ there exists a solution for
{\small
\[ \left\{ \displaystyle\sum_{(i,j) \in E} w_{ij}y_{ij} \geq \beta, \qquad \sum_{(i,j)\in E} \shortcost{\boldsymbol{\eta}}_{ij} y_{ij} = \frac{\gamma^c}{(1-\delta)} - \shiftcost{\boldsymbol{\eta}}, \qquad \by \in \Q^c \right \} 
\]
}
We can now apply Theorem~\ref{apxtheorem} with $f_1 = \beta > 0, f_2 = \frac{\gamma^c}{(1-\delta)} - \shiftcost{\boldsymbol{\eta}}$ (by Lemma~\ref{lem:pij}, $f_2 \geq 0$) and $\P_1=\Q_c$ and $\P_2=\Pc$. Note that $
\shortcost{\boldsymbol{\eta}} \geq \mathbf{0}$ by Lemma~\ref{lem:pij}. Finally $\mathbf{0} \in \Qc \subseteq \Pc$. 
and the algorithm desired by Theorem~\ref{apxtheorem} is provided by Theorem~\ref{caprecursive}. 
Therefore we have a solution of \ref{lp:exists}. 
The number of iterations in  
Theorem~\ref{apxtheorem} is $O(\ln (2/\delta))$ each of which invokes Theorem~\ref{caprecursive}.  Theorem~\ref{caprecursive} involves Theorem~\ref{capapx} repeatedly.
The bound
on $\sum_{(i,j)\in \hat{E}} w_{ij}c_{ij}$ follows from the fact that
we average solutions of Theorem~\ref{uncapapx} for which $\sum_{(i,j):y_{ij}>0} w_{ij}c_{ij} \leq 8 \beta^{*,c}_b$ (the bipartite maximum) which can be bounded by $16 \beta^{*,c}$.  
\end{proof}

\noindent
We can now conclude Theorem~\ref{thmtwo}.

\begin{ntheorem}{\ref{thmtwo}}  
Given any non-bipartite graph, for any $\frac{3}{\sqrt{n}}<\delta\leq 1/16$, we
  find a $(1-O(\delta))$-approximate 
fractional solution to \ref{lp:cbm-long-int} using  $O(mR/\delta + \min\{ B, m \}
  \poly \{\delta^{-1},\ln n\})$ time,  additional ``work'' space
$O(\min\{m,B\} \poly \{\delta^{-1},\ln n\})$
making $R=O(\delta^{-4} (\ln^2 (1/\delta)) \ln n)$ 
passes over the list of edges where $B=\sum_i b_i$.  The algorithm
  returns a solution $\{\h{y}_{ij}\}=\yshort{\by^c}$ such that the subgraph
  $\h{E}=\{(i,j)| (i,j) \in E, \h{y}_{ij}>0\}$ satisfies $\sum_{(i,j)
    \in \h{E}} w_{ij}c_{ij} \leq 16R\beta^{*,c}$ where $\beta^{*,c}$ is the weight of
  the optimum integral capacitated $b$--matching. 
\end{ntheorem}

\subsection{Rounding Capacitated $b$-Matchings}
\label{sec:crounding} We prove Theorem~\ref{thm:crounding} based on Algorithm~\ref{alg:round2}.

\begin{algorithm}[H]
{\small
\begin{algorithmic}[1]
\STATE First Phase: {\bf Removing edges with large multiplicities} 
(no change from Algorithm~\ref{alg:round1} except tracking edge capacities). Let $t=\lceil 2/\delta \rceil$ and $\mzero_c=\emptyset$.

\begin{enumerate}[(a)]\parskip=0in
\item If $y_{ij}\geq t$ add $\h{y}^{(0)}_{ij}=\lfloor y_{ij} \rfloor - 1$ copies of $(i,j)$ to $\mzero_c$.

 \item Set $y^{(1)}_{ij} = \left\{\begin{array}{l l}
   0      & \mbox{if }y_{ij}\geq t \\
   y_{ij} & \mbox{otherwise}    
  \end{array}\right.$. Set $b^{(1)}_i = \min \left \{ b_i - \sum_j \h{y}^{(0)}_{ij}, \lceil \sum_j  \yone_{ij} \rceil + 1 \right \}$ and $c^{(1)}_{ij} = \min \{ c_{ij}, \lceil \yone_{ij} \rceil + 1\}$. This describes the graph $G^{(1)}_c=(V,\eone)$. Note $c^{(1)}_{ij} \leq t+1$.
\end{enumerate}
\STATE Second Phase: {\bf Subdividing vertices with large multiplicities.} 
(no change from Algorithm~\ref{alg:round1} except tracking edge capacities). 
We set $\ctwo_{i'j'} = \cone_{ij}$ where the edge $(i,j)$ got assigned to 
$i'$ and $j'$ which are copies of $i$ and $j$ respectively. This defines 
$G^{(2)}_c=(\vtwo,\etwo)$.
Note only vertices are split, -- the edges are not split, even though they can be assigned to 
a copy of an original vertex, i.e., $|\eone|=|\etwo|$.  
Let $\boldsymbol{\mathcal W} =\sum_{(i,j) \in \etwo} \ctwo_{ij}w_{ij}$.
~\vspace{0.1in} 

\STATE Third Phase: {\bf Reducing the problem to a weighted matching on small graph.} (different from Algorithm~\ref{alg:round1}). Given  $G^{(2)}_c$, define $G^{(3)}_c$ as follows:

\begin{enumerate}[(a)]\parskip=0in
 \item For each $i \in V^{(2)}$ with $b^{(2)}_i$, create $i(1), i(2), \cdots, i(b^{(2)}_i)$. 
For each edge $e=(i,j)$, we create $2c^{(2)}_{ij}$ vertices
  $p_{ei,1},p_{ei,2},\cdots,p_{ei,c^{(2)}_{ij}},p_{ej,1},p_{ej,2},\cdots,p_{ej,c^{(2)}_{ij}}$.
 \item Add edges $(p_{ei,\ell},p_{ej,\ell})$ with edge weight $w_{ij}$. Add a complete bipartite graph between $i_1,i_2,\cdots$
  and $p_{ei,1},p_{ei,2},\cdots$ with edge weight $w_{ij}$. 
\item Run any fast approximation for finding a $(1-\frac{\delta}{32R})$-approximate maximum weighted matching in $G^{(3)_c}$. Let this matching be $\mthreepa_c$ of weight $W$.

\item Observe that given any integral matching in $G^{(3)_c}$, we can construct a matching of 
same or greater weight such that every one of the vertices $p_{ei,\ell},p_{ej,\ell}$ (for all $e=(i,j),\ell$) are matched -- if for some $e,\ell$ neither $p_{ei,\ell},p_{ej,\ell}$ are matched then we can match them, if only one of the pair is matched then we delete the matching edge incident to the other one in the pair and add the matching edge between $p_{ei,\ell},p_{ej,\ell}$ which is of the same weight.
Applying this procedure to $\mthreepa_c$ we get $\mthreepb_c$ of weight at least $W$.  
\end{enumerate}

\STATE We now merge all the vertices $i(\ell)$ to $i$, $p_{ei,\ell}$ to $p_{ei}$ and $p_{ej,\ell}$ to $p_{ej}$ for all $e=(i,j),\ell$. Observe $G^{(3)}_c$ reduces to $\ylong{G^{(2)}_c}$ with 
different edge weights, i.e., for an edge $e=(i,j)$ of weight $w_{ij}$ in the original graph we 
have the weights of $(i,p_{ei}), (p_{ei}p_{ej})$ and $(p_{ej},j)$ are all $w_{ij}$ instead of 
$\frac12w_{ij},0,\frac12w_{ij}$ as in the definition of $\ylong{G^{(2)}_c}$. 

However if 
we merge all the corresponding edges of $\mthreepb_c$ then we get a matching $\mthree$ 
such that the vertices $p_{ei}$ and $p_{ej}$ are matched to capacity $c_{ij}$ for every 
edge $e=(i,j)$. Note that $\mthree_c$ has weight at least $W$. $\mthree$ provides
a $b$--matching $\mtwo_c$ in $G^{(2)}_c$ of weight at least $W-\boldsymbol{\mathcal W}$, where we set $y^\dagger_{ij}=y^\dagger_{ip_{ei}}$. $\mtwo_c$
provides a $b$--matching $\mone_c$ in $G^{(1)}_c$ of same weight (merge vertices).

\smallskip
\STATE {\bf Output} $\mzero_c \cup \mone_c$.
\end{algorithmic}
}
\caption{Rounding capacitated $b$--matchings}
\label{alg:round2}
\end{algorithm}

For example, in Step 3(b)

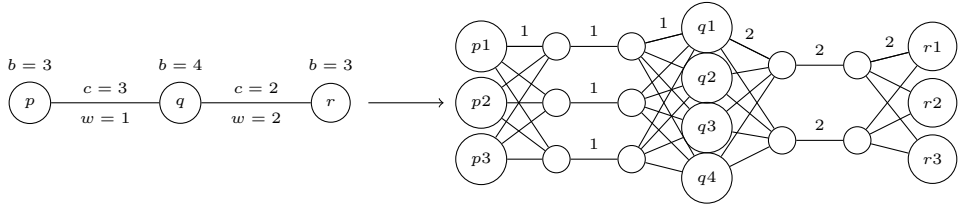
 \begin{tikzpicture}[font=\tiny,scale=1]
 \tikzstyle{vertex}=[draw,shape=circle,radius=0.25cm]
 \node[vertex] at (-2,0) (u) {$p$};
 \node[above,yshift=0.3cm] at (u) {$b=3$};
 \node[vertex] at (0,0) (v) {$q$};
 \node[above,yshift=0.3cm] at (v) {$b=4$};
 \node[vertex] at (2,0) (w) {$r$};
 \node[above,yshift=0.3cm] at (w) {$b=3$};
 \draw (u) to node[above] {$c=3$} node[below] {$w=1$} (v);
 \draw (v) to node[above] {$c=2$} node[below] {$w=2$} (w);

 \draw[->] (2.5,0) to (3.5,0);
 \tikzstyle{vertex}=[draw,shape=circle,radius=0.25cm]
 \node[vertex] at (4,0.75)  (u1) {$p1$};
 \node[vertex] at (4,0)  (u2) {$p2$};
 \node[vertex] at (4,-0.75) (u3) {$p3$};
 \node[vertex] at (5,0.75) (eu1) {};
 \node[vertex] at (5,0) (eu2) {};
 \node[vertex] at (5,-0.75) (eu3) {};
 \node[vertex] at (6,0.75) (ev1) {};
 \node[vertex] at (6,0) (ev2) {};
 \node[vertex] at (6,-0.75) (ev3) {};
 \node[vertex] at (7,1)  (v1) {$q1$};
 \node[vertex] at (7,0.33) (v2) {$q2$};
 \node[vertex] at (7,-0.33) (v3) {$q3$};
 \node[vertex] at (7,-1) (v4) {$q4$};
 \node[vertex] at (8,0.5)  (fv1) {};
 \node[vertex] at (8,-0.5) (fv2) {};
 \node[vertex] at (9,0.5)  (fw1) {};
 \node[vertex] at (9,-0.5) (fw2) {};
 \node[vertex] at (10,0.75)  (w1) {$r1$};
 \node[vertex] at (10,0)  (w2) {$r2$};
 \node[vertex] at (10,-0.75) (w3) {$r3$};
 \foreach \x in {1,2,3}
  \foreach \y in {1,2,3}
   \draw (u\x) -- (eu\y); 
   \draw (u1) to node[above]{$1$} (eu1);
 \foreach \x in {1,2,3}
  \draw (eu\x) to node[above]{$1$} (ev\x) ;
 \foreach \x in {1,2,3}
  \foreach \y in {1,2,3,4}
   \draw (ev\x) -- (v\y);
   \draw (ev1) to node[above]{$1$} (v1);
 \foreach \x in {1,2}
  \foreach \y in {1,2,3,4}
   \draw (fv\x) -- (v\y);
 \draw (fv1) to node[above]{2} (v1);
 \foreach \x in {1,2}
  \draw (fv\x)  to node[above]{2} (fw\x);
 \foreach \x in {1,2}
  \foreach \y in {1,2,3}
   \draw (fw\x) -- (w\y);
   \draw (fw1) to node[above]{2} (w1);
 \end{tikzpicture}

which in turn reduces to 

\hspace*{3cm}
\begin{tikzpicture}[font=\tiny]
 \tikzstyle{vertex}=[draw,shape=circle,radius=0.25cm]
 \node[vertex] at (4,0)  (u1) {$p$};
 \node[above,yshift=0.3cm] at (u1) {$b=3$};
 \node[vertex] at (5,0) (eu1) {};
 \node[above,yshift=0.3cm] at (eu1) {$b=3$};
 \node[vertex] at (6,0) (ev1) {};
\node[above,yshift=0.3cm] at (ev1) {$b=3$};
 \node[vertex] at (7,0) (v1) {$q$};
\node[above,yshift=0.3cm] at (v1) {$b=4$};
 \node[vertex] at (8,0)  (fv1) {};
\node[above,yshift=0.3cm] at (fv1) {$b=2$};
 \node[vertex] at (9,0)  (fw1) {};
\node[above,yshift=0.3cm] at (fw1) {$b=2$};
 \node[vertex] at (10,0)  (w1) {$r$};
\node[above,yshift=0.3cm] at (w1) {$b=3$};
   \draw (u1) to node[above]{$1$} (eu1); 
  \draw (eu1) to node[above]{$1$} (ev1) ;
   \draw (ev1) to node[above]{$1$} (v1);
  \draw (fv1) to node[above]{2} (v1);
  \draw (fv1)  to node[above]{2} (fw1);
   \draw (fw1) to node[above]{2} (w1);
 \end{tikzpicture}

\begin{lemma}
\label{lem:crounding-odd-sets}
$\yone_{ij}$ is a feasible fractional capacitated 
$b$--matching in $G^{(1)}_c$.
\end{lemma}

\begin{proof} 
Consider $\ylong{{\mathbf \yone}}$ and $\ylong{G^{(1)}_c}$ with the new capacities $b^{(1)}_i,c^{(1)}_{ij}$ for the vertices and edges in $G^{(1)}_c$.
The only vertices whose capacities were affected in $\ylong{G^{(1)}_c}$ are the following vertices:
(i) the corresponding vertex in $G$ has an edge incident to it in $\mone_c$ and (ii)
the corresponding edge $(i,j) \in G$ had $c_{ij} > \lceil \yone_{ij} \rceil + 1$. 
In both cases the difference between the sum of the new edge
multiplicities and the new capacities (the slack) is at least $1$ and
the first part of Lemma~\ref{lem:roundp1} tells us that these vertices in $\ylong{G^{(1)}_c}$
cannot be part of a violated odd-set in $\ylong{G^{(1)}_c}$. Therefore ${\mathbf \yone}$ is a feasible fractional (uncapacitated) $b$--matching. 
The lemma follows from Theorem~\ref{capacity:thm}.
\end{proof}

\noindent 
Therefore the remaining task is to find a $(1-\delta)$
approximate rounding of the fractional solution $\yone_{ij}$ on
$G^{(1)}_c=(V,\eone)$ with vertex and edge capacities $\{\bone_{ij}\}$ and $\{\cone_{ij}\}$ 
respectively.

\begin{lemma}
\label{cwlemma}
Let $\boldsymbol{\mathcal W} =\sum_{(i,j) \in \etwo} \ctwo_{ij}w_{ij}$.
Then $\boldsymbol{\mathcal W} \leq 16R\beta^{*,c}$.
\end{lemma}

\begin{proof}
Observe that $|\etwo|=|\eone|$ and $\eone \subseteq \hat{E}$ as defined in the statement of Theorem~\ref{thm:crounding}. Moreover $\ctwo_{i'j'}  = \cone_{ij} \leq c_{ij}$. Therefore:
{\small $$\boldsymbol{\mathcal W}=\sum_{(i',j') \in \etwo}  \ctwo_{i'j'} w_{i'j'}= \sum_{(i,j) \in \eone}
\cone_{ij} w_{ij} \leq \sum_{(i,j) \in \eone} c_{ij}w_{ij} \leq \sum_{(i,j) \in \hat{E}} c_{ij}w_{ij} \leq 16R\beta^{*,c}$$}
\end{proof}

\begin{lemma}
\label{cwlemma2}
Algorithm~\ref{alg:round2} outputs a capacitated $b$--matching of weight at least $(1-\delta) \sum_{(i,j)\in E} w_{ij} y_{ij} - \delta \beta^*$.
\end{lemma}

\begin{proof} 
Let the weight of the maximum
matching of this graph $G^{(3)}_c$ be $w(\M^*)$. Then 
{\small
$$ 2 \boldsymbol{\mathcal W} \geq w(\M^*) \geq
\sum_{(i,j) \in \etwo} w_{ij} \ytwo_{ij} + \boldsymbol{\mathcal W}$$}
since each edge $(i,j) \in 
G^{(2)}_c$ can contribute at most $2c^{(2)}_{ij}w_{ij}$ to $w(\M^*)$.

Suppose that we find a $\left(1 - \frac{\delta}{32R}\right)$-approximate maximum
matching in $G^{(3)}_c$, using the algorithm in
\cite{DuanP10,DuanPS11} which takes time $|E(G^{(3)}_c)|$ times
$O(\frac{R}{\delta}
\log (R/\delta))$ which is  $O(m'R\delta^{-3} \log (R/\delta))$. This gives us a
matching of weight at least $W$ where $W \geq w(\M^*) - \frac{\delta}{32R} w(\M^*)$ which corresponds to a 
capacitated $b$--matching in $G^{(2)}_c$ with weight at least $w(\M^*) - \frac{\delta}{32R} w(\M^*) - \boldsymbol{\mathcal W}$. Now
{\small
\begin{align*}
& w(\M^*)  - \frac{\delta}{32R}w(\M^*) - \boldsymbol{\mathcal W} \geq \sum_{(i,j) \in \etwo} w_{ij} \ytwo_{ij} + \boldsymbol{\mathcal W} - \frac{\delta}{32R}w(\M*) - \boldsymbol{\mathcal W} \\
& = \sum_{(i,j) \in \etwo} w_{ij} \ytwo_{ij} - \frac{\delta w(\M^*)}{32R} 
\geq \sum_{(i,j) \in \etwo} w_{ij} \ytwo_{ij} - \delta \beta^{*,c} 
\end{align*}
}

Since the second phase is exactly the same as in the uncapacitated case in
Section~\ref{sec:rounding}, we have 
{\small
$$\sum_{(i,j) \in \etwo}
w_{ij} \ytwo_{ij} \geq (1-\delta) \sum_{(i,j) \in \eone} w_{ij} \yone_{ij}$$
}  
Thus we get a matching $\mone_c$ in $G^{(1)}_c$ of weight $w (\mone_c) \geq (1-\delta) \sum_{(i,j) \in \eone} w_{ij} \yone_{ij} - \delta \beta^{*,c}$. 
Observe that $w (\mzero_c) \geq (1-\delta) \sum_{(i,j)\in E} w_{ij} \left
  (y_{ij} - \yone_{ij} \right)$ where $w (\mzero_c) = \sum_{(i,j) \in E}
y_{ij}^{(0)} w_{ij} $. Then $w (\mzero_c)+w (\mone_c)$ is at least 
$(1-\delta) \sum_{(i,j)\in E} w_{ij} y_{ij} - \delta \beta^*$ as desired.
This proves Lemma~\ref{cwlemma2}.
\end{proof}

\noindent
Therefore we can conclude
Theorem \ref{thm:crounding}.

\smallskip
\begin{ntheorem}{\ref{thm:crounding}}
  Given a fractional capacitated $b$-matching $\by^c$ which is feasible 
  for \ref{lp:cbm-long-int}. Let $\by=\yshort{\by^c}$ and $\hat{E}=\{(i,j)|y_{ij} > 0 \}$.
  Further suppose we are promised that $\sum_{(i,j) \in \hat{E}}
  w_{ij}c_{ij} \leq 16R\beta^{*,c}$.
 We find an integral $b$-matching of weight at least
 $(1-\delta)\sum_{(i,j)} w_{ij}y_{ij} - \delta\beta^{*,c} $ in 
$O(m'R\delta^{-3}\ln (R/\delta))$ time and
 $O(m'/\delta^2)$ space where
 $m'= |\hat{E}|$ is the number of nontrivial edges (as defined by the linear program) 
in the fractional solution. As a consequence
we have a $(1-O(\delta))$-approximate integral solution.
\end{ntheorem}

\end{document}